\documentclass[twocolumn]{article}

\usepackage{amsfonts, amsmath, amsfonts, amssymb, braket, bbm, amsthm, cite}%, tamssymb, mathrsfs, xcolor}

\usepackage{tikz,graphicx}
\tikzset{every picture/.style={scale=.38}}

\def\u{\mathtt u}
\def\v{\mathtt v}
\def\s{\mathtt s}
\def\a{\mathtt a}
\def\b{\mathtt b}
\def\c{\mathtt c}
\def\d{\mathtt d}
\def\e{\mathtt e}

\def\A{\mathcal A}
\def\H{\mathcal H}
\def\Zn{\mathbb Z_n}
\def\Cq{\mathbb C^q}
\newcommand{\prd}[1]{\mbox{$\prod_{#1}$}}

\def\unity{\mathbbm 1}
\def\mitg{\frac 1 2}

\theoremstyle{definition}
\newtheorem{Lemma}{Lemma}

\title{\bf Discrete holography in dual-unitary circuits}
\author{Llu\'{\i}s Masanes\footnote{l.masanes@ucl.ac.uk}
\\
\em London Centre for Nanotechnology, University College London, UK
\\
\em Department of Computer Science, University College London, UK
}

\newcommand{\xg}[2]{
  \draw[thin] (#1-.5,#2-.5) --  (#1+.5,#2+.5);
  \draw[thin] (#1-.5,#2+.5) --  (#1+.5,#2-.5);
}

\newcommand{\porta}[2]{
  \draw[fill=black!10, thin] (#1-.4,#2-.4) rectangle (#1+.4,#2+.4); 
  \draw[thick,color=black!50] (#1-.2,#2-.3) -- (#1-.2,#2); \draw[draw=none,fill=black!50] (#1-.2,#2+.3) -- (#1-.3,#2-.15) -- (#1-.1,#2-.15);
  \draw[thin] (#1-.4,#2-.4) --  (#1-.5,#2-.5);
  \draw[thin] (#1+.4,#2-.4) --  (#1+.5,#2-.5);
  \draw[thin] (#1+.4,#2+.4) --  (#1+.5,#2+.5);
  \draw[thin] (#1-.4,#2+.4) --  (#1-.5,#2+.5);
}
\newcommand{\portac}[2]{
  \draw[fill=black!40, thin] (#1-.4,#2-.4) rectangle (#1+.4,#2+.4); 
  \draw[thick,color=white] (#1-.2,#2-.3) -- (#1-.2,#2); \draw[draw=none,fill=white] (#1-.2,#2+.3) -- (#1-.3,#2-.15) -- (#1-.1,#2-.15);
  \draw[thin] (#1-.4,#2-.4) --  (#1-.5,#2-.5);
  \draw[thin] (#1+.4,#2-.4) --  (#1+.5,#2-.5);
  \draw[thin] (#1+.4,#2+.4) --  (#1+.5,#2+.5);
  \draw[thin] (#1-.4,#2+.4) --  (#1-.5,#2+.5);
}
\newcommand{\portabar}[2]{
  \draw[fill=black!10, thin] (#1-.4,#2-.4) rectangle (#1+.4,#2+.4); 
  \draw[thick,color=black!50] (#1-.2,#2-.3) -- (#1-.2,#2); \draw[draw=none,fill=black!50] (#1-.2,#2+.3) -- (#1-.3,#2-.15) -- (#1-.1,#2-.15);
  \draw[thick,color=black!50] (#1+.2,#2+.3) -- (#1+.2,#2); \draw[draw=none,fill=black!50] (#1+.2,#2-.3) -- (#1+.3,#2+.15) -- (#1+.1,#2+.15);
  \draw[thin] (#1-.4,#2-.4) --  (#1-.5,#2-.5);
  \draw[thin] (#1+.4,#2-.4) --  (#1+.5,#2-.5);
  \draw[thin] (#1+.4,#2+.4) --  (#1+.5,#2+.5);
  \draw[thin] (#1-.4,#2+.4) --  (#1-.5,#2+.5);
}
\newcommand{\portabart}[2]{
  \draw[fill=black!10, thin] (#1-.4,#2-.4) rectangle (#1+.4,#2+.4); 
  \draw[thick,color=black!50] (#1-.2,#2+.3) -- (#1-.2,#2); \draw[draw=none,fill=black!50] (#1-.2,#2-.3) -- (#1-.3,#2+.15) -- (#1-.1,#2+.15);
  \draw[thick,color=black!50] (#1+.2,#2-.3) -- (#1+.2,#2); \draw[draw=none,fill=black!50] (#1+.2,#2+.3) -- (#1+.3,#2-.15) -- (#1+.1,#2-.15);
  \draw[thin] (#1-.4,#2-.4) --  (#1-.5,#2-.5);
  \draw[thin] (#1+.4,#2-.4) --  (#1+.5,#2-.5);
  \draw[thin] (#1+.4,#2+.4) --  (#1+.5,#2+.5);
  \draw[thin] (#1-.4,#2+.4) --  (#1-.5,#2+.5);
}
\newcommand{\portabarc}[2]{
  \draw[fill=black!40, thin] (#1-.4,#2-.4) rectangle (#1+.4,#2+.4); 
  \draw[thick,color=white] (#1-.2,#2-.3) -- (#1-.2,#2); 
  \draw[draw=none,fill=white] (#1-.2,#2+.3) -- (#1-.3,#2-.15) -- (#1-.1,#2-.15);
  \draw[thick,color=white] (#1+.2,#2+.3) -- (#1+.2,#2); 
  \draw[draw=none,fill=white] (#1+.2,#2-.3) -- (#1+.3,#2+.15) -- (#1+.1,#2+.15);
  \draw[thin] (#1-.4,#2-.4) --  (#1-.5,#2-.5);
  \draw[thin] (#1+.4,#2-.4) --  (#1+.5,#2-.5);
  \draw[thin] (#1+.4,#2+.4) --  (#1+.5,#2+.5);
  \draw[thin] (#1-.4,#2+.4) --  (#1-.5,#2+.5);
}

\newcommand{\portacbar}[2]{
  \draw[fill=black!40, thin] (#1-.4,#2-.4) rectangle (#1+.4,#2+.4); 
  \draw[thick,color=white] (#1+.3,#2+.2) -- (#1,#2+.2);
  \draw[draw=none,fill=white] (#1-.3,#2+.2) -- (#1+.15,#2+.3) -- (#1+.15,#2+.1);
  \draw[thick,color=white] (#1-.3,#2-.2) -- (#1,#2-.2);
  \draw[draw=none,fill=white] (#1+.3,#2-.2) -- (#1-.15,#2-.3) -- (#1-.15,#2-.1);
  \draw[thin] (#1-.4,#2-.4) --  (#1-.5,#2-.5);
  \draw[thin] (#1+.4,#2-.4) --  (#1+.5,#2-.5);
  \draw[thin] (#1+.4,#2+.4) --  (#1+.5,#2+.5);
  \draw[thin] (#1-.4,#2+.4) --  (#1-.5,#2+.5);
}
\newcommand{\portaabar}[2]{
  \draw[fill=black!40, thin] (#1-.4,#2-.4) rectangle (#1+.4,#2+.4); 
  \draw[thick,color=white] (#1+.2,#2-.3) -- (#1+.2,#2); \draw[draw=none,fill=white] (#1+.2,#2+.3) -- (#1+.3,#2-.15) -- (#1+.1,#2-.15);
  \draw[thick,color=white] (#1-.2,#2+.3) -- (#1-.2,#2); \draw[draw=none,fill=white] (#1-.2,#2-.3) -- (#1-.3,#2+.15) -- (#1-.1,#2+.15);
  \draw[thin] (#1-.4,#2-.4) --  (#1-.5,#2-.5);
  \draw[thin] (#1+.4,#2-.4) --  (#1+.5,#2-.5);
  \draw[thin] (#1+.4,#2+.4) --  (#1+.5,#2+.5);
  \draw[thin] (#1-.4,#2+.4) --  (#1-.5,#2+.5);
}
\newcommand{\portacabar}[2]{
  \draw[fill=black!10, thin] (#1-.4,#2-.4) rectangle (#1+.4,#2+.4); 
  \draw[thick,color=black!50] (#1+.3,#2-.2) -- (#1,#2-.2);
  \draw[draw=none,fill=black!50] (#1-.3,#2-.2) -- (#1+.15,#2-.3) -- (#1+.15,#2-.1);
  \draw[thick,color=black!50] (#1-.3,#2+.2) -- (#1,#2+.2);
  \draw[draw=none,fill=black!50] (#1+.3,#2+.2) -- (#1-.15,#2+.3) -- (#1-.15,#2+.1);
  \draw[thin] (#1-.4,#2-.4) --  (#1-.5,#2-.5);
  \draw[thin] (#1+.4,#2-.4) --  (#1+.5,#2-.5);
  \draw[thin] (#1+.4,#2+.4) --  (#1+.5,#2+.5);
  \draw[thin] (#1-.4,#2+.4) --  (#1-.5,#2+.5);
}
\newcommand{\portatq}[2]{
  \draw[fill=black!10, thin] (#1-.4,#2-.4) rectangle (#1+.4,#2+.4); 
  \draw[thick,color=black!50] (#1+.3,#2-.2) -- (#1,#2-.2);
  \draw[draw=none,fill=black!50] (#1-.3,#2-.2) -- (#1+.15,#2-.3) -- (#1+.15,#2-.1);
  \draw[thin] (#1-.4,#2-.4) --  (#1-.5,#2-.5);
  \draw[thin] (#1+.4,#2-.4) --  (#1+.5,#2-.5);
  \draw[thin] (#1+.4,#2+.4) --  (#1+.5,#2+.5);
  \draw[thin] (#1-.4,#2+.4) --  (#1-.5,#2+.5);
}
\newcommand{\portacabars}[2]{
  \draw[fill=black!10, thin] (#1-.4,#2-.4) rectangle (#1+.4,#2+.4); 
  \draw[thick,color=black!50] (#1-.3,#2-.2) -- (#1,#2-.2);
  \draw[draw=none,fill=black!50] (#1+.3,#2-.2) -- (#1-.15,#2-.3) -- (#1-.15,#2-.1);
  \draw[thick,color=black!50] (#1+.3,#2+.2) -- (#1,#2+.2);
  \draw[draw=none,fill=black!50] (#1-.3,#2+.2) -- (#1+.15,#2+.3) -- (#1+.15,#2+.1);
  \draw[thin] (#1-.4,#2-.4) --  (#1-.5,#2-.5);
  \draw[thin] (#1+.4,#2-.4) --  (#1+.5,#2-.5);
  \draw[thin] (#1+.4,#2+.4) --  (#1+.5,#2+.5);
  \draw[thin] (#1-.4,#2+.4) --  (#1-.5,#2+.5);
}
\newcommand{\portarr}[2]{
  \draw[fill=black!10, thin] (#1-.4,#2-.4) rectangle (#1+.4,#2+.4); 
  \draw[thick,color=black!50] (#1+.2,#2+.3) -- (#1+.2,#2);
  \draw[draw=none,fill=black!50] (#1+.2,#2-.3) -- (#1+.3,#2+.15) -- (#1+.1,#2+.15);
  \draw[thin] (#1-.4,#2-.4) --  (#1-.5,#2-.5);
  \draw[thin] (#1+.4,#2-.4) --  (#1+.5,#2-.5);
  \draw[thin] (#1+.4,#2+.4) --  (#1+.5,#2+.5);
  \draw[thin] (#1-.4,#2+.4) --  (#1-.5,#2+.5);
}
\newcommand{\portas}[2]{
  \draw[fill=black!10, thin] (#1-.4,#2-.4) rectangle (#1+.4,#2+.4); 
  \draw[thick,color=black!50] (#1+.2,#2-.3) -- (#1+.2,#2);
  \draw[draw=none,fill=black!50] (#1+.2,#2+.3) -- (#1+.3,#2-.15) -- (#1+.1,#2-.15);
  \draw[thin] (#1-.4,#2-.4) --  (#1-.5,#2-.5);
  \draw[thin] (#1+.4,#2-.4) --  (#1+.5,#2-.5);
  \draw[thin] (#1+.4,#2+.4) --  (#1+.5,#2+.5);
  \draw[thin] (#1-.4,#2+.4) --  (#1-.5,#2+.5);
}
\newcommand{\portasr}[2]{
  \draw[fill=black!10, thin] (#1-.4,#2-.4) rectangle (#1+.4,#2+.4); 
  \draw[thick,color=black!50] (#1-.3,#2+.2) -- (#1,#2+.2);
  \draw[draw=none,fill=black!50] (#1+.3,#2+.2) -- (#1-.15,#2+.3) -- (#1-.15,#2+.1);
  \draw[thin] (#1-.4,#2-.4) --  (#1-.5,#2-.5);
  \draw[thin] (#1+.4,#2-.4) --  (#1+.5,#2-.5);
  \draw[thin] (#1+.4,#2+.4) --  (#1+.5,#2+.5);
  \draw[thin] (#1-.4,#2+.4) --  (#1-.5,#2+.5);
}
\newcommand{\portat}[2]{
  \draw[fill=black!10, thin] (#1-.4,#2-.4) rectangle (#1+.4,#2+.4); 
  \draw[thick,color=black!50] (#1-.2,#2+.3) -- (#1-.2,#2);
  \draw[draw=none,fill=black!50] (#1-.2,#2-.3) -- (#1-.3,#2+.15) -- (#1-.1,#2+.15);
  \draw[thin] (#1-.4,#2-.4) --  (#1-.5,#2-.5);
  \draw[thin] (#1+.4,#2-.4) --  (#1+.5,#2-.5);
  \draw[thin] (#1+.4,#2+.4) --  (#1+.5,#2+.5);
  \draw[thin] (#1-.4,#2+.4) --  (#1-.5,#2+.5);
}
\newcommand{\portad}[2]{
  \draw[fill=black!40, thin] (#1-.4,#2-.4) rectangle (#1+.4,#2+.4); 
  \draw[thick,color=white] (#1-.2,#2+.3) -- (#1-.2,#2);
  \draw[draw=none,fill=white] (#1-.2,#2-.3) -- (#1-.3,#2+.15) -- (#1-.1,#2+.15);
  \draw[thin] (#1-.4,#2-.4) --  (#1-.5,#2-.5);
  \draw[thin] (#1+.4,#2-.4) --  (#1+.5,#2-.5);
  \draw[thin] (#1+.4,#2+.4) --  (#1+.5,#2+.5);
  \draw[thin] (#1-.4,#2+.4) --  (#1-.5,#2+.5);
}
\newcommand{\portacs}[2]{
  \draw[fill=black!40, thin] (#1-.4,#2-.4) rectangle (#1+.4,#2+.4); 
  \draw[thick,color=white] (#1+.2,#2-.3) -- (#1+.2,#2);
  \draw[draw=none,fill=white] (#1+.2,#2+.3) -- (#1+.3,#2-.15) -- (#1+.1,#2-.15);
  \draw[thin] (#1-.4,#2-.4) --  (#1-.5,#2-.5);
  \draw[thin] (#1+.4,#2-.4) --  (#1+.5,#2-.5);
  \draw[thin] (#1+.4,#2+.4) --  (#1+.5,#2+.5);
  \draw[thin] (#1-.4,#2+.4) --  (#1-.5,#2+.5);
}
\newcommand{\portadrr}[2]{
  \draw[fill=black!40, thin] (#1-.4,#2-.4) rectangle (#1+.4,#2+.4); 
  \draw[thick,color=white] (#1+.3,#2+.2) -- (#1,#2+.2);
  \draw[draw=none,fill=white] (#1-.3,#2+.2) -- (#1+.15,#2+.3) -- (#1+.15,#2+.1);
  \draw[thin] (#1-.4,#2-.4) --  (#1-.5,#2-.5);
  \draw[thin] (#1+.4,#2-.4) --  (#1+.5,#2-.5);
  \draw[thin] (#1+.4,#2+.4) --  (#1+.5,#2+.5);
  \draw[thin] (#1-.4,#2+.4) --  (#1-.5,#2+.5);
}
\newcommand{\portadr}[2]{
  \draw[fill=black!40, thin] (#1-.4,#2-.4) rectangle (#1+.4,#2+.4); 
  \draw[thick,color=white] (#1-.3,#2-.2) -- (#1,#2-.2);
  \draw[draw=none,fill=white] (#1+.3,#2-.2) -- (#1-.15,#2-.3) -- (#1-.15,#2-.1);
  \draw[thin] (#1-.4,#2-.4) --  (#1-.5,#2-.5);
  \draw[thin] (#1+.4,#2-.4) --  (#1+.5,#2-.5);
  \draw[thin] (#1+.4,#2+.4) --  (#1+.5,#2+.5);
  \draw[thin] (#1-.4,#2+.4) --  (#1-.5,#2+.5);
}
\newcommand{\portau}[2]{
  \draw[fill=black!10, thin] (#1-.4,#2-.4) rectangle (#1+.4,#2+.4); 
  \draw (#1,#2) node {$\u$};
  \draw[thin] (#1-.4,#2-.4) --  (#1-.5,#2-.5);
  \draw[thin] (#1+.4,#2-.4) --  (#1+.5,#2-.5);
  \draw[thin] (#1+.4,#2+.4) --  (#1+.5,#2+.5);
  \draw[thin] (#1-.4,#2+.4) --  (#1-.5,#2+.5);
}
\newcommand{\portav}[2]{
  \draw[fill=black!10, thin] (#1-.4,#2-.4) rectangle (#1+.4,#2+.4); 
  \draw (#1,#2) node {$\v$};
  \draw[thin] (#1-.4,#2-.4) --  (#1-.5,#2-.5);
  \draw[thin] (#1+.4,#2-.4) --  (#1+.5,#2-.5);
  \draw[thin] (#1+.4,#2+.4) --  (#1+.5,#2+.5);
  \draw[thin] (#1-.4,#2+.4) --  (#1-.5,#2+.5);
}
\newcommand{\portaud}[2]{
  \draw[fill=black!10, thin] (#1-.4,#2-.4) rectangle (#1+.4,#2+.4); 
  \draw (#1,#2) node {$\,\u^{\hspace{-.6mm}\dag}$};
  \draw[thin] (#1-.4,#2-.4) --  (#1-.5,#2-.5);
  \draw[thin] (#1+.4,#2-.4) --  (#1+.5,#2-.5);
  \draw[thin] (#1+.4,#2+.4) --  (#1+.5,#2+.5);
  \draw[thin] (#1-.4,#2+.4) --  (#1-.5,#2+.5);
}
\newcommand{\portavd}[2]{
  \draw[fill=black!10, thin] (#1-.4,#2-.4) rectangle (#1+.4,#2+.4); 
  \draw (#1,#2) node {$\,\v^{\hspace{-.6mm}\dag}$};
  \draw[thin] (#1-.4,#2-.4) --  (#1-.5,#2-.5);
  \draw[thin] (#1+.4,#2-.4) --  (#1+.5,#2-.5);
  \draw[thin] (#1+.4,#2+.4) --  (#1+.5,#2+.5);
  \draw[thin] (#1-.4,#2+.4) --  (#1-.5,#2+.5);
}
\newcommand{\barra}[2]{
  \draw[thin] (#1-.5,#2-.5) -- (#1+.5,#2+.5);
}
\newcommand{\antibarra}[2]{
  \draw[thin] (#1-.5,#2+.5) -- (#1+.5,#2-.5);
}
\newcommand{\cables}[2]{
  \draw[thin] (#1-.5,#2-.5) -- (#1-.5,#2-.7);
  \draw[thin] (#1+.5,#2-.5) -- (#1+.5,#2-.7);
}
\newcommand{\cablesh}[2]{
  \draw[thin] (#1+.5,#2+.5) -- (#1+.7,#2+.5);
  \draw[thin] (#1+.5,#2-.5) -- (#1+.7,#2-.5);
}

\begin{document}
\maketitle
\begin{abstract}
%We introduce a discrete-spacetime analog of conformal field theory.
We introduce a family of dual-unitary circuits in 1+1 dimensions which constitute a discrete analog of conformal field theories. These circuits are quantum cellular automata which are invariant under the joint action of Lorentz and scale transformations. 
Dual unitaries are four-legged tensors which satisfy the unitarity condition across the time as well as the space direction, a property that makes the model mathematically tractable.
%These have attracted a lot of attention in the last three years. One reason being that they allow for the construction of many-body models which display quantum chaos and are analytically tractable. 
Using dual unitaries too, we construct tensor-network states for our 1+1 model, which are interpreted as spatial slices of curved 2+1 discrete geometries, where the metric distance is defined by the entanglement structure of the state, following Ryu-Takayanagi's prescription. The dynamics of the circuit induces a natural dynamics on these geometries, which we study for flat and anti-de Sitter spaces, and in the presence or absence of matter. We observe that the dynamics of a particle in such spaces strongly depends on the presence of other particles, suggesting gravitational back-reaction.
\end{abstract}

%%%%%%%%%%%%%%%%%%%%%%%%%%%%%%%%%%%%%%%%%%%%%%%%%%%%%%%%%%%

%\bigskip {\bf
%It seems there is  no backreaction, but this is because the tensor-netrowk representation of states with matter is not unique, and we have chosen the inserted operators reepresenting matter in such a way that the evolution of the background geometry is not altered (1-particle sstates). But in the evolution of two particles we observe interaction, and hence backreaction.
%}

\section{Introduction}

%Inspired by the thermodynamics of black holes....which conjectures that the maximum entropy in any region scales with the radius squared, and not cubed as might be expected.
%The holographic principle is a supposed property of quantum gravity that states that the description of a volume of space can be thought of as encoded on a lower-dimensional boundary to the region. First proposed by Gerard 't Hooft, it was given a precise string-theory interpretation by Leonard Susskind,[2] who

In 1997 Juan Maldacena proposed a duality between (i) certain theories of quantum gravity with negative cosmological constant in $d+1$ spacetime dimensions, and (ii) conformally-symmetric quantum field theories (CFT) in $d$ spacetime dimensions \cite{Maldacena_1999}. 
This duality is also known as the AdS/CFT correspondence, because when the cosmological constant is negative, the ground state of classical gravity is anti-de Sitter space (AdS).
Important aspects of the correspondence were soon elaborated by other authors \cite{Gubser_1998, Witten_98}, and today Maldacena's paper gathers 22889 citations.

A central feature of this duality is that the spacetime curvature on the gravity side corresponds to the entanglement structure and its dynamics on the CFT side \cite{Ryu_2006, Ryu_2007, Hubeny, Lashkari_2014, Czech_2012, Wall_2014, Headrick_2014, Esp_ndola_2018, Aaronson_2022, Bao_2015}. 
In \cite{Bao_2015} this is expressed as holography being a hydrodynamic description of the boundary entanglement with entropies as its macroscopic phase space.
%Roughly speaking, the entanglement entropy of any region $\mathcal R \subseteq \mathbb R^{d-1}$ in the CFT, is proportional to the area of the extremal surface $\mathcal S \subseteq \mathbb R^{d}$ satisfying $\partial\mathcal S = \partial\mathcal R$.
This opens the possibility of describing the entanglement dynamics of other many-body systems (different than CFT) in terms of geometry with an extra dimension.
One direction where this exploration has been fruitful is that of discrete holography, where one considers quantum systems with discrete degrees of freedom (e.g.~spin chains) that are dual to discrete geometries \cite{Swingle_12, Pastawski_2015, Hayden_2016, Basteiro_2022, Erdmenger_2022, Niermann_2022, Osborne_2020, Jahn_2019, Asaduzzaman_2020, Jahn_2021, Jahn_2020, Brower_2021, Jahn_2022, Asaduzzaman_2022}.
Some advantages of the discrete approach are: mathematical tractability and  rigour, more straightforward simulation on classical and quantum computers, and more accessible experiments.

In this work we present a family of quantum circuits which are invariant under the joint action of Lorentz and scale transformations. (Recall that invariance under Lorentz transformations alone is impossible on the lattice.)
Our circuits are constructed with any given two-site dual-unitary gate $\tikz[baseline]{\porta{0}{.3}}$ and its complex conjugated $\tikz[baseline]{\portac{0}{.3}}$ following the pattern
\begin{align}\label{intro T}
  T=\cdots 
\begin{tikzpicture}[baseline=13]
  \portadrr{-1}{3}\portadr{-1}{1}\cables{-1}{4.2}
  \cables{0}{0}
  \porta{0}{0} \portadrr{1}{1} \portarr{0}{2} \portadr{1}{3}
  \cables{1}{4.2}
  \cables{2}{0}
  \portarr{2}{0} \portadr{3}{1} \porta{2}{2} \portadrr{3}{3}
  \cables{3}{4.2}
  \cables{4}{0}
  \porta{4}{0} \portadrr{5}{1} \portarr{4}{2} \portadr{5}{3}
  \cables{5}{4.2}
  \cables{6}{0}
  \portarr{6}{0} \portadr{7}{1} \porta{6}{2} \portadrr{7}{3}
  \cables{7}{4.2}
  \cables{8}{0} 
  \porta{8}{0} \portadrr{9}{1} \portarr{8}{2} \portadr{9}{3}
  \cables{9}{4.2}
  \portarr{10}{0} \porta{10}{2}\cables{10}{0}
\end{tikzpicture}
  \cdots
\end{align}
so that causality is strictly respected. 
This type of dynamics is called quantum cellular automata (QCA) \cite{Arrigui, G_tschow_2010, Freedman_2022, Farrelly_2020}, and it is the discrete-spacetime version of quantum field theories.
(Recall that in ``lattice field theory" time is continuous, which implies the loss of either causality or unitarity, both respected by QCAs.)
In this work we introduce conformal QCAs in 1+1 dimensions, and leave the generalisation to higher dimensions for the future. 
To our knowledge, these are the first QCAs whose evolution operator $T$ has a form of Lorentz and scale invariance.
An essential ingredient of these QCAs are dual unitaries: four-legged tensors which satisfy the unitarity condition across the space as well as the time direction (\ref{eq:unitarity}-\ref{eq:dual unitarity}), implying the unitarity of circuit \eqref{intro T}.
Dual-unitary circuits have recently attracted a lot of attention because they provide mathematically-tractable models of quantum chaos \cite{Bertini_2019, Bertini_2021, Piroli_2020, kos_20}.

After a detailed presentation of conformal QCAs (Section~\ref{sec:model}) we explore the holographic properties of these models. 
In Section~\ref{sec:holography empty} we analyse how tensor-network states for 1+1-dimensional QCAs define 2+1 discrete geometries with metric distance defined by the entanglement structure of the state, following Ryu–Takayanagi's prescription \cite{Ryu_2006, Ryu_2007, Hubeny, Lashkari_2014}. 
General Relativity in 2+1 dimensions does not have propagating gravitational degrees of freedom, but it can have non-trivial dynamics for the manifold boundary, which is what we analyse in this work.
The fact that our tensor-network states are constructed with the same dual-unitary tensor as the evolution operator $T$ makes the dynamics of these discrete geometries very natural (see Figure~\ref{flat_g_cycle}), avoiding the need of any duality map.
We construct tensor networks states representing AdS space, AdS double-sided black hole and thermal AdS. 
The discreteness of time in QCAs implies the absence of ground and thermal states, which opens the possibility of describing other spaces, like the finite piece of flat space with boundary shown in Figure~\ref{flat_g_cycle}.

\begin{figure*}
  \centering
  \includegraphics[width=175mm]{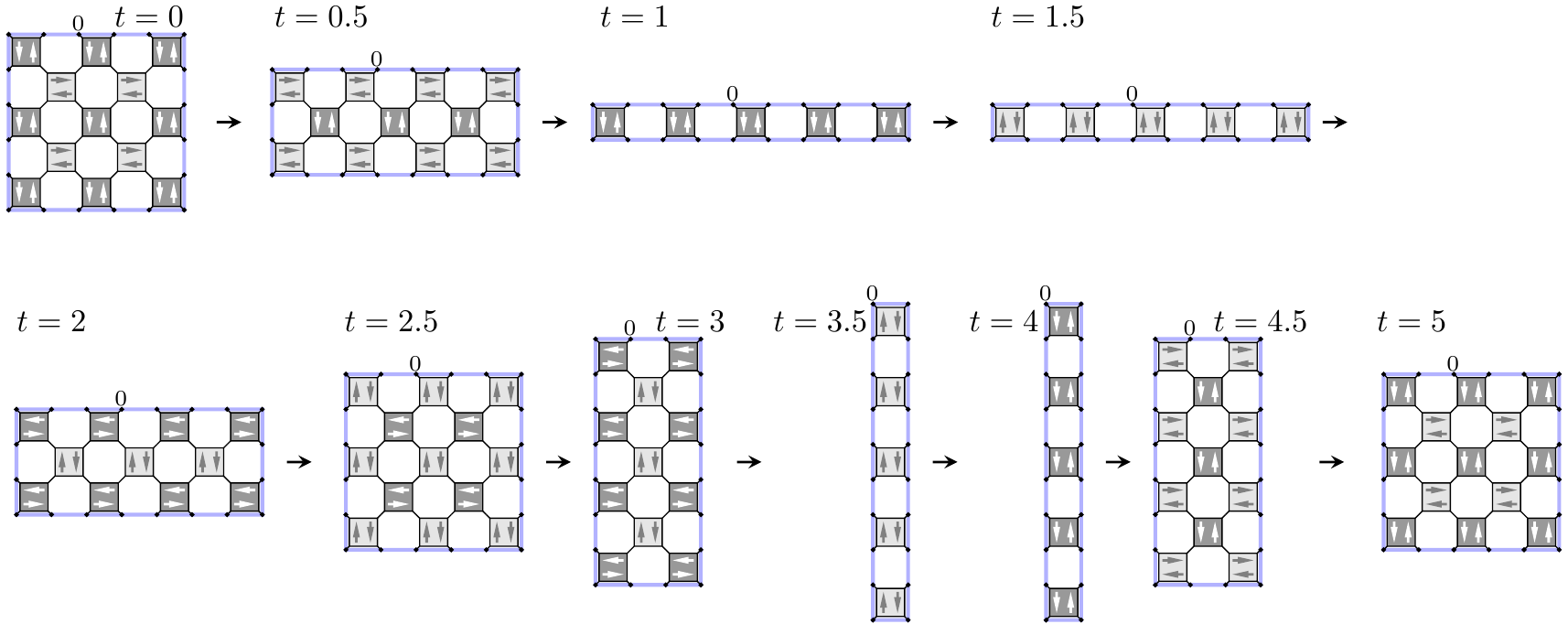}  
  \vspace{-2mm}
  \caption{\textbf{Dynamics of flat space with a boundary} of $n=20$ sites. 
  Each rectangle represents a tensor-network state of the QCA, where black dots in a blue line are the free legs that correspond to $\Cq$ systems, and site $x=0$ is marked. 
  At each time step $t$ we depict the state $T^t \! \ket{\Psi^{\rm fl}} \in \H_{20}$, and at semi-integer times $t$ we depict $T_\mathrm{even} T^{\lfloor t\rfloor}\! \ket{\Psi^{\rm fl}}$, where $\lfloor t \rfloor$ is the largest integer less than or equal to $t$. The sequence produces an orbit of period $\Delta t = 5$. If we approximate one of these rectangles by an infinite spatial strip, then this dynamics resembles that of General Relativity, where the width of the spatial strip decreases until it collapses and bounces back (see \cite{Barcelo:2000ta, Barvinsky_2006}).
  }
  \label{flat_g_cycle}
\end{figure*}

In Section~\ref{sec:matter} we analyse the dynamics of these spaces in the presence of matter, concluding that states with one or zero particles have a very different behaviour than states with two or more.
When there is only one particle, this oscillates from one point of the boundary to its antipodal counterpart, producing a closed orbit with a short period, implying that the state vector evolves within a small subspace.
%That is, the state vector evolves similarly in the cases of zero and one particle. 
On the other hand, when there are two particles or more, the dynamics explores a large space of states and does not close an orbit. In other words, the dynamics of a particle strongly depends on the presence of other particles, which can be interpreted as gravitational interaction or back-reaction.
%The complexity of this dynamics justifies describing these states by their time average, leading to a mixed state. Also, it suggests the possibility that these states contain some probability amplitude of having a black hole.
Finally, simulations indicate that conformal QCAs do not have quantum-chaotic dynamics \cite{Bertini_2021}. That is, even if we construct a circuit \eqref{intro T} with a random dual unitary $\tikz[baseline]{\porta{0}{.3}}$ (which displays Wigner-Dyson level statistics), the particular structure of the evolution operator \eqref{intro T} produces Poissonian level statistics. This is a very anomalous phenomenon that deserves to be explored with more detail in future work. 

% complexity \cite{Susskind_16}
%Holographic principle \cite{Hooft_93,Susskind_1995},

\section{Conformal QCAs}\label{sec:model}

In this section we introduce conformal QCAs and discuss some of their properties, like scale and Lorentz invariance.

%%%%%%%%%%%%%%%%%%%%%%%%%%%%%%%%%%%%%%%%%%%%%%%%%%%%%%%%%%%
\subsection{Quantum circuits in 1+1 dimensions}

Consider a chain of $n$ sites labelled by integers  $x\in \Zn$ modulo $n$, where $n$ is multiple of four.  In each site there is a quantum systems with Hilbert space $\Cq$, so the Hilbert space of the chain is $\mathcal H_n = (\Cq)^{\otimes n}$.
%Each site $x$ has and associated algebra $\A_x$ isomorphic to the matrices acting on $\mathbb C^q$. We also associate an algebra $\A_{\mathcal R}$ to any finite region $\mathcal R \in \mathbb Z$ in the standard way. 
We denote by $\A$ the algebra of matrices acting on $\Cq$, and by $\A_{x,y,\ldots}$ the algebra of matrices acting on sites $x,y,\ldots \in \Zn$.
The dynamics of the system is given by the time-translation operator $T$ defined through the circuit
\begin{align}\label{def:T}
  T &= 
  \big ( \prd{x\, \mathrm{odd}} \v_{x} \big ) 
  \big (\prd{x\, \mathrm{even}} \u_{x} \big )
  \\ &=\cdots 
\begin{tikzpicture}[baseline]
  \portau{0}{0}\cables{0}{0}
  \portav{1}{1}\cables{1}{2.2}
  \portau{2}{0}\cables{2}{0}
  \portav{3}{1}\cables{3}{2.2}
  \portau{4}{0}\cables{4}{0}
  \portav{5}{1}\cables{5}{2.2}
  \portau{6}{0}\cables{6}{0}
  \portav{7}{1}\cables{7}{2.2}
  \portau{8}{0}\cables{8}{0}
  \portav{9}{1}\cables{9}{2.2}
  \draw(1.2,-1) node {$_{\ x=}$};
  \draw(2.4,-1) node {$_{-1}$};
  \draw(3.6,-1) node {$_0$};
  \draw(4.6,-1) node {$_1$};
  \draw(5.6,-1) node {$_2$};
  \draw(1.2,2) node {$_{\ x=}$};
  \draw(2.4,2) node {$_{-1}$};
  \draw(3.6,2) node {$_0$};
  \draw(4.6,2) node {$_1$};
  \draw(5.6,2) node {$_2$};
\end{tikzpicture}
  \cdots\ \ 
\begin{tikzpicture}[baseline]
  \draw[thick,->,>=stealth] (0,-1) -- (0,2);
  \draw (.4,.5) node {$t$};
\end{tikzpicture}
\end{align}
where the two-site operators $\u_{x}, \v_{x} \in \A_{x,x+1}$ are unitary.
%This circuit is not a well-defined global operator, but instead, it defines an endomorphism on the quasi-local algebra, so that each operator $\b \in \A_{\mathcal R}$ acting in a finite region $\mathcal R$ has an image $T\b T^\dagger$. For example, if $\a\in \A_0$ then
The action of $T$ on a local operator at the origin $\a \in \A_0$ is 
\begin{align}
  T\a T^\dagger 
  %= \v_{1,2}\, \v_{3,4}\, \u_{2,3}\, \a_2\, (\v_{1,2}\, \v_{3,4}\, \u_{2,3})^\dagger
  %\u_{2,3}^\dagger \v_{1,2}^\dagger \v_{3,4}^\dagger 
  =
  \begin{tikzpicture}[baseline=-2.5]
    \portaud{0}{-1}\portau{0}{1}
    \draw (-.5,-.5)--(-.5,-.3);
    \draw (-.5,.5)--(-.5,.3);
    \draw (-.5,0) node {$\a$};
    \draw (.5,-.5)--(.5,.5);
    %\draw (.5,1.5)--(.7,1.3)--(.7,-1.3)--(.5,-1.5);
    \portav{-1}{2}\portav{1}{2}
    \portavd{-1}{-2}\portavd{1}{-2}
    \draw (-1.5,1.5)--(-1.5,-1.5);
    \draw (1.5,1.5)--(1.5,-1.5);
  \end{tikzpicture}
  \in \A_{-1,0,1,2}\ ,
\end{align}
reflecting the causality of $T$.
The (two-site) space-translation operator
\begin{align}\label{def:S}
  S=\ \cdots
  \begin{tikzpicture}[baseline]
    \draw[thin,color=black!10] (-.5,-.7) -- (-.5,1.7);
    \draw[thin,color=black!10] (.5,-.7) -- (.5,1.7);
    \draw[thin,color=black!10] (1.5,-.7) -- (1.5,1.7);
    \draw[thin,color=black!10] (2.5,-.7) -- (2.5,1.7);
    \draw[thin,color=black!10] (3.5,-.7) -- (3.5,1.7);
    \draw[thin,color=black!10] (4.5,-.7) -- (4.5,1.7);
    \draw[thin,color=black!10] (5.5,-.7) -- (5.5,1.7);
    \draw[thin,color=black!10] (6.5,-.7) -- (6.5,1.7);
    \draw[thin,color=black!10] (7.5,-.7) -- (7.5,1.7);
    \draw[thin,color=black!10] (8.5,-.7) -- (8.5,1.7);
    \cables{0}{0} \barra{0}{1} \cables{0}{2.2}
    \barra{0}{0} \barra{1}{1}
    \barra{1}{0} \barra{2}{1}
    \cables{2}{0} \cables{2}{2.2}
    \barra{2}{0} \barra{3}{1}
    \barra{3}{0} \barra{4}{1}
    \barra{4}{0} \barra{5}{1}
    \cables{4}{0} \cables{4}{2.2}
    \barra{5}{0} \barra{6}{1}
    \barra{6}{0} \barra{7}{1}
    \cables{6}{0} \cables{6}{2.2}
    \barra{7}{0} \barra{8}{1}
    \barra{8}{0}
    \cables{8}{0} \cables{8}{2.2}
  \end{tikzpicture}
  \cdots
\end{align}
allows for writing the translation-invariance of the dynamics as
\begin{align}
  S T = T S\ .
\end{align}

%%%%%%%%%%%%%%%%%%%%%%%%%%%%%%%%%%%%%%%%%%%%%%%%%%%%%%%%%%%
\subsection{Dual unitaries}\label{sec:dual unitaries}

If we represent the two-site unitary $\u\in \A_{0,1}$ as a four-legged tensor $\u= \tikz[baseline]{\porta{0}{.3}}$ then we can write its transpose as $\u^T=\tikz[baseline]{\portat{0}{.3}}\,$ and its conjugation by the swap operator $\s= \tikz[baseline]{\draw(-.4,-.1)--(.4,.7); \draw(-.4,.7)--(.4,-.1);}\,$ as $\s\u \s^\dagger = \tikz[baseline] {\portas{0}{.3}}\,$.
%\begin{align}
%  \begin{tikzpicture}[baseline]
%    \cables{0}{.2}\cables{0}{1.4}
%    \portas{0}{.2}
%  \end{tikzpicture}
%  =
%  \begin{tikzpicture}[baseline]
%    \cables{0}{-0.8}\cables{0}{2.4}
%    \porta{0}{.2}
%    \barra{0}{1.2}\antibarra{0}{1.2}
%    \barra{0}{-.8}\antibarra{0}{-.8}
%  \end{tikzpicture}\ .
%\end{align} 
Also, we denote complex conjugation with a darker shade $\u^*=\tikz[baseline]{\portac{0}{.3}}\,$, so that the Hermitian conjugate is $\u^\dagger =\tikz[baseline]{\portad{0}{.3}}\,$.
With this notation we can say that $\tikz[baseline]{\porta{0}{.3}}$ is unitary if it satisfies the two equivalent conditions
\begin{align}\label{eq:unitarity}
  \begin{tikzpicture}[baseline]
    \cables{0}{-.5} \portad{0}{-.5}
    \cables{0}{.7} \porta{0}{.7} \cables{0}{1.9}
    \draw (1.2,.1) node {$=$};
    \draw[thin] (2,-1.2) -- (2,1.4);
    \draw[thin] (3,-1.2) -- (3,1.4);
  \end{tikzpicture}
  \quad \mbox{and} \quad
  \begin{tikzpicture}[baseline]
    \cables{0}{-.5} \porta{0}{-.5}
    \cables{0}{.7} \portad{0}{.7} \cables{0}{1.9}
    \draw (1.2,.1) node {$=$};
    \draw[thin] (2,-1.2) -- (2,1.4);
    \draw[thin] (3,-1.2) -- (3,1.4);
  \end{tikzpicture}\ .
\end{align}
Also, we say that $\tikz[baseline]{\porta{0}{.3}}$ is a dual unitary if, in addition to unitarity, it satisfies the two equivalent conditions
\begin{align}\label{eq:dual unitarity}
  \begin{tikzpicture}[baseline]
    \porta{0}{.3} \portacs{1.2}{.3}
    \cablesh{-1.2}{.3}\cablesh{0}{.3}
    \cablesh{1.2}{.3}
    \draw (2.6,.2) node {$=$};
    \draw[thin] (3.4,.8) -- (4.6,.8);
    \draw[thin] (3.4,-.2) -- (4.6,-.2);
  \end{tikzpicture}
  \quad \mbox{and} \quad
  \begin{tikzpicture}[baseline]
    \portacs{0}{.3} \porta{1.2}{.3}
    \cablesh{-1.2}{.3}\cablesh{0}{.3}
    \cablesh{1.2}{.3}
    \draw (2.6,.2) node {$=$};
    \draw[thin] (3.4,.8) -- (4.6,.8);
    \draw[thin] (3.4,-.2) -- (4.6,-.2);
  \end{tikzpicture}
  \ ,
\end{align}
which can be phrased as ``unitarity in the space direction".

Dual unitarity implies that, for any traceless local operator $\a\in \A_0$, the partial trace of $\u\a\u^\dag \in \A_{0,1}$ on site $1$ vanishes,
\begin{align}
  \mathrm{tr}_1\! (\u_0\a_0\u_0^\dag)
  =
  \begin{tikzpicture}[baseline=-2]
    \portad{0}{-1}\porta{0}{1}
    \draw (-.5,-.5)--(-.5,-.3);
    \draw (-.5,.5)--(-.5,.3);
    \draw (-.5,0) node {$\a$};
    \draw (.5,-.5)--(.5,.5);
    \draw (.5,1.5)--(.7,1.5)--(.9,1.3)--(.9,-1.3)--(.7,-1.5)--(.5,-1.5);
  \end{tikzpicture}
  =\,
  \begin{tikzpicture}[baseline=-2]
    \draw (-.5,.3)--(-.5,.5)--(-.4,.6)--(-.1,.6)--(0,.5)--(0,-.5)--(-.1,-.6)--(-.4,-.6)--(-.5,-.5)--(-.5,-.3);
    \draw (-.5,0) node {$\a$};
    \draw (-1,-1.4)--(-1,1.4);
  \end{tikzpicture}
  = \unity_{\!0}\, \mathrm{tr} (\a) =0\ .
\end{align}
%$\A_{\{0,1\}} =  \A_0 + \A_1 + (\A_0 + \A_1)^\perp$. In other words, the evolved operator $\u\a\u^\dagger$ is a linear combination of terms of the form $\unity \otimes \c$ and $\b \otimes \c$ with traceless $\b,\c$, but not of the form $\b\otimes \unity$.
In other words, the operator $\u\a\u^\dagger \in \A_{0,1}$ is a linear combination of terms of the form $\unity_0 \otimes \c_1$ and $\b_0 \otimes \c_1$ with traceless $\b, \c$, but not of the form $\b_0 \otimes \unity_1$.
That is, all terms in $\u\a\u^\dagger$ act non-trivially on site $1$.

Returning to our model, if $\u,\v$ are dual unitaries then a local operator $\a \in \A_x$ on an even site $x$ evolving as $T^t \a T^{-t}$ grows towards the right at maximal speed, which in lattice units is $c=2$.
The operator $T^t \a T^{-t}$ may also grow towards the left and develop a highly non-local and complex structure, but that is not necessary.
Alternatively, if the initial operator $\a \in \A_x$ is located on an odd site $x$ then $T^t \a T^{-t}$ grows towards the left at maximal speed $-2$.
Hence, we see that the even/odd location $x\in \Zn$ plays the role of a momentum $\pm 2$ quantum number.
In summary, every perturbation in a dual-unitary circuit grows at maximal speed towards the right, the left, or both, as in CFT.

%%%%%%%%%%%%%%%%%%%%%%%%%%%%%%%%%%%%%%%%%%%%%%%%%%%%%%
\subsection{Free particles and quantum chaos}

In order to simplify the discussion of this subsection (only) we restrict ourselves to circuits \eqref{def:T} with $\v=\u$.
The first thing to do when we are given a dual unitary $\u=\tikz[baseline]{\porta{0}{.3}}\,$ is to obtain the spectral decomposition of the maps $\Omega_+: \A_0 \to \A_1$ and $\Omega_-: \A_1 \to \A_0$, defined as
\begin{align}
  \Omega_+(\a_0) 
  &= \frac 1 q \, \mathrm{tr}_0 (\u_0\a_0\u_0^\dag)
  = \frac 1 q \,
  \begin{tikzpicture}[baseline=-2]
    \portad{0}{-1}\porta{0}{1}
    \draw (-.5,-.5)--(-.5,-.3);
    \draw (-.5,.5)--(-.5,.3);
    \draw (-.5,0) node {$\a$};
    \draw (.5,-.5)--(.5,.5);
    \draw (-.5,1.5)--(-.7,1.5)--(-.9,1.3)--(-.9,-1.3)--(-.7,-1.5)--(-.5,-1.5);
  \end{tikzpicture}\ ,
  \\
  \Omega_-(\a_1) 
  &= \frac 1 q \, \mathrm{tr}_1 (\u_0\a_1\u_0^\dag)
  = \frac 1 q \ \, 
  \begin{tikzpicture}[baseline=-2]
    \portad{0}{-1}\porta{0}{1}
    \draw (.5,-.5)--(.5,-.3);
    \draw (.5,.5)--(.5,.3);
    \draw (.5,0) node {$\a$};
    \draw (-.5,-.5)--(-.5,.5);
    \draw (.5,1.5)--(.7,1.5)--(.9,1.3)--(.9,-1.3)--(.7,-1.5)--(.5,-1.5);
  \end{tikzpicture}\ .
\end{align}
The eigenvectors of $\Omega_+$ with unimodular eigenvalue $\Omega_+(\e) = e^{im} \e$ satisfy 
\begin{align}
  T\e_x T^\dagger = e^{im} \e _{x+2}\ ,
\end{align}
for all even $x$.
Analogously, the eigenvectors of $\Omega_-$ with unimodular eigenvalue $\Omega_-(\e) = e^{im} \e$ satisfy 
\begin{align}
  T\e_x T^\dagger = e^{im} \e _{x-2}\ ,
\end{align}
for all odd $x$.
In the dual-unitary and QCA literature, these operators are respectively called right/left-moving solitons \cite{kos_20} and gliders \cite{G_tschow_2010}.
When acting on a state, these operators can create a free particle with velocity $\pm 2$ and quasi-mass $m$.

The eigenvectors $\Omega_\pm (\e) = \lambda \e$ with eigenvalue modulus less than one $|\lambda| <1$ grow under $T$ in a scrambled fashion which fills up all the lightcone. This dynamics displays many signatures of quantum chaos, including the profile of the spectral form factor \cite{Bertini_2019, Bertini_2021, Piroli_2020}.

%Massive solitons provide a time scale on the evolution, preventing the theory to be scale invariant.

%%%%%%%%%%%%%%%%%%%%%%%%%%%%%%%%%%%%%%%%%%%%%%%%%%%%%%%%%%%
\subsection{Definition of conformal QCA} 

%stitching spacetime tucks with entanglement

Let us define the family of circuits introduced and analysed in this work. For any given dual unitary $\u = \tikz[baseline]{\porta{0}{.3}}$ we define the following time-translation operator
\begin{align}\label{eq:T gauge}
  T=\cdots 
\begin{tikzpicture}[baseline=13]
  \portadrr{-1}{3}\portadr{-1}{1}\cables{-1}{4.2}
  \cables{0}{0}
  \porta{0}{0} \portadrr{1}{1} \portarr{0}{2} \portadr{1}{3}
  \cables{1}{4.2}
  \cables{2}{0}
  \portarr{2}{0} \portadr{3}{1} \porta{2}{2} \portadrr{3}{3}
  \cables{3}{4.2}
  \cables{4}{0}
  \porta{4}{0} \portadrr{5}{1} \portarr{4}{2} \portadr{5}{3}
  \cables{5}{4.2}
  \cables{6}{0}
  \portarr{6}{0} \portadr{7}{1} \porta{6}{2} \portadrr{7}{3}
  \cables{7}{4.2}
  \cables{8}{0} 
  \porta{8}{0} \portadrr{9}{1} \portarr{8}{2} \portadr{9}{3}
  \cables{9}{4.2}
  \draw(3.6,-1) node {$_0$};
  \draw(3.6,4) node {$_0$};
\end{tikzpicture}
  \cdots
\end{align}
where site $x=0$ is marked.
Note that, for any four local unitaries $\a, \b, \c, \d$, the new dual unitary 
\begin{align}
  \u' = \a_0 \b_1 \u\, \c_0 \d_1 = 
  \begin{tikzpicture}[baseline=-1.6]
    \porta{0}{0} \cables{0}{0} \cables{0}{1.2}
    \draw (-.5,-1.05) node {$\a$};
    \draw (.5,-1) node {$\b$};
    \draw (-.5,1) node {$\c$};
    \draw (.5,1.05) node {$\d$};
  \end{tikzpicture}
\end{align}
defines a new circuit $T'$ via \eqref{eq:T gauge}  which is equal to $T$ up to a local change of basis, 
\begin{align}
  \nonumber
  T'= (\cdots \a_0 \b_1 \d_2 \c_3 \a_4 \cdots) T 
  (\cdots \a_0 \b_1 \d_2 \c_3 \a_4 \cdots)^\dagger\ .
\end{align}
This reminds the structure of a gauge theory, but it is not the same.
%Hence, each dual unitary $\tikz[baseline]{\porta{0}{.3}}$ defines a gauge theory \eqref{eq:T gauge} with gauge group ${\rm SU}(q)$.

Another property of the structure of $T$ is that it produces a cancellation of the phases of travelling solitons, forcing them all to be massless. 
That is, if $\e_0$ satisfies $\u_0 \e_0 \u_0 = e^{im} \e_1$ then $T\e_x T^\dagger = \e_{x+4}$ for all even $x$ (and analogously for odd $x$).
Recall that CFTs can only have massless particles - however this is not sufficient to be conformal. In the next section we show that $T$ is invariant under scale transformations.

\subsection{Scale invariance}

The discussion in this section requires the size of the chain $n$ to be a multiple of $8$.
We start by defining the contraction isometry $C: \H_n \to \H_{\frac n 2}$ as
\begin{align}\label{def:C}
  C = 
  q^{-\frac n 8} \Big(\cdots
  \begin{tikzpicture}[baseline=-4]
  \draw (-.5,-1) node {\small $_0$};
  \barra{-1}{0}\porta{0}{0}\antibarra{1}{0}
  \cables{-1}{0}\cables{1}{0}
  \draw (2.5,-.7)--(2.5,.5);
  \draw (3.5,-.7)--(3.5,.5);
  \draw (4.5,-.7)--(4.5,.5);
  \draw (5.5,-.7)--(5.5,.5);
%  \draw (-.5,-1) node {$_x$};
  \barra{7}{0}\porta{8}{0}\antibarra{9}{0}
  \cables{7}{0}\cables{9}{0}
  \draw (10.5,-.7)--(10.5,.5);
  \draw (11.5,-.7)--(11.5,.5);
  \draw (12.5,-.7)--(12.5,.5);
  \draw (13.5,-.7)--(13.5,.5);
  \end{tikzpicture}
  \cdots\Big)
\end{align}
which maps an $n$-site chain to an $\frac n 2$-site chain. 
The particular structure of operators $C$ and $T$ together with the dual unitarity of their building block $\tikz[baseline]{\porta{2}{.3}}$ allows to easily calculate the product $CT$ in a manner that is independent of the choice of $\tikz[baseline]{\porta{2}{.3}}\,$: 
\begin{align}
  \nonumber
  CT=&\cdots 
\begin{tikzpicture}[baseline=13]
  \draw (2.5,3.5)--(2.5,4.5);
  \draw (3.5,3.5)--(3.5,4.5);
  \draw (4.5,3.5)--(4.5,4.5);
  \draw (5.5,3.5)--(5.5,4.5);
  \portadrr{-1}{3}\portadr{-1}{1}
  %\draw(2.5,4.5)--(2.5,4.7);
  \cables{0}{0}
  \porta{0}{4}\barra{-1}{4}\antibarra{1}{4}
  \porta{0}{0} \portadrr{1}{1} \portarr{0}{2} \portadr{1}{3}
  %\cables{3}{4.2}
  \cables{2}{0}
  \portarr{2}{0} \portadr{3}{1} \porta{2}{2} \portadrr{3}{3}
  %\cables{3}{4.2}
  \cables{4}{0}
  \porta{4}{0} \portadrr{5}{1} \portarr{4}{2} \portadr{5}{3}
  %\cables{5}{4.2}
  \cables{6}{0}
  \portarr{6}{0} \portadr{7}{1} \porta{6}{2} 
  \portadrr{7}{3}
  %\cables{7}{4.2}
  \cables{8}{0} 
  \porta{8}{0} \portadrr{9}{1} \portarr{8}{2} 
  \portadr{9}{3}
  %\cables{9}{4.2}
  \porta{8}{4}\barra{7}{4}\antibarra{9}{4}
\end{tikzpicture}
  \cdots
  \\ \nonumber
  =&\cdots 
\begin{tikzpicture}[baseline=13]
  \draw (2.5,3.5)--(2.5,4.5);
  \draw (3.5,3.5)--(3.5,4.5);
  \draw (4.5,3.5)--(4.5,4.5);
  \draw (5.5,3.5)--(5.5,4.5);
  \barra{-1}{3}
  \portadr{-1}{1}
  \cables{0}{0}
  \barra{0}{4}\barra{0}{3}
  \antibarra{1}{4}
  \porta{0}{0} \portadrr{1}{1} \portarr{0}{2} \portadr{1}{3}
  \cables{2}{0}
  \portarr{2}{0} \portadr{3}{1} \porta{2}{2} \portadrr{3}{3}
  %\cables{3}{4.2}
  \cables{4}{0}
  \porta{4}{0} \portadrr{5}{1} \portarr{4}{2} \portadr{5}{3}
  %\cables{5}{4.2}
  \cables{6}{0}
  \portarr{6}{0} \portadr{7}{1} \porta{6}{2} 
  \cables{8}{0} 
  \porta{8}{0} \portadrr{9}{1} \portarr{8}{2} 
  \portadr{9}{3}
  \barra{8}{4}\barra{7}{3}\antibarra{9}{4}
  \barra{8}{3}
\end{tikzpicture}
  \cdots
  \\ \nonumber
  =&\cdots 
\begin{tikzpicture}[baseline=13]
  \draw (2.5,3.5)--(2.5,4.5);
  \draw (3.5,3.5)--(3.5,4.5);
  \draw (4.5,3.5)--(4.5,4.5);
  \draw (5.5,3.5)--(5.5,4.5);
  \barra{-1}{3}
  \portadr{-1}{1}
  \cables{0}{0}
  \barra{0}{4}\barra{1}{3}
  \barra{0}{2}\barra{1}{2}
  \antibarra{1}{4}
  \porta{0}{0} \portadrr{1}{1} 
  %\portarr{0}{2} \portadr{1}{3}
  \cables{2}{0}
  \portarr{2}{0} \portadr{3}{1} \porta{2}{2} \portadrr{3}{3}
  \cables{3}{4.2}
  \cables{4}{0}
  \porta{4}{0} \portadrr{5}{1} \portarr{4}{2} \portadr{5}{3}
  \cables{5}{4.2}
  \cables{6}{0}
  \portarr{6}{0} \portadr{7}{1} \porta{6}{2} 
  \cables{8}{0} 
  \porta{8}{0} \portadrr{9}{1} 
  %\portarr{8}{2} \portadr{9}{3}
  \barra{8}{4}\barra{7}{3}\antibarra{9}{4}
  \barra{8}{2}\barra{9}{3}\barra{9}{2}
\end{tikzpicture}
  \cdots
  \\ 
  =&\cdots 
\begin{tikzpicture}[baseline=13]
  \portadr{-1}{1}
  \cables{0}{0}
  \barra{2}{1}
  \antibarra{-1}{2}\porta{0}{0} 
  \barra{1}{1} 
  \cables{2}{0}
  \portarr{2}{0} \portadr{3}{1} \barra{2}{2} \portadrr{3}{3}
  \cables{3}{4.2}
  \cables{4}{0}
  \porta{4}{0} \portadrr{5}{1} \portarr{4}{2} \portadr{5}{3}
  \cables{5}{4.2}
  \cables{6}{0}
  \portarr{6}{0} \portadr{7}{1} \porta{6}{2} 
  \cables{8}{0} 
  \porta{8}{0} \portadrr{9}{1} 
  \antibarra{7}{2}
  \barra{9}{2}
\end{tikzpicture}
  \cdots
\end{align}
Continuing in a similar fashion we obtain
\begin{align}  \nonumber
  CT=&\cdots 
\begin{tikzpicture}[baseline=13]
  \antibarra{-1}{1}
  \cables{0}{0}
  \barra{3}{0}
  \porta{0}{0} 
  \barra{1}{1} 
  \cables{2}{0}
  \barra{2}{0} 
  \barra{3}{1} \barra{2}{2} \portadrr{3}{3}
  \cables{3}{4.2}
  \cables{4}{0}
  \porta{4}{0}\antibarra{5}{0} 
  \antibarra{5}{1} \portarr{4}{2} \portadr{5}{3}
  \cables{5}{4.2}
  \cables{6}{0}
  \antibarra{6}{0}\antibarra{7}{1}\antibarra{6}{2} 
  \cables{8}{0} 
  \porta{8}{0} \barra{9}{1} 
  %\antibarra{7}{2}
  %\barra{9}{2}
\end{tikzpicture}
  \cdots
  \\ \label{CT=TC}
  =&\cdots 
\begin{tikzpicture}[baseline=13]
  \portadrr{-1}{3}\cables{-1}{4.2}
  \portadr{1}{3}\cables{1}{4.2}
  \portarr{0}{2}\porta{2}{2}
  \barra{-1}{0}\barra{-1}{1}\barra{0}{1}
  \cables{-1}{0}
  \barra{3}{0}\barra{0}{0}\barra{1}{1} 
  \cables{1}{0}
  \barra{2}{0}\barra{3}{1}\barra{2}{1}\barra{1}{0} 
  \portadrr{3}{3}
  \cables{3}{4.2}
  \cables{3}{0}
  \porta{4}{0}\antibarra{5}{0} 
  \antibarra{5}{1}\portarr{4}{2}
  \portadr{5}{3}\cables{5}{4.2}
  \portadr{9}{3} 
  \cables{7}{4.2}\cables{9}{4.2}
  \cables{5}{0}
  \antibarra{6}{0}\antibarra{7}{1}
  \porta{6}{2}\portadrr{7}{3}\portarr{8}{2} 
  \cables{7}{0}\cables{9}{0} 
  \antibarra{8}{0}\antibarra{7}{0}\antibarra{6}{1} \antibarra{9}{1}\antibarra{8}{1}\antibarra{9}{0}
  %\barra{9}{2}
\end{tikzpicture}
  \cdots
\end{align}
The above can be synthesised as
\begin{align}
  C_{2n} T_{2n}^2 = T_{n} C_{2n} \ ,  
\end{align}
where we have added a subindex to the operators to indicate the size of the chain where they act on.
The above equation tells us that, the action of $C_n$ produces a rescaling of space and time by a factor $\mitg$.

By using the dual unitary constraints \eqref{eq:unitarity} and \eqref{eq:dual unitarity} we can calculate the action of $C$ on a local operator $\a \in \A_x$, 
\begin{align}\label{eq:C on a}
  C \a_x C^\dagger  = \left\{
  \begin{array}{ll}
    0 & \mbox{if $x=0,1,2,3$ mod 8}
    \\
    \a_{\frac x 2} & \mbox{if $x=4,5,6,7$ mod 8}
  \end{array}
  \right.\ .
\end{align}
This action depends on the position $x$ in a very non-smooth way.
However, the action of $C$ is smooth on operators having a good field-theory limit
\begin{align}\label{def: Phi}
  \Phi_n(x) = \sum_{y\in \Zn} \varphi(y-x) \a_y\ ,
\end{align}
where $\varphi(y)$ is a smearing function (e.g.~a gaussian) centred around the origin and spreaded over a large number of lattice units. In this case we have
\begin{align}\label{action of C}
  C \Phi_n(x) C^\dagger 
  \approx \frac 1 2 \Phi_{\frac n 2} (\mbox{$\frac x 2$})\ ,
\end{align}
where the subindex of $\Phi_n$ stresses that the field operators on the left- and right-hand sides act on chains of different sizes.

It is possible to construct contraction isometries with scale factor different than $\mitg$. This can be achieved by separating the vectors $\tikz[baseline]{\barra{1}{.3}\porta{2}{.3}\antibarra{3}{.3}}$ in \eqref{def:C} by a length different than 4.
%An alternative isometry $C'$ implementing a contraction is
%\begin{align}
%  C' = 
%  q^{-n/4} \Big(\cdots
%  \begin{tikzpicture}[baseline=-4]
%  \barra{-1}{0}\porta{0}{0}
%  \cables{-1}{0}\cables{1}{0}
%  \porta{2}{0}\antibarra{3}{0}
%  \cables{3}{0}\cables{1}{1.2}
%  \draw (4.5,-.7)--(4.5,.7);
%  \draw (5.5,-.7)--(5.5,.7);
  %
%  \barra{7}{0}\porta{8}{0}
%  \cables{7}{0}\cables{9}{0}
%  \porta{10}{0}\antibarra{11}{0}
%  \cables{11}{0}\cables{9}{1.2}
%  \draw (12.5,-.7)--(12.5,.7);
%  \draw (13.5,-.7)--(13.5,.7);  
%  \draw (-.5,-1) node {$_x$};
%  \barra{7}{0}\porta{8}{0}\antibarra{9}{0}
%  \cables{7}{0}\cables{9}{0}
%  \draw (10.5,-.7)--(10.5,.5);
%  \draw (11.5,-.7)--(11.5,.5);
%  \draw (12.5,-.7)--(12.5,.5);
%  \draw (13.5,-.7)--(13.5,.5);
%  \end{tikzpicture}
%  \cdots\Big) \ .
%\end{align}

\subsection{Lorentz transformations}\label{sub LT}

In this subsection we give a brief summary of Lorentz transformations on conformal QCAs, and refer to Section~\ref{sec:lorentz transf} for the complete presentation.
Lorentz transformations are more clearly discussed in an infinite chain, so here and in Section~\ref{sec:lorentz transf} we assume $n=\infty$.

In Section~\ref{sec:lorentz transf} we define the isometry $R_l :\H_\infty \to \H_\infty$ which jointly implements a contraction and a spacetime transformation which resembles a Lorentz boost towards the right, with velocity $v$ parametrised by the positive integer $l$ as 
\begin{align}\label{eq:v of l}
  v = \frac 2 {\sqrt{4-2 l +l^2}}\ .
\end{align}
These transformations commute with the translations in the diagonal direction $(x,t)=(1,1)$, 
\begin{align}
  R_l (ST) = (ST) R_l\ .
\end{align}

The action of $R_l$ on a local operator $\a(x,t) := T^t \a_x T^{-t}$ can be informally described as 
\begin{align}\label{eq:21}
  R_l \a(x,t) R_l 
  = \left\{
  \begin{array}{ll}
    \a(x',t') & \mbox{ for most $(x,t)$}
    \\
    \mbox{complicated} & \mbox{ for a few $(x,t)$}
  \end{array}\right.
\end{align}
where, in the regime $|x|\gg l$, the transformed coordinates $(x',t')$ can be written as
\begin{align}\label{B-transform}
  \left.\begin{array}{ll}
    x' &= 
    \left(1-\frac 1 l \right) x +\frac 2 l t 
    \\
    t' &= 
    \left(1-\frac 1 l \right) t +\frac 1 {2l} x
  \end{array}\right\} \ .   
\end{align}
The label ``complicated" in \eqref{eq:21} stands for a transformation that is not purely spacetime, as in the first case.
Next, note that transformation \eqref{B-transform} preserves Minkowski's metric up to a scale factor 
\begin{align}
  (ct')^2-x'^2 = \left(1-\frac 2 l \right) \left[ (ct)^2-x^2 \right]\ .
\end{align}
(Recall that the speed of light is $c=2$.) 
Now, we can remove the scale transformation from \eqref{B-transform} by dividing the new coordinates $(x',t')$ by the scale factor $\sqrt{1-2/l}$, so obtaining the pure Lorentz transformation. 
Once this is done, we simplify the first equation by imposing $x=0$, obtaining
\begin{align}
  \frac {x'} {\sqrt{1-2/l}} = \frac {2/l} {\sqrt{1-2/l}} t = \frac {v} {\sqrt{1-v^2}} t\ ,
\end{align}
where the second equality follows from the standard form of a Lorentz transformation with velocity $v$ and $x=0$. 
This second equality can be use to isolate $v$ as a function of $l$ and confirm the relation \eqref{eq:v of l}.

In Section~\ref{sec:lorentz transf} we also define the isometry $L_l :\H_\infty \to \H_\infty$, which jointly implements a contraction and a spacetime transformation which resembles a Lorentz boost towards the left, with velocity $v$ parametrised by the positive integer $l$ as 
\begin{align}\label{eq:-v of l}
  v = \frac {-2} {\sqrt{4-2 l +l^2}}\ .
\end{align}
The action of $L_l$ on a local operator $\a(x,t) := T^t \a_x T^{-t}$ is 
\begin{align}\label{eq:27}
  L_l \a(x,t) L_l 
  = \left\{
  \begin{array}{ll}
    \a(x',t') & \mbox{ for most $(x,t)$}
    \\
    \mbox{complicated} & \mbox{ for a few $(x,t)$}
  \end{array}\right.
\end{align}
where, in the regime $|x|\gg l$, the transformed coordinates $(x',t')$ can be written as
\begin{align}\label{L-transform}
  \left.\begin{array}{ll}
    x' &= 
    \left(1-\frac 1 l \right) x -\frac 2 l t 
    \\
    t' &= 
    \left(1-\frac 1 l \right) t -\frac 1 {2l} x
  \end{array}\right\} \ .   
\end{align}
Note that this transformation also preserves Minkowski's metric up to the same scale factor $\sqrt{1-2/l}$. 

Like with the scale transformations described in the previous subsection, the action of $B_l$ and $L_l$ becomes smooth on operators of the form $\Phi(x,t) = T^t \Phi(x) T^{-t}$, where the smeared operator $\Phi(x)$ is defined in \eqref{def: Phi}. In particular we have
\begin{align}\label{eq:smooth R}
  R_l \Phi(x,t) R_l ^\dagger
  \approx 
  \left(1-\frac 1 l\right) \Phi(x',t') 
\end{align}
for all $(x,t)$, not just ``most", avoiding the ``complicated" cases of \eqref{eq:21} and \eqref{eq:27}. Naturally, the coordinates $(x',t')$ in \eqref{eq:smooth R} satisfy \eqref{B-transform}.

Interestingly, the conjugated operators $R_l^\dagger$ and $L_l^\dagger$ implement a Lorentz boost together with a dilation (instead of a contraction). And in this case, the boost directions are reversed: $R_l^\dagger$ is a left boost and $L_l^\dagger$ a right boost.
Therefore, the composition $R_l^\dagger L_{l}$ produces a Lorentz transformation without a scale transformation, but it has a non-trivial kernel, and hence, it is not an isometry.
 
\subsection{Two-layer conformal circuits}

To simplify notation we redefine $\u$ as the coarse-grained (dual) unitary
\begin{align}
  \begin{tikzpicture}[baseline]
    \draw (-10,1) node {$\u\ =$};
    \portabar{-7}{1}
    \barra{-6}{2} \barra{-8}{0}
    \antibarra{-8}{2} \antibarra{-6}{0}
    \draw (-4,1) node {$=$};
    \portadr{-1}{1} \portadrr{1}{1}
    \porta{0}{0} \portarr{0}{2} 
    \barra{2}{2} \barra{1}{3} 
    \barra{-1}{-1} \barra{-2}{0}
    \antibarra{-2}{2} \antibarra{-1}{3}
    \antibarra{1}{-1} \antibarra{2}{0}
\end{tikzpicture}
\end{align}
%and $\v = \tikz[baseline]{\portacbar{0}{.3}}$,
where the double-arrow notation encapsulates the new symmetry $\u^T = \s \u \s^\dagger$.
We also redefine $q$ so that the Hilbert space of a coarse-grained site has dimension $q$.
%$\tikz[baseline]{\porta{0}{.3}} := \tikz[baseline]{\portau{0}{.3}}$ 
%\begin{align}
%\begin{tikzpicture}
%  \porta{0}{0}
%  \draw (1,0) node {$=$};
%  \portarr{2}{0}
%  \draw (4,0) node {and};
%  \portav{6}{0}
%  \draw (7,0) node {$=$};
%  \portadrr{8}{0}
%  \draw (9,-.3) node {.};
%\end{tikzpicture}
%\end{align}
%Note that the dual unitarity of $\u$ implies the dual unitarity of $\v$ and viceversa. Also, the first condition can be written as $\u^T = \s \u \s^\dagger$ or $\v^T = \s \v \s^\dagger$.
Now we can write the evolution operator \eqref{eq:T gauge} as a two-layer circuit 
\begin{align}\label{eq:2 layer}
  T = \cdots
\begin{tikzpicture}[baseline=2.5]
  \portabar{0}{0}\cables{0}{0}\cables{1}{2.2}
  \portabar{2}{0}\cables{2}{0}\cables{3}{2.2}
  \portabar{4}{0}\cables{4}{0}\cables{5}{2.2}
  \portabar{6}{0}\cables{6}{0}\cables{7}{2.2}
  \portabar{8}{0}\cables{8}{0}\cables{9}{2.2}
  \portacbar{1}{1}
  \portacbar{3}{1}
  \portacbar{5}{1}
  \portacbar{7}{1}
  \portacbar{9}{1}
\end{tikzpicture}  
  \cdots
\end{align}
In the rest of this work, our starting point is a dual unitary $\u = \tikz[baseline]{\portabar{2}{.3}}$ with the symmetry $\u^T = \s\u\s^\dagger$, and the time-translation operator \eqref{eq:2 layer}. Note that this dynamics is more general than the coarse-grained four-layer circuit.

\section{Discrete holography}\label{sec:holography empty}

In this section we construct tensor-network states for 1+1 QCAs, and show that they can be interpreted as spacial slices of 2+1 discrete geometries with metric distance defined by their entanglement structure.

\subsection{Tensor-network states and dynamics}\label{sec:flat}

By using the four-legged tensor \tikz[baseline]{\portabar{0}{.3}}\,, its complex conjugated and their rotated versions, we can construct tensor-network states for the chain $\H_n$. One example of such states for a chain of $n=20$ sites is 
\begin{align}%\label{eq:2 layer}
  \ket{\Psi^{\rm fl}} = \frac 1 {q^5}\, 
\begin{tikzpicture}[baseline=-2]
  \draw[blue!30, very thick] (-2.5,2.5)--(2.5,2.5)--(2.5,-2.5)--(-2.5,-2.5)--(-2.5,2.5);
  \portaabar{-2}{-2} \portaabar{0}{-2} \portaabar{2}{-2}
  \portacabar{-1}{-1} \portacabar{1}{-1}
  \portaabar{-2}{0} \portaabar{0}{0} \portaabar{2}{0}
  \portacabar{-1}{1} \portacabar{1}{1}
  \portaabar{-2}{2} \portaabar{0}{2} \portaabar{2}{2}
  %\cables{0}{0}\cables{1}{2.2}
  \filldraw (-2.5,-2.5) circle (.05);
  \filldraw (-1.5,-2.5) circle (.05);
  \filldraw (-.5,-2.5) circle (.05);
  \filldraw (.5,-2.5) circle (.05);
  \filldraw (1.5,-2.5) circle (.05);
  \filldraw (2.5,-2.5) circle (.05);
  \filldraw (-2.5,2.5) circle (.05);
  \filldraw (-1.5,2.5) circle (.05);
  \filldraw (-.5,2.5) circle (.05);
  \filldraw (.5,2.5) circle (.05);
  \filldraw (1.5,2.5) circle (.05);
  \filldraw (2.5,2.5) circle (.05);
  \filldraw (-2.5,-1.5) circle (.05);
  \filldraw (-2.5,-.5) circle (.05);
  \filldraw (-2.5,.5) circle (.05);
  \filldraw (-2.5,1.5) circle (.05);
  \filldraw (2.5,-1.5) circle (.05);
  \filldraw (2.5,-.5) circle (.05);
  \filldraw (2.5,.5) circle (.05);
  \filldraw (2.5,1.5) circle (.05);
  \draw (-1.8,2.9) node {$_{-1}$};
  \draw (-.5,2.9) node {$_0$};
  \draw (.5,2.9) node {$_1$};
  \draw (1.5,2.9) node {$_2$};
\end{tikzpicture}\  ,
\end{align}
where each black dot in the blue line represents a free leg of the tensor network, and hence, a $\Cq$ system of the chain. 
The label ``fl" stands for ``flat".
Naturally, the black dots represent the sites of the chain $\mathbb Z_{20}$ in the same order, like the marked sites $x=-1,0,1,2$. 

When the state evolves via the quantum circuit $\ket{\Psi^{\rm fl}} \to T \ket{\Psi^{\rm fl}}$, the corresponding tensor network also evolves. 
To simplify the calculation of this evolution, it is convenient to separate each of the two layers of the time-translation operator as $T= T_\mathrm{odd} T_\mathrm{even}$. With this notation we can write the evolution of $\ket{\Psi^{\rm fl}}$ as
\begin{align}
  \nonumber
  T_\mathrm{even} \ket{\Psi^{\rm fl}} &=
\begin{tikzpicture}[baseline=-2]
  \draw[blue!30, very thick] (-2.5,2.5)--(2.5,2.5)--(2.5,-2.5)--(-2.5,-2.5)--(-2.5,2.5);
  \portaabar{-2}{-2} \portaabar{0}{-2} \portaabar{2}{-2}
  \portacabar{-1}{-1} \portacabar{1}{-1}
  \portaabar{-2}{0} \portaabar{0}{0} \portaabar{2}{0}
  \portacabar{-1}{1} \portacabar{1}{1}
  \portaabar{-2}{2} \portaabar{0}{2} \portaabar{2}{2}
  %  \cables{0}{0}\cables{1}{2.2}
  \portabar{-2}{3} \portabar{0}{3} \portabar{2}{3}
  \portacabar{-3}{1} \portacabar{-3}{-1}
  \portabar{-2}{-3} \portabar{0}{-3} \portabar{2}{-3}
  \portacabar{3}{1} \portacabar{3}{-1}
  \filldraw (-2.5,-2.5) circle (.05);
  \filldraw (-1.5,-2.5) circle (.05);
  \filldraw (-.5,-2.5) circle (.05);
  \filldraw (.5,-2.5) circle (.05);
  \filldraw (1.5,-2.5) circle (.05);
  \filldraw (2.5,-2.5) circle (.05);
  \filldraw (-2.5,2.5) circle (.05);
  \filldraw (-1.5,2.5) circle (.05);
  \filldraw (-.5,2.5) circle (.05);
  \filldraw (.5,2.5) circle (.05);
  \filldraw (1.5,2.5) circle (.05);
  \filldraw (2.5,2.5) circle (.05);
  \filldraw (-2.5,-1.5) circle (.05);
  \filldraw (-2.5,-.5) circle (.05);
  \filldraw (-2.5,.5) circle (.05);
  \filldraw (-2.5,1.5) circle (.05);
  \filldraw (2.5,-1.5) circle (.05);
  \filldraw (2.5,-.5) circle (.05);
  \filldraw (2.5,.5) circle (.05);
  \filldraw (2.5,1.5) circle (.05);
\end{tikzpicture}
  \\ \label{eq: Teven Psi} &=  
\begin{tikzpicture}[baseline=-2]
  \draw[blue!30, very thick] (-3.5,1.5)--(3.5,1.5)--(3.5,-1.5)--(-3.5,-1.5)--(-3.5,1.5);
  \portacabar{-1}{-1} \portacabar{1}{-1}
  \portaabar{-2}{0} \portaabar{0}{0} \portaabar{2}{0}
  \portacabar{-1}{1} \portacabar{1}{1}
  \portacabar{-3}{1} \portacabar{-3}{-1}
  \portacabar{3}{1} \portacabar{3}{-1}
  \filldraw (-2.5,-1.5) circle (.05);
  \filldraw (-1.5,-1.5) circle (.05);
  \filldraw (-.5,-1.5) circle (.05);
  \filldraw (.5,-1.5) circle (.05);
  \filldraw (1.5,-1.5) circle (.05);
  \filldraw (2.5,-1.5) circle (.05);
  \filldraw (-2.5,1.5) circle (.05);
  \filldraw (-1.5,1.5) circle (.05);
  \filldraw (-.5,1.5) circle (.05);
  \filldraw (.5,1.5) circle (.05);
  \filldraw (1.5,1.5) circle (.05);
  \filldraw (2.5,1.5) circle (.05);
  \filldraw (-3.5,-1.5) circle (.05);
  \filldraw (-3.5,-.5) circle (.05);
  \filldraw (-3.5,.5) circle (.05);
  \filldraw (-3.5,1.5) circle (.05);
  \filldraw (3.5,-1.5) circle (.05);
  \filldraw (3.5,-.5) circle (.05);
  \filldraw (3.5,.5) circle (.05);
  \filldraw (3.5,1.5) circle (.05);
\end{tikzpicture}  \ ,
\end{align}
and
\begin{align}
  \nonumber
  T_\mathrm{odd} T_\mathrm{even} \ket{\Psi^{\rm fl}} &=
\begin{tikzpicture}[baseline=-2]
  \draw[blue!30, very thick] (-3.5,1.5)--(3.5,1.5)--(3.5,-1.5)--(-3.5,-1.5)--(-3.5,1.5);
  \portacabar{-1}{-1} \portacabar{1}{-1}
  \portaabar{-2}{0} \portaabar{0}{0} \portaabar{2}{0}
  \portacabar{-1}{1} \portacabar{1}{1}
  \portacabar{-3}{1} \portacabar{-3}{-1}
  \portacabar{3}{1} \portacabar{3}{-1}
  \portaabar{-4}{0}\portaabar{4}{0}
  \portacbar{-3}{2}\portacbar{-1}{2}\portacbar{1}{2}\portacbar{3}{2}
  \portacbar{-3}{-2}\portacbar{-1}{-2}\portacbar{1}{-2}\portacbar{3}{-2}
  \filldraw (-2.5,-1.5) circle (.05);
  \filldraw (-1.5,-1.5) circle (.05);
  \filldraw (-.5,-1.5) circle (.05);
  \filldraw (.5,-1.5) circle (.05);
  \filldraw (1.5,-1.5) circle (.05);
  \filldraw (2.5,-1.5) circle (.05);
  \filldraw (-2.5,1.5) circle (.05);
  \filldraw (-1.5,1.5) circle (.05);
  \filldraw (-.5,1.5) circle (.05);
  \filldraw (.5,1.5) circle (.05);
  \filldraw (1.5,1.5) circle (.05);
  \filldraw (2.5,1.5) circle (.05);
  \filldraw (-3.5,-1.5) circle (.05);
  \filldraw (-3.5,-.5) circle (.05);
  \filldraw (-3.5,.5) circle (.05);
  \filldraw (-3.5,1.5) circle (.05);
  \filldraw (3.5,-1.5) circle (.05);
  \filldraw (3.5,-.5) circle (.05);
  \filldraw (3.5,.5) circle (.05);
  \filldraw (3.5,1.5) circle (.05);
\end{tikzpicture}  
  \\ \label{eq: Todd Psi} &=
\begin{tikzpicture}[baseline=-2]
  \draw[blue!30, very thick] (-4.5,.5)--(4.5,.5)--(4.5,-.5)--(-4.5,-.5)--(-4.5,.5);
  \portaabar{-2}{0} \portaabar{0}{0} \portaabar{2}{0}
  \portaabar{-4}{0}\portaabar{4}{0}
  \filldraw (-2.5,-.5) circle (.05);
  \filldraw (-1.5,-.5) circle (.05);
  \filldraw (-.5,-.5) circle (.05);
  \filldraw (.5,-.5) circle (.05);
  \filldraw (1.5,-.5) circle (.05);
  \filldraw (2.5,-.5) circle (.05);
  \filldraw (-2.5,.5) circle (.05);
  \filldraw (-1.5,.5) circle (.05);
  \filldraw (-.5,.5) circle (.05);
  \filldraw (.5,.5) circle (.05);
  \filldraw (1.5,.5) circle (.05);
  \filldraw (2.5,.5) circle (.05);
  \filldraw (-3.5,-.5) circle (.05);
  \filldraw (-4.5,-.5) circle (.05);
  \filldraw (-4.5,.5) circle (.05);
  \filldraw (-3.5,.5) circle (.05);
  \filldraw (3.5,-.5) circle (.05);
  \filldraw (4.5,-.5) circle (.05);
  \filldraw (4.5,.5) circle (.05);
  \filldraw (3.5,.5) circle (.05);
\end{tikzpicture}  \ ,
\end{align}
where here and in the rest of the paper we ignore the normalisation of states.
Note that the action of $T$ on the tensor network is not fully determined until we know the position of one site, like for example $x=0$.

The complete evolution of $\ket{\Psi^{\rm fl}}$ is depicted in Figure~\ref{flat_g_cycle}.
Remarkably, this evolution is cyclic and has a very small period ($\Delta t =5$) when compared with the approximate recurrence time of a typical state in $\H_{20}$, which would be doubly exponential in the size ($\sim\exp q^{20}$).
This (flat-space) tensor network can be generalised to any boundary size $n$ multiple of four, with corresponding period $\Delta t = \frac n 4$.
Of course, there are many other tensor-network states, and we analyse some of them below.

Consider a state $\ket\Phi$ and a positive integer $\Delta t$ satisfying $T^{\Delta t} \ket\Phi = \ket\Phi$. If $\Delta t$ is the smallest such integer, then they generate the orbit $T^t \ket\Phi$ for $t= 0,\ldots, \Delta t -1$. (An example of orbit is depicted in Figure~\ref{flat_g_cycle}.) Given any such orbit we can construct some eigenstates of $T$ as
\begin{align}\label{eq:eigen}
  &\ket{\Phi_\omega} = \sum_{t=0}^{\Delta t -1}
  e^{i \frac 1 {\Delta t} \omega t}\, T^t \ket\Phi\, ,
  \\
  &\ T \ket{\Phi_\omega} = 
  e^{i \frac 1 {\Delta t} \omega} \ket{\Phi_\omega}
\end{align}
for $\omega = 0,\ldots, \Delta t -1$.
Hence, it is meaningful to associate to this subspace a dynamical mode of (quasi) energy
\begin{align}\label{E delta t}
  E = \frac 1 {\Delta t}\ .
\end{align}

\subsection{Entanglement geometry}

The above tensor-network states can be interpreted as 2D spatial geometries, where the metric distance is fixed the entanglement structure of the state, via Ryu-Takayanagi's prescription \cite{Ryu_2006, Ryu_2007, Hubeny, Lashkari_2014}. This identifies the von Neumann entropy of a set of consecutive sites in the chain with the length of the shortest curve beginning and ending at the boundary of the set (see Figure~\ref{fig:R-T}).
Note that we can apply this prescription to spaces other than AdS.

\begin{figure}
  \centering  
  \includegraphics[width=85mm]{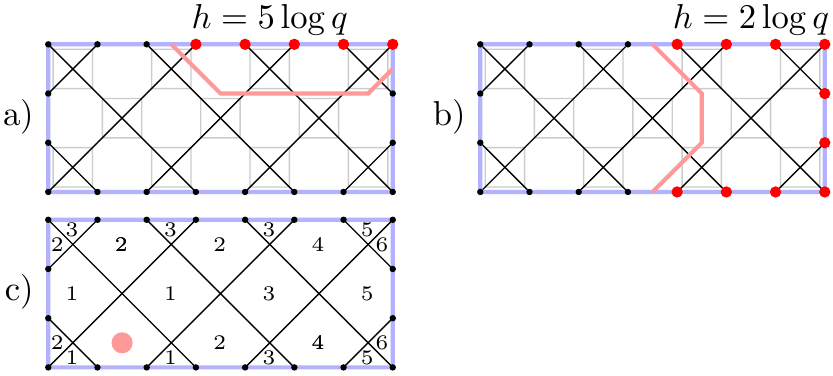}
  \vspace{-2mm}
  \caption{\textbf{Entanglement and geodesics in flat space} (tensor network from Figure~\ref{flat_g_cycle} at $t=.5$). Figures a) b): the entanglement entropy $h$ between the region consisting of red dots and the rest of the chain is equal to the smallest number of black lines (times $\log q$) that are crossed by a curve enclosing the red dots. In very light grey, the position of the gates.
  Figure c): distance between each location and that with the pink dot.
Each white rhombus represents a location in the bulk, which in this case has flat curvature. Locations in the boundary of the bulk correspond to links in the chain $\mathbb Z_{20}$.}
\label{fig:R-T}
\end{figure}

Let us obtain the metric distance of the simplest tensor-network state
\begin{align}
 \ket{\Psi^{\rm fl}} = \frac 1 q 
\begin{tikzpicture}[baseline=-2]
  \draw[blue!30, very thick] (-.5,-.5)--(-.5,.5)--(.5,.5)--(.5,-.5)--(-.5,-.5);
  \portabar{0}{0}
  \filldraw (.5,.5) circle (.05);
  \filldraw (-.5,.5) circle (.05);
  \filldraw (-.5,-.5) circle (.05);
  \filldraw (.5,-.5) circle (.05);
  \draw (-.8,.8) node {$_0$};
  \draw (.8,.8) node {$_1$};
  \draw (.8,-.8) node {$_2$};
  \draw (-.8,-.8) node {$_3$};
\end{tikzpicture}\ .
\end{align}
The unitary condition \eqref{eq:unitarity} implies that the reduced density matrix of subsystem $\{0,1\}$ is proportional to the identity matrix, which implies that $\ket{\Psi^{\rm fl}}$ is maximally entangled with respect to the bipartition $01|23$.
Similarly, the dual-unitary condition \eqref{eq:dual unitarity} implies that $\ket{\Psi^{\rm fl}}$ is maximally entangled with respect to the bipartition $03|12$. This in turn implies that each site (0,1,2 or 3) is maximally entangled with the rest.
This precise entanglement structure is contained in the geometry
\begin{tikzpicture}[baseline=-3]
  \draw[blue!30, very thick] (-.5,-.5)--(-.5,.5)--(.5,.5)--(.5,-.5)--(-.5,-.5);
  \xg{0}{0}
  \filldraw (.5,.5) circle (.05);
  \filldraw (-.5,.5) circle (.05);
  \filldraw (-.5,-.5) circle (.05);
  \filldraw (.5,-.5) circle (.05);
\end{tikzpicture}\,,
where each of the 4 triangles represents a location in the surface, and the distance between two locations is given by the number of black lines that are crossed when travelling from one location to the other. 

In the previous example, the emergent geometry does not contain interior points. Figure~\ref{fig:R-T} depicts the geometry of a more complex tensor-network state, which has a non-trivial interior. We have checked the equality between entropy and distance for all bipartitions, but we have not proven that this geometry uniquely captures the entanglement structure of the state. But in any case, this geometry is special, because it has the same structure as the underlying tensor network.
%Note that the boundary of the bulk is dual to the chain $\mathbb Z_n$, in the sense that the sites of one correspond to the links of the other.
Now, we can use the distance defined in Figure~\ref{fig:R-T} to calculate the length of any curve in the bulk, not necessarily the shortest one connecting a pair boundary points.
This reveals that the geometry in question is a piece flat space with boundary.

In the continuum setup (AdS/CFT) there is an extensive literature \cite{Czech_2012, Wall_2014, Headrick_2014, Esp_ndola_2018, Bao_2015, Aaronson_2022} addressing the problem of how to obtain the bulk geometry given the entanglement structure of the CFT state. 
In the following subsections we discuss the geometry of relevant states in $\H_n$, which generate discrete versions of AdS space with and without a black hole, and the double-sided AdS black hole.

%\begin{quotation}\noindent\em 
%The length of a curve within a tensor network is equal to the number of black lines crossed by the curve.  
%\end{quotation}
%Note that this prescription cannot be applied to general superpositions of tensor-network states, where one may apply entanglement reconstruction methods 

\subsection{Anti-de Sitter state}

The dilation isometry $D_n : \H_n \to \H_{2n}$ maps a chain of length $n$ (multiple of 4) to a chain of length $2n$. Its particular form is
\begin{align}\label{def:D}
  D_n = 
  q^{-n/8} \Big( \cdots
\begin{tikzpicture}[baseline=-4]
  \draw (-3.5,-.5)--(-3.5,.7);
  \draw (-2.5,-.5)--(-2.5,.7);
  \antibarra{-1}{0}\portacbar{0}{0}\barra{1}{0}
  \cables{-1}{1.2}\cables{1}{1.2}
  \draw (2.5,-.5)--(2.5,.7);
  \draw (3.5,-.5)--(3.5,.7);
  \draw (4.5,-.5)--(4.5,.7);
  \draw (5.5,-.5)--(5.5,.7);
  \antibarra{7}{0}\portacbar{8}{0}\barra{9}{0}
  \cables{7}{1.2}\cables{9}{1.2}
%  \draw (10.5,.7)--(10.5,-.5);
  %
  \draw (-3.7,1) node {\small $_{-2}$};
  \draw (-2.7,1) node {\small $_{-1}$};
  \draw (-1.5,1) node {\small $_{0}$};
  \draw (-.5,1) node {\small $_1$};
  \draw (.5,1) node {\small $_2$};
  \draw (1.5,1) node {\small $_3$};
  \draw (2.5,1) node {\small $_4$};
  \draw (3.5,1) node {\small $_5$};
  \draw (4.5,1) node {\small $_6$};
  \draw (5.5,1) node {\small $_7$};
  \draw (6.5,1) node {\small $_8$};
  \draw (7.5,1) node {\small $_9$};
  \draw (8.5,1) node {\small $_{10}$};
  \draw (9.5,1) node {\small $_{11}$};
  \draw (-3.5,-.9) node {\small $_0$};
  \draw (-2.5,-.9) node {\small $_1$};
  \draw (2.5,-.9) node {\small $_2$};
  \draw (3.5,-.9) node {\small $_3$};
  \draw (4.5,-.9) node {\small $_4$};
  \draw (5.5,-.9) node {\small $_5$};
\end{tikzpicture}
  \cdots \Big) ,
\end{align}
where we have included the site labels $x$ of the input and output chains.
Note that the dilation $D$ is the Hermitian conjugate of the contraction $C$ defined in \eqref{def:C}, except that $C$ is defined for the four-layer circuit \eqref{eq:T gauge} and here we define $D$ for the two-layer circuit \eqref{eq:2 layer}.

\begin{figure*}
  \centering
  \includegraphics[width=3.7cm]{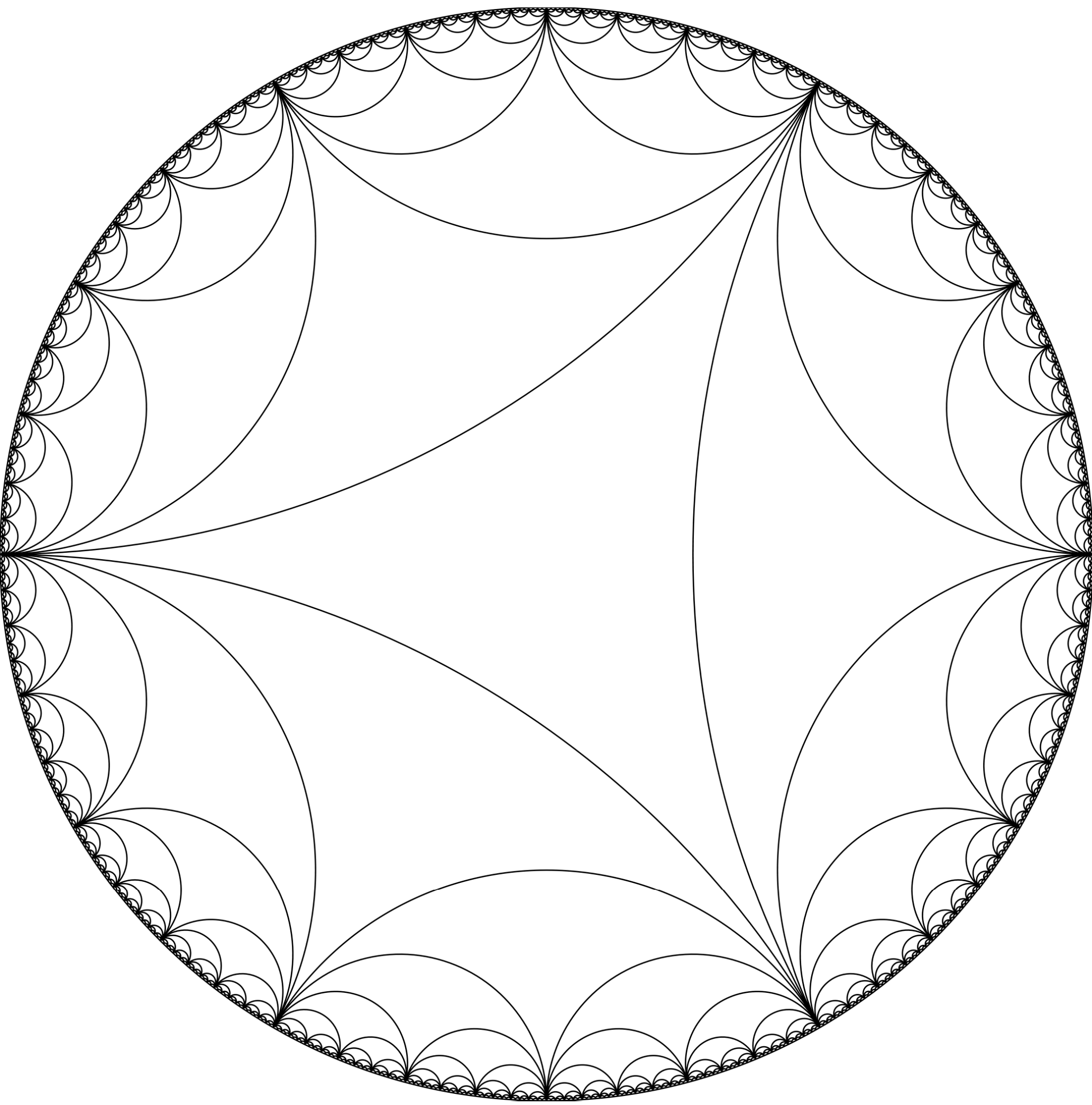}
  \hspace{2mm}
  \includegraphics[width=13.8cm]{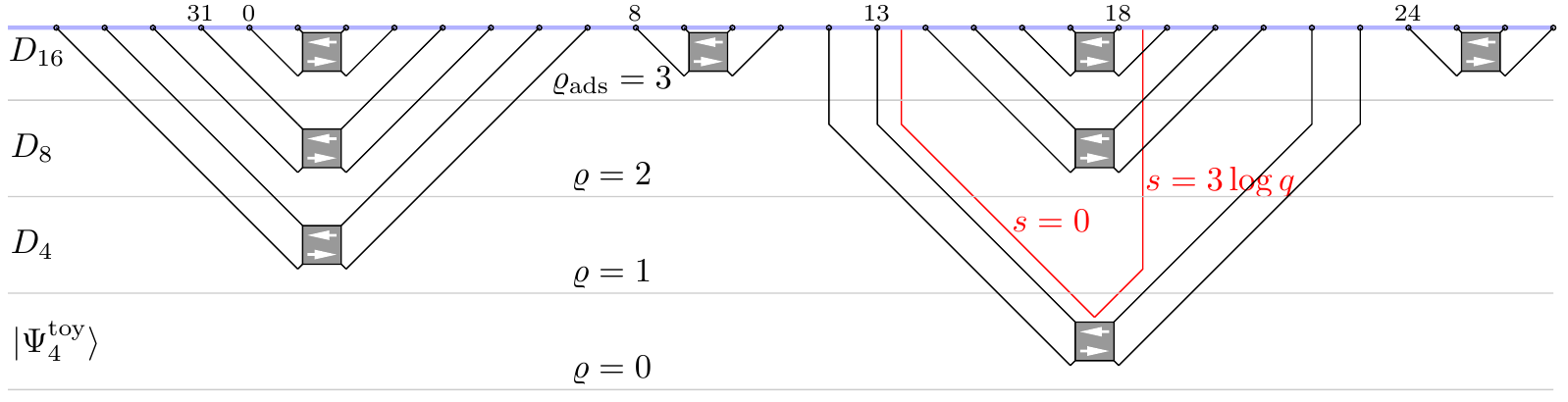}
  \caption{\textbf{Toy AdS state.} The left figure shows (continuum) AdS tiled with equal-size triangles. The right figure displays the the tensor network of the toy AdS state for a chain of size $n=32$ with periodic boundary conditions. Each level of the four recursions is identified by the corresponding value of the radial coordinate $\varrho$. In red, the minimal length paths from boundary locations $x=13,18$ to the centre, and their length $s$. Which shows that this ``toy" version of AdS is very non-smooth.}
  \label{fig:pure_AdS}
\end{figure*}

The toy AdS state is recursively defined as
\begin{align}\label{def:toy 4}
  \ket{\Psi_4^{\rm toy}} 
  &= \frac 1 q 
\begin{tikzpicture}[baseline=-2]
  \draw[blue!30, very thick] (-.5,-.5)--(-.5,.5)--(.5,.5)--(.5,-.5)--(-.5,-.5);
  \portaabar{0}{0}
  \filldraw (.5,.5) circle (.05);
  \filldraw (-.5,.5) circle (.05);
  \filldraw (-.5,-.5) circle (.05);
  \filldraw (.5,-.5) circle (.05);
  \draw (-.7,.8) node {$_0$};
  \draw (.7,.8) node {$_1$};
  \draw (.7,-.8) node {$_2$};
  \draw (-.7,-.8) node {$_3$};
\end{tikzpicture}\ ,
  \\
  \ket{\Psi_{2n}^{\mathrm{toy}}}
  &=
  D_{n}  
  \ket{\Psi_n ^{\mathrm{toy}}}\ ,
\end{align}
which results in
\begin{align}\label{eq: p ads s}
  \ket{\Psi_n ^{\mathrm{toy}}} = 
  \underbrace{D_{\frac n 2} \cdots D_{16} D_8 D_4}_{\varrho_{\rm ads}} 
  \ket{\Psi_4 ^{\mathrm{toy}}}\ .
\end{align} 
Note that the size of the chain is $n=2^{\varrho_{\rm ads}}4$, where the positive integer $\varrho_{\rm ads}$ denotes the number of recursions. 
This tensor network is depicted in Figure~\ref{fig:pure_AdS}, where we can see that each recursion corresponds to a value of the radial coordinate $\varrho=0,1,\ldots, \varrho_{\rm ads}$.
Also, the intermediate state $\ket{\Psi_{2^\varrho 4} ^{\mathrm{toy}}}$ in the recursion \eqref{eq: p ads s} represents a spatial slice of AdS with the radial coordinate restricted to the interval $[0, \varrho]$.

By proceeding as in \eqref{eq: Teven Psi} and \eqref{eq: Todd Psi} we can check that the 4-site toy AdS state \eqref{def:toy 4} is an eigenstate of the evolution operator 
\begin{align}\label{eq:T4 eigen}
  T_4 \ket{\Psi_4^{\rm toy}} = \ket{\Psi_4^{\rm toy}}\, .
\end{align}
Also, by proceeding as in \eqref{CT=TC} we obtain the equality
\begin{align}\label{eq TD=DT}
  T^4_{2n} D_n = D_n T^2_{n} \, ,
\end{align}
which implies that, when the radial coordinate $\varrho$ decreases by one unit, time slows down by a factor of two:
\begin{align}\label{eq: T4 Psi}
  T_{2n}^4 \ket{\Psi_{2n}^{\mathrm{toy}}} 
  = 
  D_{n} T_{n}^2 \ket{\Psi_{n}^{\mathrm{toy}}}\, .
\end{align}
In Section~\ref{sec:matter} we show that the  relationship \eqref{eq: T4 Psi} between time at different radial locations applies to local clocks made of matter. 
%Therefore, the relation between time $t$ at the boundary (i.e.~chain time) and time $t_{\varrho=1}$ at the centre $\varrho=1$ is $t= 2^{\varrho_{\rm ads}-1} t_{\varrho=1}$. In order to mimic the AdS distance \eqref{AdS distance}, we define the time variable $\tau$ in \eqref{av dist} as a re-scaling of time at the centre $\tau = \log q\, t_{\varrho=1}$.

\begin{figure*}
  \centering
  \includegraphics[width=18cm]{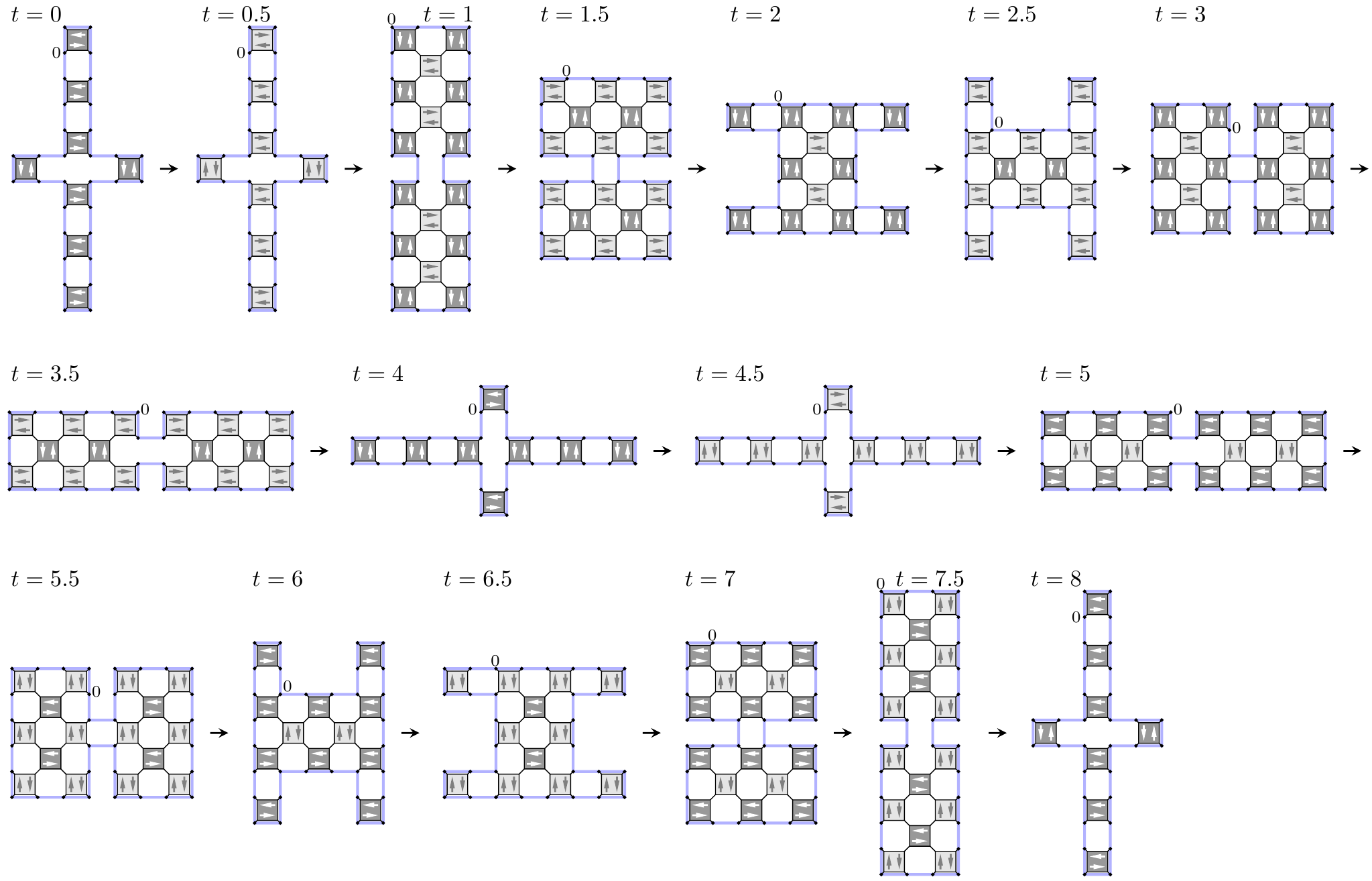}
  %\vspace{-5mm}
  \caption{\textbf{Dynamics of the toy AdS state} $\ket{\Psi_{32}^{\mathrm{toy}}}$ represented in Figure~\ref{fig:pure_AdS}. The position of the chain location $x=0$ is represented at each time step.}
  \label{hp_cycle_32}
\end{figure*}

The state $\ket{\Psi_{n}^{\mathrm{toy}}}$ is not an eigenstate of $T_n$, it evolves in time through the orbit shown in Figure~\ref{hp_cycle_32}. 
The corresponding period can be calculated by using equations \eqref{eq:T4 eigen} and \eqref{eq TD=DT}, obtaining 
\begin{align}
  (T_n) ^{2^{\varrho_{\rm ads}+1}} \ket{\Psi_n^{\mathrm{toy}}} 
  =
  \ket{\Psi_n^{\mathrm{toy}}}\ .
\end{align}
That is, the period of this orbit is proportional to the length of the chain $\Delta t = 2^{\varrho_{\rm ads}+1} = \frac n 2$.  
Again, note that this is much shorter than the recurrence time of a typical state in $\H_n$, which would be doubly exponential in the size ($\sim\exp q^n$).
Using \eqref{E delta t} we associate to empty AdS an energy
\begin{align}\label{E ads}
  E_{\rm ads} = \frac 2 n \ .
\end{align}

The geometry of the tensor network of $\ket{\Psi_n ^{\mathrm{toy}}}$ is not regular along the radial direction. For instance, the geodesic distance from a boundary point $x$ to the centre ($\varrho=0$) depends on $x$, as shown in the red lines of Figure~\ref{fig:pure_AdS}. 
Interestingly, the $T$-eigenstate 
\begin{align}\label{eigen T}
  \ket{\Psi_n ^{\mathrm{ads}}}
  = %\sqrt{\frac 2 n}\,
  \sum_{t=0}^{\frac n 2 -1} T^t 
  \ket{\Psi_n ^{\mathrm{toy}}}\ ,
\end{align}
has a more regular geometry in the large-$q$ limit.
Specifically, it produces the metric distance 
\begin{align}\label{av dist}
  \Delta s^2_{\Psi^{\rm ads}} = \log^2\! q \left( 
  -2^{2\varrho} \Delta\tau^2 
  +\Delta\varrho^2
  +2^{2\varrho} \frac {\Delta\theta^2} {\pi^2} \right),
\end{align}
to leading order in $q$. 
(The proof of this fact will be presented elsewhere.)
Recall that this distance characterises a discrete geometry, hence, the increments $\Delta\tau, \Delta\varrho, \Delta\theta$ are discrete.
The radial coordinate $\varrho\in [0,\varrho_{\rm ads}] \subset \mathbb Z$ has already been introduced. 
The time coordinate $\tau$ at the boundary is related to the QCA time via $t= \log q\, 2^{\varrho_{\rm ads}} \tau$. Note that the proper time defined by \eqref{av dist} reflects the radial dependence implied by \eqref{eq: T4 Psi}. 
The angular coordinate $\theta \in [0,2\pi]$ is discretised in $2^\varrho 4$ steps of size $\Delta \theta = 2\pi 2^{-\varrho-2}$. Hence, each step has one unit of proper distance and the proper distance of a complete circle is $\log q\, 2^{\varrho}4$, the logarithm of the dimension of the chain in the corresponding intermediate recurrence step.
Finally, note that the distance \eqref{av dist} strongly resembles AdS's distance in global coordinates
\begin{align}\label{AdS distance}
  ds^2_{\rm AdS} = \alpha^2 \left( 
  -\cosh^2\! \varrho\, d\tau^2 
  +d\varrho^2
  +\sinh^2\!\varrho\, d\theta^2 \right)\ .
\end{align}

\subsection{Thermofield double state}

In the previous subsection we saw that eigenstates produce smoother geometries. However, eigenstates require superpositions of tensor networks, and hence, are less easy to visualise. For this reason, in this subsection, we continue using non-eigenstate tensor-network states.

\begin{figure}
  \centering
  \includegraphics[width=88mm]{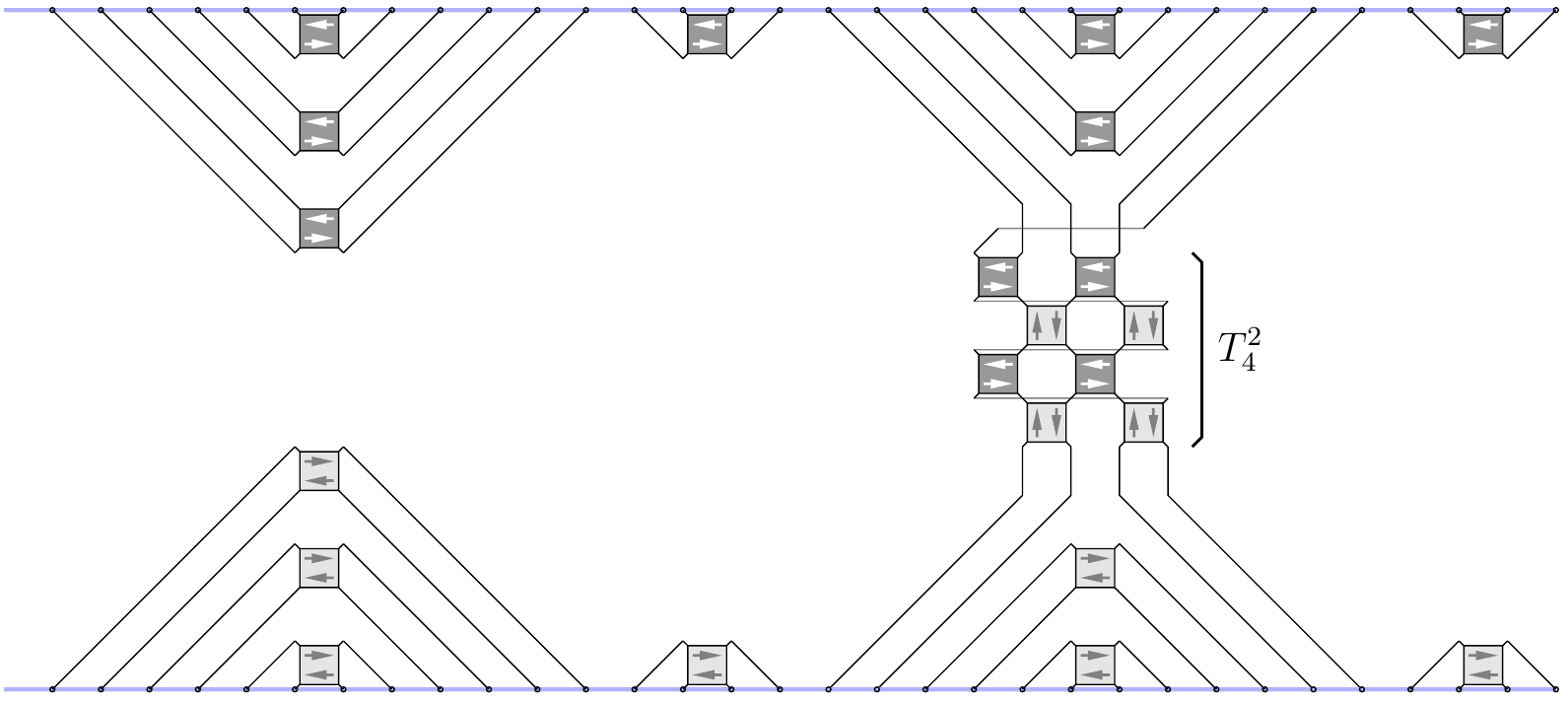}
  \caption{\textbf{Thermofield double state.} This tensor network represents a toy version of a two sided black hole in AdS with boundary radius $\varrho_{\rm ads} =3$ and horizon radius $\varrho_{\rm h} =0$, which implies boundary size $n=32$ and horizon size $a=4$. Periodic boundary conditions are understood in the two chains and the throat (the piece connecting the two symmetric sides). The throat has been growing for 2 local time steps, which implies that the QCA has been evolving for $t= 2\times 2^3 = 16$ time steps.}
  \label{fig:wormhole_AdS}
\end{figure}

The thermofield double (TFD) is a joint state of two identical chains $\H_n \otimes \H_n$ evolving in time via the dynamics $T\otimes T^\intercal$, where $T^\intercal$ is the transpose of $T$ (see Figure~\ref{fig:wormhole_AdS}). This state is characterised by the largest and smallest values of the radial coordinate outside the throat $\varrho\in [\varrho_{\rm h},\varrho_{\rm ads}] \subseteq \mathbb Z$.
These parameters fix the size of each chain to $n= 2^{\varrho_{\rm ads}} 4$, and the area (length) of the horizon to $a \log q$, where we define $a= 2^{\varrho_{\rm h}} 4$. 
The TFD with $n=a$ is the maximally-entangled state between the two chains
\begin{align}
  &\ket{\Psi_{a,a}^{\mathrm{tfd}}}
  = \bigotimes_{x\in \mathbb Z _a} \ket{\psi}_x\ ,
  \\
  &\ket{\psi}_x 
  = \frac 1 {\sqrt q} \sum_{k=1}^q \ket k_x \otimes \ket k_x\ ,
\end{align}
where $\ket{\psi}_x$ is the maximally-entangled state between site $x$ of one chain and site $x$ of the other chain.
When $n>a$ the TFD can be recursively generated via
\begin{align}
  \ket{\Psi_{2n,a}^{\mathrm{tfd}}}
  &=
  D_n \otimes D_n^* 
  \ket{\Psi_{n,a}^{\mathrm{tfd}}}\ ,
\end{align}
where $D_n^*$ is the complex conjugate of $D_n$, defined in \eqref{def:D}.
The TFD can be explicitly written as
\begin{align}%\label{eq: p ads s}
  \ket{\Psi_{n,s}^{\mathrm{tfd}}}
  = 
  (\underbrace{D_{\frac n 2} \cdots D_{2a} D_a}_{\varrho_{\rm ads}-\varrho_{\rm h}}) 
  \otimes
  (\underbrace{D_{\frac n 2} \cdots D_{2a} D_a}_{\varrho_{\rm ads}-\varrho_{\rm h}})^* 
  \ket{\Psi_{a,a}^{\mathrm{tfd}}}\ .
\end{align} 
%As in \eqref{rnd r_rho} we can perform random angular shifts at each value of the radial coordinate $\varrho$ in order to obtain a smoother geometry.

Next, let us analyse the dynamics of the TFD. Proceeding as in \eqref{CT=TC} we obtain the identity
\begin{align}\label{eq DT=TD}
  (T^\intercal _{2n})^4 D_n^* =   
  D_n^* (T^\intercal _n)^2 \ ,
\end{align}
which is not equivalent to \eqref{eq TD=DT}, although here it produces a similar result: when the radial coordinate $\varrho$ is decreased by one, time slows down by a factor two
\begin{align}
  \nonumber
  (T_{2n} \otimes T^\intercal _{2n})^4 \ket{\Psi_{2n,a}^{\mathrm{tfd}}}
  =
  (D_n \otimes D_n^*) (T_n \otimes T^\intercal _{n})^2
  \ket{\Psi_{n,a}^{\mathrm{tfd}}}\, .
\end{align}
Now we recover the standard fact that the throat wormhole grows linearly in time, since the action of $T \otimes T^\intercal$ cannot be simplified
\begin{align}
  T_a \otimes T^\intercal _a \ket{\Psi_{a,a}^{\mathrm{tfd}}}
  &=
  T_a^2 \otimes \unity \ket{\Psi_{a,a}^{\mathrm{tfd}}}\, ,
\end{align}
as also illustrated in Figure~\ref{fig:wormhole_AdS}.
Note however, that the propper time in the throat is slower than that on the boundary by an exponential factor $2^{(\varrho_{\rm ads} -\varrho_{\rm h})}$.

\subsection{Thermal anti-de Sitter state}

\begin{figure}\begin{center}
  \includegraphics[width=88mm]{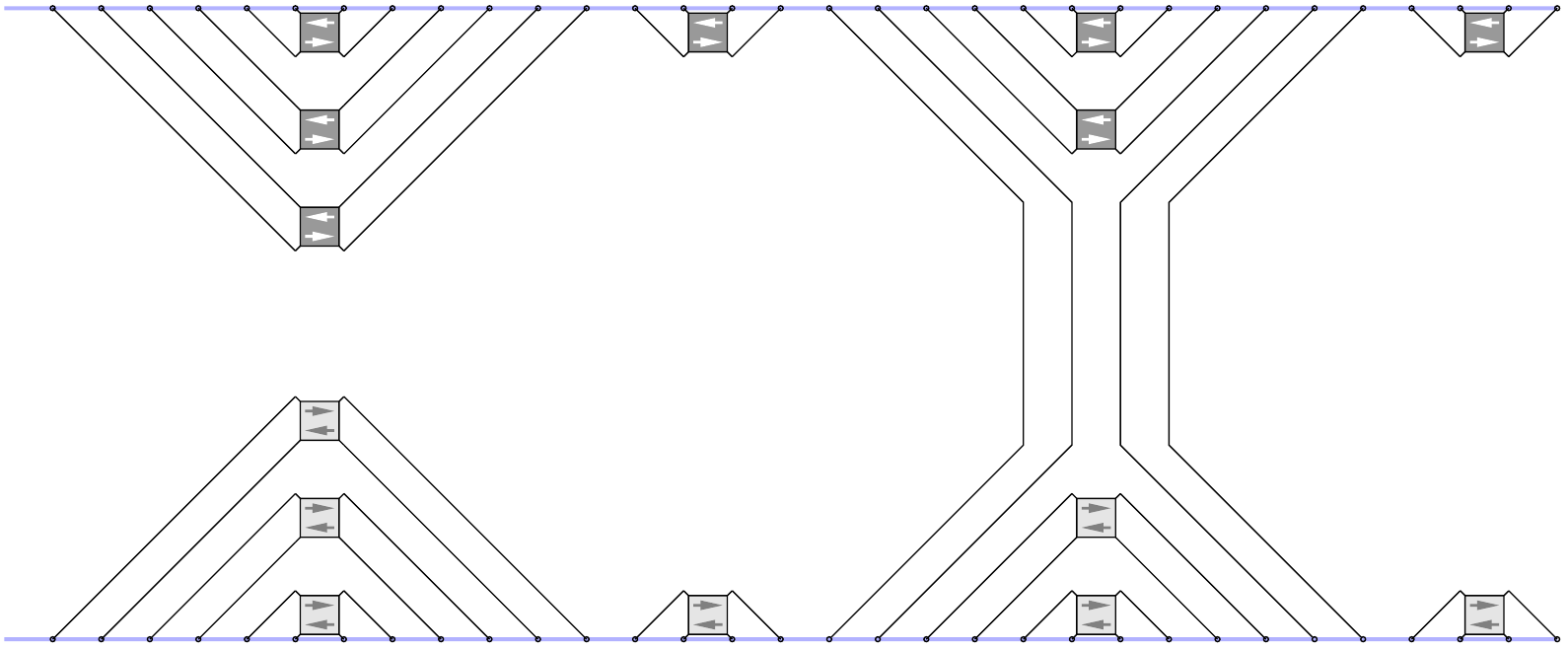}
  \caption{\textbf{Thermal AdS state.} This tensor network represents the toy version of  a simple black hole in AdS, for a chain of size $n=32$ and horizon area $a=4$. The two blue lines correspond to the input and output of the density matrix. 
  Periodic boundary conditions are understood in the chain and the horizon.}
  \label{fig:thermal_AdS}
\end{center}\end{figure}

If we perform the partial trace on the right chain in the TFD then we obtain the mixed state
\begin{align}
  \nonumber
  \rho_{n,a}%^{\mathrm{ads}}
  &= {\rm tr}_{\rm right} |\Psi_{n,a}^{\mathrm{tfd}} 
  \rangle\!\langle \Psi_{n,a}^{\mathrm{tfd}}|
  \\ &=
  (\underbrace{D_{\frac n 2} \cdots D_{2a} D_a}_{\varrho_{\rm ads}-\varrho_{\rm h}}) 
  \unity_a 
  (\underbrace{D_{\frac n 2} \cdots D_{2a} D_a}_{\varrho_{\rm ads}-\varrho_{\rm h}}) ^\dagger\ ,
\end{align}
where $\unity_a$ is the identity acting on $\H_a$ (see Figure~\ref{fig:thermal_AdS}).
The fact that $D_n$ are isometries implies that the entropy of $\rho_{n,a}$ is $a \log q$.
When $n=a$ the state $\rho_{a,a}^{\mathrm{ads}}$ is time-independent and maximally mixed, which corresponds to infinite temperature. 

Like in the previous variants of AdS, here we also have the following property. When $n>a$, when the radial coordinate $\varrho$ decreases by one unit, time slows down by a factor two
\begin{align}
  (T_{2n})^4 \rho_{2n,a} (T_{2n}^\dagger)^4
  = 
  D_n\, (T_{n})^2 \rho_{n,a} (T_{n}^\dagger)^2 D_n^\dagger\ .
\end{align}
This implies that the evolution of the thermal AdS state undergoes a cycle of period $\Delta t = 2^{\varrho_{\rm ads}-\varrho_{\rm h} +1} = \frac {2n} a$. Hence, we associate to it an energy
\begin{align}
  \frac E {E_{\rm ads}} = \frac a 4  ,
\end{align}
where we have substituted the energy of pure AdS $E_{\rm ads}$ obtained in \eqref{E ads}.
At this stage, it is not clear how to interpret this identity.

\section{Spaces with matter}
\label{sec:matter}

In the previous section we considered tensor-network states \eqref{eq: p ads s} constructed with the building blocks $\tikz[baseline]{\portabar{0}{.3}}$ and $\tikz[baseline]{\portacbar{0}{.3}}$, and interpreted them as curved empty spaces for the bulk. 
In this section we consider bulk spaces with matter.
We say that a state contains matter (when interpreted as a bulk state) when it cannot be written as a tensor network, or a super-position thereof, constructed with the building blocks $\tikz[baseline]{\portabar{0}{.3}}$ and $\tikz[baseline]{\portacbar{0}{.3}}$.
The states with matter that we analyse are empty spaces with the addition of some (non-building blocks) operators on a small number of links. 
When only one link is affected we interpret it as the position of a particle.

\subsection{Ambiguity in the position of particles}

Let us see that, if the dual unitary $\tikz[baseline]{\portabar{0}{.3}}$ is chaotic (has no solitons), then there are particle states with a well-defined position for the particle.
Let us start by considering the (empty) flat space state
\begin{align}\label{eq:empty space}
  \ket{\Psi^{\rm fl}} = 
\begin{tikzpicture}[baseline=-2]
  \draw[blue!30, very thick] (-2.5,2.5)--(2.5,2.5)--(2.5,-2.5)--(-2.5,-2.5)--(-2.5,2.5);
  \portaabar{-2}{-2} \portaabar{0}{-2} \portaabar{2}{-2}
  \portacabar{-1}{-1} \portacabar{1}{-1}
  \portaabar{-2}{0} \portaabar{0}{0} \portaabar{2}{0}
  \portacabar{-1}{1} \portacabar{1}{1}
  \portaabar{-2}{2} \portaabar{0}{2} \portaabar{2}{2}
  %\cables{0}{0}\cables{1}{2.2}
  \filldraw (-2.5,-2.5) circle (.05);
  \filldraw (-1.5,-2.5) circle (.05);
  \filldraw (-.5,-2.5) circle (.05);
  \filldraw (.5,-2.5) circle (.05);
  \filldraw (1.5,-2.5) circle (.05);
  \filldraw (2.5,-2.5) circle (.05);
  \filldraw (-2.5,2.5) circle (.05);
  \filldraw (-1.5,2.5) circle (.05);
  \filldraw (-.5,2.5) circle (.05);
  \filldraw (.5,2.5) circle (.05);
  \filldraw (1.5,2.5) circle (.05);
  \filldraw (2.5,2.5) circle (.05);
  \filldraw (-2.5,-1.5) circle (.05);
  \filldraw (-2.5,-.5) circle (.05);
  \filldraw (-2.5,.5) circle (.05);
  \filldraw (-2.5,1.5) circle (.05);
  \filldraw (2.5,-1.5) circle (.05);
  \filldraw (2.5,-.5) circle (.05);
  \filldraw (2.5,.5) circle (.05);
  \filldraw (2.5,1.5) circle (.05);
  \draw (-.5,2.9) node {$_0$};
\end{tikzpicture}\  ,
\end{align}
and apply an arbitrary operator $\tikz{\filldraw[red] (0,.5) circle (.1);} = \a_0\in \A_0$ at the chain site $x=0$ 
\begin{align}\label{eq:a psi}
  \a_0 \ket{\Psi^{\rm fl}} =
\begin{tikzpicture}[baseline=-2]
  \draw[blue!30, very thick] (-2.5,2.5)--(2.5,2.5)--(2.5,-2.5)--(-2.5,-2.5)--(-2.5,2.5);
  \portaabar{-2}{-2} \portaabar{0}{-2} \portaabar{2}{-2}
  \portacabar{-1}{-1} \portacabar{1}{-1}
  \portaabar{-2}{0} \portaabar{0}{0} \portaabar{2}{0}
  \portacabar{-1}{1} \portacabar{1}{1}
  \portaabar{-2}{2} \portaabar{0}{2} \portaabar{2}{2}
  %\cables{0}{0}\cables{1}{2.2}
  \filldraw (-2.5,-2.5) circle (.05);
  \filldraw (-1.5,-2.5) circle (.05);
  \filldraw (-.5,-2.5) circle (.05);
  \filldraw (.5,-2.5) circle (.05);
  \filldraw (1.5,-2.5) circle (.05);
  \filldraw (2.5,-2.5) circle (.05);
  \filldraw (-2.5,2.5) circle (.05);
  \filldraw (-1.5,2.5) circle (.05);
  \filldraw (-.5,2.5) circle (.05);
  \filldraw (.5,2.5) circle (.05);
  \filldraw (1.5,2.5) circle (.05);
  \filldraw (2.5,2.5) circle (.05);
  \filldraw (-2.5,-1.5) circle (.05);
  \filldraw (-2.5,-.5) circle (.05);
  \filldraw (-2.5,.5) circle (.05);
  \filldraw (-2.5,1.5) circle (.05);
  \filldraw (2.5,-1.5) circle (.05);
  \filldraw (2.5,-.5) circle (.05);
  \filldraw (2.5,.5) circle (.05);
  \filldraw (2.5,1.5) circle (.05);
  %
  %\draw (-.5,2.9) node {$_0$};
  \filldraw[red] (-.5,2.5) circle (.1);
\end{tikzpicture}\  .
\end{align}
We interpret this state as the empty space \eqref{eq:empty space} with one particle at its boundary. 
However, the same sate can also be written as 
\begin{align}\label{eq:a psi 2}
  \a_0 \ket{\Psi^{\rm fl}} =
\begin{tikzpicture}[baseline=-2]
  \draw[blue!30, very thick] (-2.5,2.5)--(2.5,2.5)--(2.5,-2.5)--(-2.5,-2.5)--(-2.5,2.5);
  \portaabar{-2}{-2} \portaabar{0}{-2} \portaabar{2}{-2}
  \portacabar{-1}{-1} \portacabar{1}{-1}
  \portaabar{-2}{0} \portaabar{0}{0} \portaabar{2}{0}
  \portacabar{-1}{1} \portacabar{1}{1}
  \portaabar{-2}{2} \portaabar{0}{2} \portaabar{2}{2}
  %\cables{0}{0}\cables{1}{2.2}
  \filldraw (-2.5,-2.5) circle (.05);
  \filldraw (-1.5,-2.5) circle (.05);
  \filldraw (-.5,-2.5) circle (.05);
  \filldraw (.5,-2.5) circle (.05);
  \filldraw (1.5,-2.5) circle (.05);
  \filldraw (2.5,-2.5) circle (.05);
  \filldraw (-2.5,2.5) circle (.05);
  \filldraw (-1.5,2.5) circle (.05);
  \filldraw (-.5,2.5) circle (.05);
  \filldraw (.5,2.5) circle (.05);
  \filldraw (1.5,2.5) circle (.05);
  \filldraw (2.5,2.5) circle (.05);
  \filldraw (-2.5,-1.5) circle (.05);
  \filldraw (-2.5,-.5) circle (.05);
  \filldraw (-2.5,.5) circle (.05);
  \filldraw (-2.5,1.5) circle (.05);
  \filldraw (2.5,-1.5) circle (.05);
  \filldraw (2.5,-.5) circle (.05);
  \filldraw (2.5,.5) circle (.05);
  \filldraw (2.5,1.5) circle (.05);
  %
  %\draw (-.5,2.9) node {$_0$};
  \filldraw[red] (-.5,1.5) circle (.1);
  \filldraw[red] (.5,1.5) circle (.1);
  \draw[red, thick] (-.5,1.5)--(.5,1.5);
\end{tikzpicture}\ ,
\end{align}
where the two-site operator \tikz{\filldraw[red] (-.5,.5) circle (.1);\filldraw[red] (.5,.5) circle (.1);\draw[red, thick] (-.5,.5)--(.5,.5);} is defined by
\begin{align}\label{eq:op 2}
\begin{tikzpicture}[baseline=-2]
  \filldraw[red] (-.5,0) circle (.1);
  \filldraw[red] (.5,0) circle (.1);
  \draw[red, thick] (-.5,0)--(.5,0);
\end{tikzpicture}
=
\begin{tikzpicture}[baseline=-2]
  \portaabar{0}{-.5} \portabar{0}{.5} 
  \filldraw[red] (-.5,0) circle (.1);
\end{tikzpicture}\ ,
\end{align}
which then satisfies
\begin{align}%\label{eq:op 4}
\begin{tikzpicture}[baseline=-2]
  \portaabar{0}{0}
  \filldraw[red] (-.5,-.5) circle (.1);
  \filldraw[red] (.5,-.5) circle (.1);
  \draw[red, thick] (-.5,-.5)--(.5,-.5);
\end{tikzpicture}
=
\begin{tikzpicture}[baseline=-2]
  \portaabar{0}{0}  
  \filldraw[red] (-.5,.5) circle (.1);
\end{tikzpicture}\ ,
\end{align}
and implies the equality between \eqref{eq:a psi} and \eqref{eq:a psi 2}.
Analogously, the operator 
\begin{align}\label{eq:op 4}
\begin{tikzpicture}[baseline=-2]
  \filldraw[red] (-1.5,0) circle (.1);
  \filldraw[red] (-.5,0) circle (.1);
  \filldraw[red] (.5,0) circle (.1);
  \filldraw[red] (1.5,0) circle (.1);
  \draw[red, thick] (-1.5,0)--(1.5,0);
\end{tikzpicture}
=
\begin{tikzpicture}[baseline=-2]
  \portacabar{-1}{-.5} \portacabar{1}{-.5} 
  \portacbar{-1}{.5} \portacbar{1}{.5} 
  \filldraw[red] (-.5,0) circle (.1);
  \filldraw[red] (.5,0) circle (.1);
  \draw[red, thick] (-.5,0)--(.5,0);
\end{tikzpicture}\ ,
\end{align}
allows to write the same state \eqref{eq:a psi} as
\begin{align}\label{eq:a psi 4}
  \a_0 \ket{\Psi^{\rm fl}} =
\begin{tikzpicture}[baseline=-2]
  \draw[blue!30, very thick] (-2.5,2.5)--(2.5,2.5)--(2.5,-2.5)--(-2.5,-2.5)--(-2.5,2.5);
  \portaabar{-2}{-2} \portaabar{0}{-2} \portaabar{2}{-2}
  \portacabar{-1}{-1} \portacabar{1}{-1}
  \portaabar{-2}{0} \portaabar{0}{0} \portaabar{2}{0}
  \portacabar{-1}{1} \portacabar{1}{1}
  \portaabar{-2}{2} \portaabar{0}{2} \portaabar{2}{2}
  %\cables{0}{0}\cables{1}{2.2}
  \filldraw (-2.5,-2.5) circle (.05);
  \filldraw (-1.5,-2.5) circle (.05);
  \filldraw (-.5,-2.5) circle (.05);
  \filldraw (.5,-2.5) circle (.05);
  \filldraw (1.5,-2.5) circle (.05);
  \filldraw (2.5,-2.5) circle (.05);
  \filldraw (-2.5,2.5) circle (.05);
  \filldraw (-1.5,2.5) circle (.05);
  \filldraw (-.5,2.5) circle (.05);
  \filldraw (.5,2.5) circle (.05);
  \filldraw (1.5,2.5) circle (.05);
  \filldraw (2.5,2.5) circle (.05);
  \filldraw (-2.5,-1.5) circle (.05);
  \filldraw (-2.5,-.5) circle (.05);
  \filldraw (-2.5,.5) circle (.05);
  \filldraw (-2.5,1.5) circle (.05);
  \filldraw (2.5,-1.5) circle (.05);
  \filldraw (2.5,-.5) circle (.05);
  \filldraw (2.5,.5) circle (.05);
  \filldraw (2.5,1.5) circle (.05);
  %
  %\draw (-.5,2.9) node {$_0$};
  \filldraw[red] (-1.5,.5) circle (.1);
  \filldraw[red] (-.5,.5) circle (.1);
  \filldraw[red] (.5,.5) circle (.1);
  \filldraw[red] (1.5,.5) circle (.1);
  \draw[red, thick] (-1.5,.5)--(1.5,.5);
\end{tikzpicture}\ .
\end{align}
And, if we use the operator 
\begin{align}\label{eq:op 2 v}
\begin{tikzpicture}[baseline=-2]
  \filldraw[red] (0,.5) circle (.1);
  \filldraw[red] (0,-.5) circle (.1);
  \draw[red, thick] (0,.5)--(0,-.5);
\end{tikzpicture}
=
\begin{tikzpicture}[baseline=-2]
  \portaabar{.5}{0} \portabar{-.5}{0}
  \filldraw[red] (0,.5) circle (.1);
\end{tikzpicture}\ ,
\end{align}
then we can write the same state as
\begin{align}\label{eq:a psi 2'}
  \a_0 \ket{\Psi^{\rm fl}} =
\begin{tikzpicture}[baseline=-2]
  \draw[blue!30, very thick] (-2.5,2.5)--(2.5,2.5)--(2.5,-2.5)--(-2.5,-2.5)--(-2.5,2.5);
  \portaabar{-2}{-2} \portaabar{0}{-2} \portaabar{2}{-2}
  \portacabar{-1}{-1} \portacabar{1}{-1}
  \portaabar{-2}{0} \portaabar{0}{0} \portaabar{2}{0}
  \portacabar{-1}{1} \portacabar{1}{1}
  \portaabar{-2}{2} \portaabar{0}{2} \portaabar{2}{2}
  %\cables{0}{0}\cables{1}{2.2}
  \filldraw (-2.5,-2.5) circle (.05);
  \filldraw (-1.5,-2.5) circle (.05);
  \filldraw (-.5,-2.5) circle (.05);
  \filldraw (.5,-2.5) circle (.05);
  \filldraw (1.5,-2.5) circle (.05);
  \filldraw (2.5,-2.5) circle (.05);
  \filldraw (-2.5,2.5) circle (.05);
  \filldraw (-1.5,2.5) circle (.05);
  \filldraw (-.5,2.5) circle (.05);
  \filldraw (.5,2.5) circle (.05);
  \filldraw (1.5,2.5) circle (.05);
  \filldraw (2.5,2.5) circle (.05);
  \filldraw (-2.5,-1.5) circle (.05);
  \filldraw (-2.5,-.5) circle (.05);
  \filldraw (-2.5,.5) circle (.05);
  \filldraw (-2.5,1.5) circle (.05);
  \filldraw (2.5,-1.5) circle (.05);
  \filldraw (2.5,-.5) circle (.05);
  \filldraw (2.5,.5) circle (.05);
  \filldraw (2.5,1.5) circle (.05);
  \filldraw[red] (.5,2.5) circle (.1);
  \filldraw[red] (.5,1.5) circle (.1);
  \draw[red, thick] (.5,2.5)--(.5,1.5);
\end{tikzpicture}\ .
\end{align}
Clearly, there is a large number of ways of writing the same state.

If $\a$ is a soliton of $\tikz[baseline]{\portabar{0}{.3}}$, or a linear combination thereof, then the operators 
\tikz{\filldraw[red] (-.5,.5) circle (.1);\filldraw[red] (.5,.5) circle (.1);\draw[red, thick] (-.5,.5)--(.5,.5);}, \tikz{\filldraw[red] (-.5,.5) circle (.1);\filldraw[red] (.5,.5) circle (.1);\filldraw[red] (1.5,.5) circle (.1);\filldraw[red] (2.5,.5) circle (.1);\draw[red, thick] (-.5,.5)--(2.5,.5);} 
and \tikz[baseline=-3]{\filldraw[red] (.5,.5) circle (.1);\filldraw[red] (.5,-.5) circle (.1);\draw[red, thick] (.5,.5)--(.5,-.5);}
defined in \eqref{eq:op 2}, \eqref{eq:op 4} and \eqref{eq:op 2 v} act non-trivially on a single site only.
Contrary, if $\tikz[baseline]{\portabar{0}{.3}}$ is generic then, in addition of having no solitons, we expect that all alternative ways of writing the one-particle state \eqref{eq:a psi} involve operators with terms acting non-trivially in more than one site, like for example $\tikz{\filldraw[red] (-.5,.5) circle (.1);\filldraw[red] (.5,.5) circle (.1);\draw[red, thick] (-.5,.5)--(.5,.5);} = \b \otimes \c +\cdots$
This fact eliminates the ambiguity of the particle position. Therefore we use the following prescription: If a state can be written as a tensor network of the building blocks, and one extra tensor on a single link, then we interpret this state as a space with one particle at the position of mentioned link. 

This prescription allows us to put a particle at any location in the bulk (not necessarily the boundary), like for example
\begin{align}%\label{eq:a psi 4}
\begin{tikzpicture}[baseline=-2]
  \draw[blue!30, very thick] (-2.5,2.5)--(2.5,2.5)--(2.5,-2.5)--(-2.5,-2.5)--(-2.5,2.5);
  \portaabar{-2}{-2} \portaabar{0}{-2} \portaabar{2}{-2}
  \portacabar{-1}{-1} \portacabar{1}{-1}
  \portaabar{-2}{0} \portaabar{0}{0} \portaabar{2}{0}
  \portacabar{-1}{1} \portacabar{1}{1}
  \portaabar{-2}{2} \portaabar{0}{2} \portaabar{2}{2}
  %\cables{0}{0}\cables{1}{2.2}
  \filldraw (-2.5,-2.5) circle (.05);
  \filldraw (-1.5,-2.5) circle (.05);
  \filldraw (-.5,-2.5) circle (.05);
  \filldraw (.5,-2.5) circle (.05);
  \filldraw (1.5,-2.5) circle (.05);
  \filldraw (2.5,-2.5) circle (.05);
  \filldraw (-2.5,2.5) circle (.05);
  \filldraw (-1.5,2.5) circle (.05);
  \filldraw (-.5,2.5) circle (.05);
  \filldraw (.5,2.5) circle (.05);
  \filldraw (1.5,2.5) circle (.05);
  \filldraw (2.5,2.5) circle (.05);
  \filldraw (-2.5,-1.5) circle (.05);
  \filldraw (-2.5,-.5) circle (.05);
  \filldraw (-2.5,.5) circle (.05);
  \filldraw (-2.5,1.5) circle (.05);
  \filldraw (2.5,-1.5) circle (.05);
  \filldraw (2.5,-.5) circle (.05);
  \filldraw (2.5,.5) circle (.05);
  \filldraw (2.5,1.5) circle (.05);
  \filldraw[red] (.5,-.5) circle (.1);
%  \filldraw[red] (-1.5,.5) circle (.1);
\end{tikzpicture}\ .
\end{align}
However, In the following subsection we see that this prescription cannot be applied to all single-particle states.

\subsection{Dynamics of spaces with one particle}

\begin{figure*}
  \centering  
  \includegraphics[width=17cm]{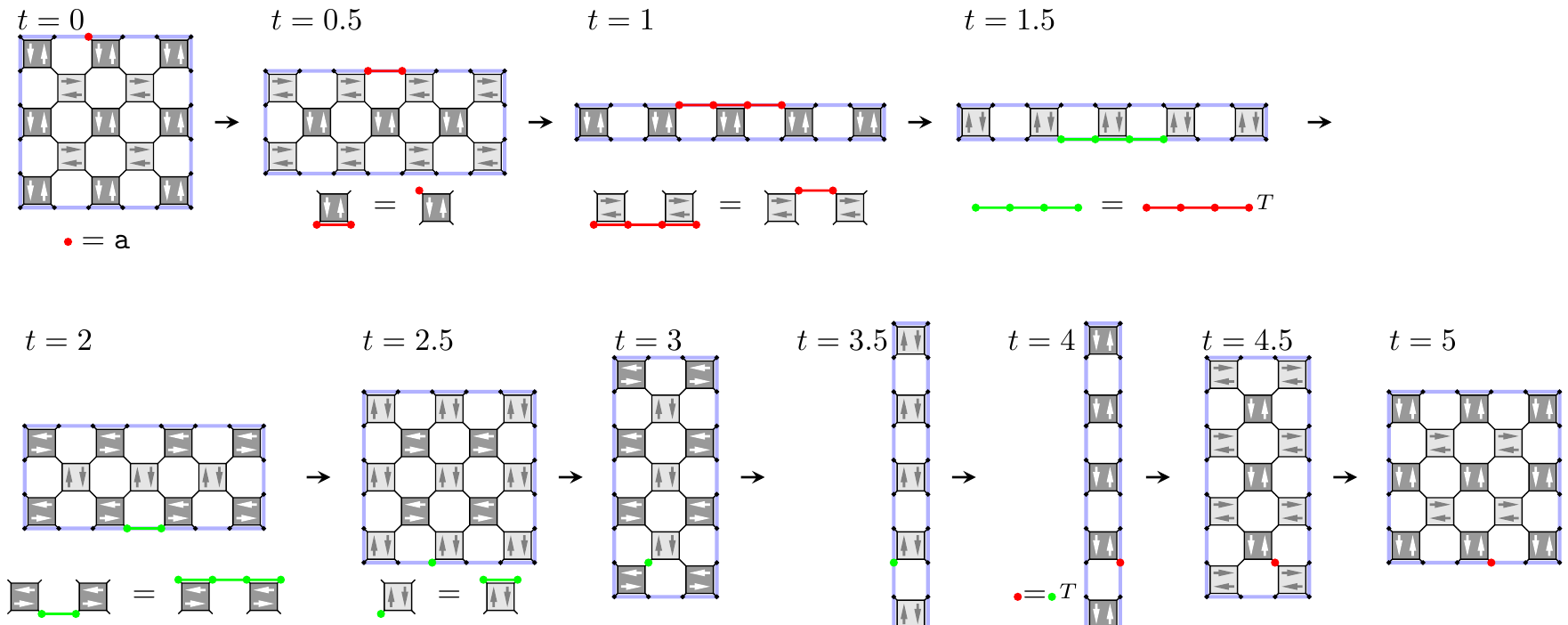}
  \caption{\textbf{Dynamics of flat space with one particle} - sequence of states at different times $t$. At $t=0$, initial state of geometry (in grey) and particle (red dot represents arbitrary operator $\a$). At each half time step $t$, state of geometry and particle, and below, definition of the operator representing the particle in terms of operators defined in previous time steps. Each green operator is the transpose of the red operator with the same form.
  At $t=2$ and $t=2.5$ the operators shrink - this requires the dual unitary 
  %\tikz[baseline=-2]{\filldraw[red] (0,0) circle (.1);}
  and the operator $\a$ to be real, because then operator equalities at $t=2$ and $t=2.5$ are equivalent to those at $t=1$ and $t=0.5$ respectively.
  At $t=5$, the geometry returns to its original state but the particle is in the antipodal position. Hence the period is $\Delta t=10$.}
  \label{flat_m_cycle} 
\end{figure*} 

We have already seen the dynamics of the empty-space state $\ket{\Psi^{\rm fl}}$ in Figure~\ref{flat_g_cycle}.
Now, let us explore what happens when we add one particle $\a_0 \ket{\Psi^{\rm fl}}$ at the boundary site $x=0$.
Figure~\ref{flat_m_cycle} shows that the geometry of $T^t \a_0 \ket{\Psi^{\rm fl}}$ evolves like that of $T^t \ket{\Psi^{\rm fl}}$, if we represent the particle with certain inserted operators, also defined in Figure~\ref{flat_m_cycle}.
This evolution has a very short period ($\Delta t=10$) when the dual unitary and the operator are real
\begin{align}\label{DU real}
  \tikz[baseline=-2]{\portabar{0}{0}}
  =
  \tikz[baseline=-2]{\portabarc{0}{0}}\ ,
  \qquad
  \a = \a^*.
\end{align}
(The same holds if, instead of real, the operator is symmetric $\a= \a^\intercal$.) 
This implies that the subsequent operators
\begin{align}
\begin{tikzpicture}
  \filldraw[red] (-1,-3) circle (.1);
  \filldraw[red] (-2,-3) circle (.1);
  \filldraw[red] (-1,-3)--(-2,-3);
  \draw(0,-3) node {$=$};  
  \portabar{1.5}{-2.5} \portabart{1.5}{-3.5}
  \filldraw[red] (1,-3) circle (.1);
\end{tikzpicture}  
\qquad
\begin{tikzpicture}
  \filldraw[red] (-1,-3) circle (.1);
  \filldraw[red] (-2,-3) circle (.1);
  \filldraw[red] (-3,-3) circle (.1);
  \filldraw[red] (-4,-3) circle (.1);
  \filldraw[red] (-1,-3)--(-4,-3);
  \draw(0,-3) node {$=$};  
  \portacabars{1.5}{-2.5} \portacabar{1.5}{-3.5}
  \portacabars{3.5}{-2.5} \portacabar{3.5}{-3.5}
  \filldraw[red] (2,-3) circle (.1);
  \filldraw[red] (3,-3) circle (.1);
  \filldraw[red] (2,-3)--(3,-3);
\end{tikzpicture}  
\end{align}
are real too; and it makes the operator equalities at $t=2$ and $t=2.5$ the transposition of equalities at $t=1$ and $t=.5$ respectively, rendering them equivalent. This implies the shrinking of operators at $t=2$ and $2.5$, and a return to the original $\a$. Note that the recurrence time of $\a_0 \ket{\Psi^{\rm fl}} \in \H_n$ is $t= \frac n 2$, while that of a typical state in $\H_n$ is $t\sim \exp q^n$.

Figure~\ref{flat_m_cycle} also shows that some intermediate states cannot be written with a single-site tensor representing the particle. Hence, as mentioned above, the position of the particle in these states is not well defined. 

When the particle $\a_x$ is located at other points of the boundary of flat space, the corresponding state $\a_x \ket{\Psi^{\rm fl}}$ experiences an evolution similar to that with $x=0$. Specifically, for any $x$, the particle reaches the antipodal point at $t= \frac n 4$, and the period of the evolution is $\Delta t= \frac n 2$. Figure~\ref{flat_m_corner} shows the evolution of the same flat space when the initial location of the particle is a corner ($x=-2$). 

Figure~\ref{ads_m_cycle} shows the dynamics of toy AdS with a particle at the boundary location $x=0$, that is $\a_0 \ket{\Psi_{32} ^{\mathrm{toy}}}$. In this case we also observe that the particle reaches the antipodal point at $t= \frac n 4$ and the recurrence happens at $t= \frac n 2$. However, contrary to flat space, certain initial positions of the particle (e.g.~$\a_9 \ket{\Psi_{32} ^{\mathrm{toy}}}$) generate a state with recurrence time much longer than $\Delta t = \frac n 2$. In the following section we observe that two-particle states also give rise to this second type of dynamics, with longer recurrence times.

%Actually, in this case I expect this recurrence to not be exact, and to happen in a much longer time scale $\Delta t\sim \exp{q^{vt}}$, where $v$ is some effective volume. In other words, the state $\a_9 \ket{\Psi_{32} ^{\mathrm{toy}}}$ has support on a subspace where the gauge circuit has chaotic dynamics. In the next subsection we analyse these chaotic subspaces with a simple example.

%{\bf [In Section~\ref{sec:matter} we show that the  relationship \eqref{eq: T4 Psi} between time at different radial locations applies to local clocks made of matter.]}

\subsection{One vs two-particle dynamics}

In this subsection we compare the dynamics of states with different numbers of particles. We do so by numerically obtaining the dimension of the effective subspace explored by the evolution of each such state. This is done with the dynamics of the four-layer circuit \eqref{eq:T gauge} generated by a randomly sampled dual unitary $\tikz[baseline]{\porta{0}{.3}}$ with real entries and local dimension $q=3$. We use the four-layer circuit to get rid of the constraint $\tikz[baseline]{\portat{0}{.3}} = \tikz[baseline]{\portas{0}{.3}}$ and, in this way, enlarge the size of the set of dual unitaries. This fact increases the similarity between the properties of different instances of the random dual unitary, due to concentration of measure.

For the sake of simplicity we perform the above-described comparison with the small empty-space state
\begin{align}
  \ket{\Psi}
  =  
\begin{tikzpicture}[baseline=-2]
  \draw[blue!30, very thick] (-.5,-.5)--(-.5,.5)--(2.5,.5)--(2.5,-.5)--(-.5,-.5);
  \portadrr{0}{0} \portadr{2}{0}
  \filldraw (.5,.5) circle (.05);
  \filldraw (-.5,.5) circle (.05);
  \filldraw (-.5,-.5) circle (.05);
  \filldraw (.5,-.5) circle (.05);
  \filldraw (2.5,.5) circle (.05);
  \filldraw (1.5,.5) circle (.05);
  \filldraw (1.5,-.5) circle (.05);
  \filldraw (2.5,-.5) circle (.05);
  \draw (.5,.9) node {$_0$};
  \draw (1.5,.9) node {$_1$};
  \draw (-.8,.85) node {$_{-1}$};
\end{tikzpicture}
  \ \in \H_8 \ .
\end{align} 
We calculate its dynamics by proceeding as in \eqref{eq: Teven Psi} and \eqref{eq: Todd Psi} but with the four-layer circuit $T= T_{[4]} T_{[3]} T_{[2]} T_{[1]}$ described in \eqref{eq:T gauge}, obtaining the cycle 
\begin{equation}\label{eq:0 particles}
\begin{array}{ll}
  T_{[1]}  \ket{\Psi}
  =  
\begin{tikzpicture}[baseline=-2]
  \draw[blue!30, very thick] (-.5,-1.5)--(-.5,1.5)--(.5,1.5)--(.5,-1.5)--(-.5,-1.5);
  \porta{0}{1} \portarr{0}{-1}
  \filldraw (-.5,1.5) circle (.05);
  \filldraw (.5,1.5) circle (.05);
  \filldraw (.5,.5) circle (.05);
  \filldraw (.5,-.5) circle (.05);
  \filldraw (.5,-1.5) circle (.05);
  \filldraw (-.5,-1.5) circle (.05);
  \filldraw (-.5,-.5) circle (.05);
  \filldraw (-.5,.5) circle (.05);
  \draw (-.5,1.9) node {$_0$};
\end{tikzpicture}\ ,
  &T_{[2]} T_{[1]}  \ket{\Psi}
  =  
\begin{tikzpicture}[baseline=-2]
  \draw[blue!30, very thick] (-.5,-1.5)--(-.5,1.5)--(.5,1.5)--(.5,-1.5)--(-.5,-1.5);
  \portacs{0}{1} \portad{0}{-1}
  \filldraw (-.5,1.5) circle (.05);
  \filldraw (.5,1.5) circle (.05);
  \filldraw (.5,.5) circle (.05);
  \filldraw (.5,-.5) circle (.05);
  \filldraw (.5,-1.5) circle (.05);
  \filldraw (-.5,-1.5) circle (.05);
  \filldraw (-.5,-.5) circle (.05);
  \filldraw (-.5,.5) circle (.05);
  \draw (-.5,1.9) node {$_0$};
\end{tikzpicture}\ ,  
\\
  T_{[3]} T_{[2]} T_{[1]}  \ket{\Psi}
  =  
\begin{tikzpicture}[baseline=-2]
  \draw[blue!30, very thick] (-.5,-.5)--(-.5,.5)--(2.5,.5)--(2.5,-.5)--(-.5,-.5);
  \portasr{2}{0} \portatq{0}{0}
  \filldraw (.5,.5) circle (.05);
  \filldraw (-.5,.5) circle (.05);
  \filldraw (-.5,-.5) circle (.05);
  \filldraw (.5,-.5) circle (.05);
  \filldraw (2.5,.5) circle (.05);
  \filldraw (1.5,.5) circle (.05);
  \filldraw (1.5,-.5) circle (.05);
  \filldraw (2.5,-.5) circle (.05);
  \draw (.5,.9) node {$_0$};
\end{tikzpicture}\ ,
  &T \ket{\Psi}
  =  
\begin{tikzpicture}[baseline=-2]
  \draw[blue!30, very thick] (-.5,-.5)--(-.5,.5)--(2.5,.5)--(2.5,-.5)--(-.5,-.5);
  \portadrr{0}{0} \portadr{2}{0}
  \filldraw (.5,.5) circle (.05);
  \filldraw (-.5,.5) circle (.05);
  \filldraw (-.5,-.5) circle (.05);
  \filldraw (.5,-.5) circle (.05);
  \filldraw (2.5,.5) circle (.05);
  \filldraw (1.5,.5) circle (.05);
  \filldraw (1.5,-.5) circle (.05);
  \filldraw (2.5,-.5) circle (.05);
  \draw (.5,.9) node {$_0$};
\end{tikzpicture}\ .
\end{array}
\end{equation} 
We observe that the period is $\Delta t=1$, and so, $\ket\Psi$ is an eigenstate of $T$.
If we add a particle at the boundary location $x=0$ by applying an operator $\tikz{\filldraw[red] (0,.6) circle (.1);} = \a \in \A$ then we obtain the state 
\begin{align}
  \a_0 \ket{\Psi}
  =  
\begin{tikzpicture}[baseline=-2]
  \draw[blue!30, very thick] (-.5,-.5)--(-.5,.5)--(2.5,.5)--(2.5,-.5)--(-.5,-.5);
  \portadrr{0}{0} \portadr{2}{0}
  \filldraw (-.5,.5) circle (.05);
  \filldraw[red] (.5,.5) circle (.1);
  \filldraw (-.5,-.5) circle (.05);
  \filldraw (.5,-.5) circle (.05);
  \filldraw (2.5,.5) circle (.05);
  \filldraw (1.5,.5) circle (.05);
  \filldraw (1.5,-.5) circle (.05);
  \filldraw (2.5,-.5) circle (.05);
  \draw (.5,.9) node {$_0$};
  \draw (1.5,.9) node {$_1$};
  \draw (-.8,.85) node {$_{-1}$};
\end{tikzpicture}\ ,
\end{align} 
which has the similar evolution
\begin{equation}\label{eq:1 particles}
\begin{array}{ll}
  T_{[1]}  \a_0 \ket{\Psi}
  =  
\begin{tikzpicture}[baseline=-2]
  \draw[blue!30, very thick] (-.5,-1.5)--(-.5,1.5)--(.5,1.5)--(.5,-1.5)--(-.5,-1.5);
  \porta{0}{1} \portarr{0}{-1}
  \filldraw (-.5,1.5) circle (.05);
  \filldraw (.5,1.5) circle (.05);
  \filldraw (.5,.5) circle (.05);
  \filldraw (.5,-.5) circle (.05);
  \filldraw (.5,-1.5) circle (.05);
  \filldraw (-.5,-1.5) circle (.05);
  \filldraw (-.5,-.5) circle (.05);
  \filldraw[green] (-.5,.5) circle (.1);
  \draw (-.5,1.9) node {$_0$};
\end{tikzpicture}\ ,
  &T_{[2]} T_{[1]}  \a_0 \ket{\Psi}
  =  
\begin{tikzpicture}[baseline=-2]
  \draw[blue!30, very thick] (-.5,-1.5)--(-.5,1.5)--(.5,1.5)--(.5,-1.5)--(-.5,-1.5);
  \portacs{0}{1} \portad{0}{-1}
  \filldraw (-.5,1.5) circle (.05);
  \filldraw (.5,1.5) circle (.05);
  \filldraw[red] (.5,.5) circle (.1);
  \filldraw (.5,-.5) circle (.05);
  \filldraw (.5,-1.5) circle (.05);
  \filldraw (-.5,-1.5) circle (.05);
  \filldraw (-.5,-.5) circle (.05);
  \filldraw (-.5,.5) circle (.05);
  \draw (-.5,1.9) node {$_0$};
\end{tikzpicture}\ ,  
\\
  T_{[3]} T_{[2]} T_{[1]} \a_0 \ket{\Psi}
  =  
\begin{tikzpicture}[baseline=-2]
  \draw[blue!30, very thick] (-.5,-.5)--(-.5,.5)--(2.5,.5)--(2.5,-.5)--(-.5,-.5);
  \portasr{2}{0} \portatq{0}{0}
  \filldraw[green] (1.5,.5) circle (.1);
  \filldraw (-.5,.5) circle (.05);
  \filldraw (-.5,-.5) circle (.05);
  \filldraw (.5,-.5) circle (.05);
  \filldraw (2.5,.5) circle (.05);
  \filldraw (.5,.5) circle (.05);
  \filldraw (1.5,-.5) circle (.05);
  \filldraw (2.5,-.5) circle (.05);
  \draw (.5,.9) node {$_0$};
\end{tikzpicture}\ ,
  &T \a_0 \ket{\Psi}
  =  
\begin{tikzpicture}[baseline=-2]
  \draw[blue!30, very thick] (-.5,-.5)--(-.5,.5)--(2.5,.5)--(2.5,-.5)--(-.5,-.5);
  \portadrr{0}{0} \portadr{2}{0}
  \filldraw (.5,.5) circle (.05);
  \filldraw (-.5,.5) circle (.05);
  \filldraw (-.5,-.5) circle (.05);
  \filldraw[red] (1.5,-.5) circle (.1);
  \filldraw (2.5,.5) circle (.05);
  \filldraw (1.5,.5) circle (.05);
  \filldraw (.5,-.5) circle (.05);
  \filldraw (2.5,-.5) circle (.05);
  \draw (.5,.9) node {$_0$};
\end{tikzpicture}\ ,
\end{array}
\end{equation} 
where $\tikz{\filldraw[green] (0,.6) circle (.1);} = \a^\intercal \in \A$ is the transpose.
That is, after one time step, the particle is in the antipodal location ($x=4$), so the period of the state $\a_0 \ket{\Psi}$ is $\Delta t=2$.
We have checked that the same dynamical behaviour happens for any initial location $x \in \mathbb Z_8$ of the particle.
In summary, all single-particle states generate a closed orbit of period $\Delta t=2$, which allows us to construct exact eigenstates via \eqref{eq:eigen}. 
Because of the existence of simple eigenstates, we consider the zero and one-particle subspace ``integrable".

Next, let us compare the above results with the dynamics of the following two and three-particle states
\begin{align}\label{eq:2 3 particle}
  \begin{tikzpicture}
  \draw[blue!30, very thick] (.5,6.2)--(3.5,6.2)--(3.5,5.2)--(.5,5.2)--(.5,6.2);
  \portadrr{1}{5.7} \portadr{3}{5.7}
  \filldraw[red] (.5,6.2) circle (.1);
  \filldraw[red] (1.5,6.2) circle (.1);
  \filldraw (2.5,6.2) circle (.05);
  \filldraw (3.5,5.2) circle (.05);
  \filldraw (3.5,5.2) circle (.05);
  \filldraw (2.5,5.2) circle (.05);
  \filldraw (1.5,5.2) circle (.05);
  \filldraw (.5,5.2) circle (.05);
  \end{tikzpicture}\ ,
  \qquad
  \begin{tikzpicture}
  \draw[blue!30, very thick] (.5,8.3)--(3.5,8.3)--(3.5,7.3)--(.5,7.3)--(.5,8.3);
  \portadrr{1}{7.8} \portadr{3}{7.8}
  \filldraw (.5,8.3) circle (.05);
  \filldraw[red] (1.5,8.3) circle (.1);
  \filldraw[red] (2.5,8.3) circle (.1);
  \filldraw (3.5,7.3) circle (.05);
  \filldraw (3.5,7.3) circle (.05);
  \filldraw (2.5,7.3) circle (.05);
  \filldraw (1.5,7.3) circle (.05);
  \filldraw (.5,7.3) circle (.05);
  \end{tikzpicture}\ ,
  \qquad
  \begin{tikzpicture}
  \draw[blue!30, very thick] (.5,11.5)--(3.5,11.5)--(3.5,10.5)--(.5,10.5)--(.5,11.5);
  \portadrr{1}{11} \portadr{3}{11}
  \filldraw[red] (.5,11.5) circle (.1);
  \filldraw[red] (1.5,11.5) circle (.1);
  \filldraw (2.5,11.5) circle (.05);
  \filldraw[red] (2.5,11.5) circle (.1);
  \filldraw (3.5,10.5) circle (.05);
  \filldraw (2.5,10.5) circle (.05);
  \filldraw (1.5,10.5) circle (.05);
  \filldraw (.5,10.5) circle (.05);
\end{tikzpicture}\ ,    
\end{align}
generated with a random matrix $\tikz{\filldraw[red] (0,.6) circle (.1);} = \a \in \A$ with real entries.
For each such pure state $|\Psi\rangle\! \langle\Psi |$ we numerically calculate its time average after $t$ time steps
\begin{align}
  \rho(t) = \frac 1 {t+1}  \sum_{r=0}^{t} T^{r} |\Psi\rangle\! \langle\Psi | T^{-r}\ .
\end{align}
And for each such mixed state we calculate the effective dimension of its support
\begin{align}
  d_{\mathrm{eff}}(t) = \left[ \mathrm{tr} \rho^2(t) \right]^{-1}\ ,
\end{align}
which takes into account the different weights of the eigenvalues, as in the second-order Renyi entropy via $\log_2 d_{\mathrm{eff}} = h_2 (\rho) = -\log_2 \rho^2$. For example, if a states $\rho$ is proportional to a projector $\rho^2 \propto \rho$ then $d_{\mathrm{eff}}$ is equal to the dimension of the projector. 
%We have seen in \eqref{eq:0 particles} that the state with zero particles $\ket\Psi$ is an eigenstate of the dynamics, therefore its evolution generates a one-dimensional subspace $d_{\mathrm{eff}}(\infty) =1$. According to \eqref{eq:1 particles}, states with one particle $\a_x \ket\Psi$ generate an orbit with period $\Delta t=2$, therefore its evolution generates a two-dimensional subspace $d_{\mathrm{eff}}(\infty) =2$ (also calculated numerically in Figure~\ref{fig:entropies}). In both cases, zero and one particle, the states generate a closed orbit, and hence, a family of eigenstates. For this reason we say that this cases are somehow ``integrable". 
The effective dimension of the above two and three-particle states is plotted in Figure~\ref{fig:entropies}. We observe that the corresponding curves attain a much larger value than those of the zero and one-particle states, which have $d_{\mathrm{eff}} (\infty) =1,2$; and take a longer time to equilibrate. In addition, the lack of convergence to an integer value suggests that the corresponding evolutions do not generate a closed orbit, and hence, do not have simple eigenstates associated to them.
This establishes a sharp distinction between states with one particle or less on one side, and two particles or more on the other. 
In other words, the dynamics of a particle strongly depends on the presence of other particles, implying an interaction between the two particles. This interaction can be interpreted as gravitational back-reaction.

The fact that in figures~\ref{flat_m_cycle}, \ref{flat_m_corner} and~\ref{ads_m_cycle} the underlaying geometry (i.e. the bare tensor network) seems to be unaffected by the presence of the particle is not in contradiction with the above-mentioned back-reaction, since there is freedom on how we represent states. That is, two different underlying tensor networks might represent the same state if the matter is represented by different inserted operators in each of them.

Figure~\ref{fig:entropies} also shows that the behaviour of the states \eqref{eq:2 3 particle} is similar to that of a random state in $\H_8$. 
This opens the possibility that the dynamics in the two or more particles subspace is quantum-chaotic.
However, the results of the following subsection suggest that the contrary is true.

\begin{figure}\centering
  \includegraphics[width=89mm]{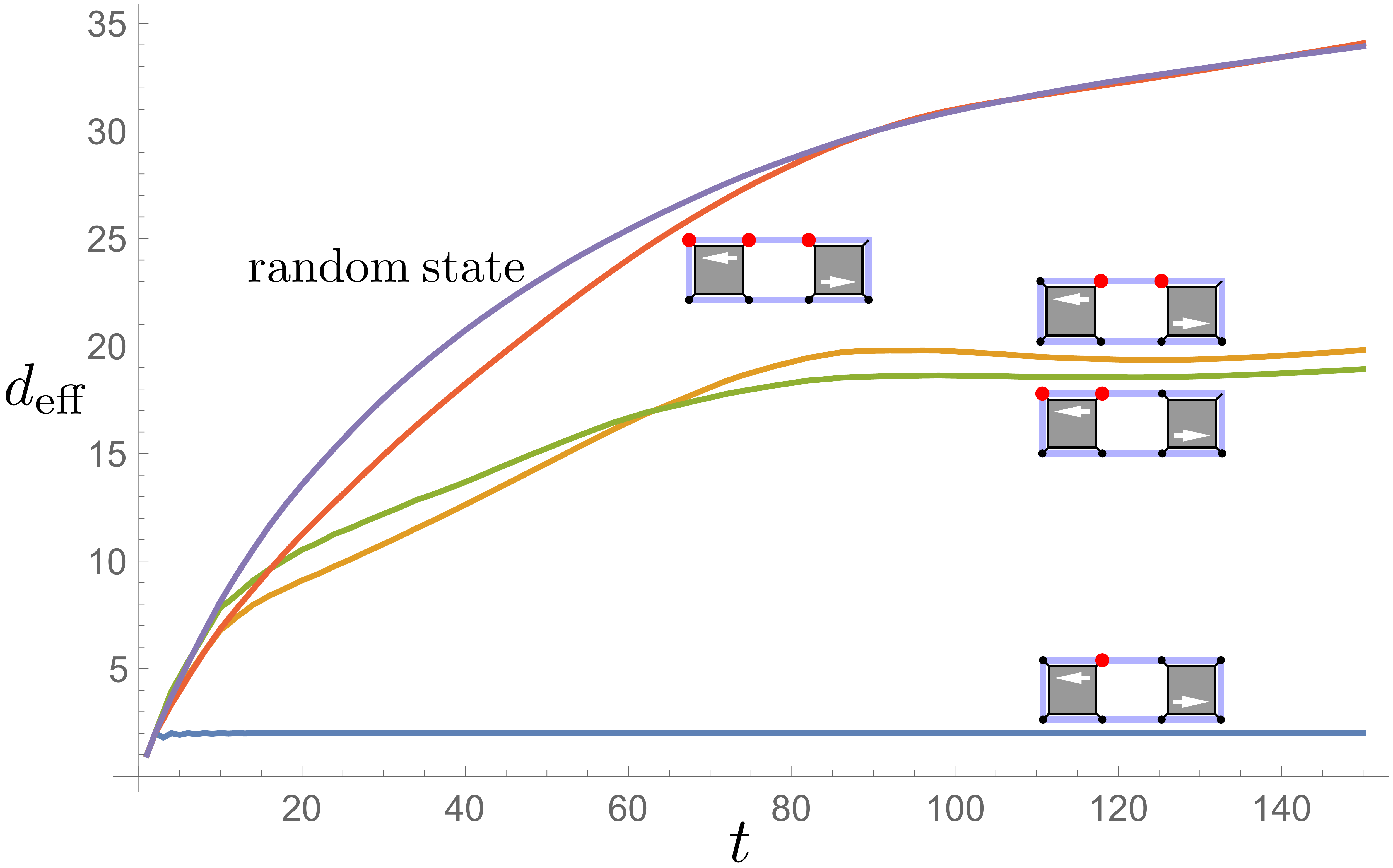}
  \vspace{-6mm}
  \caption{\textbf{Effective dimension of 1,2,3 particles.} Each curve shows the dimension $d_{\mathrm{eff}} (t)$ of the effective subspace as a function of time $t$, generated by the evolution of the tensor-network state next to it. The top curve is generated by a random state with real entries. The dynamics is generated by the four-layer circuit \eqref{eq:T gauge} with a real random dual unitary. We observe a qualitative distinction between the single-particle state and the rest. The differences between the other four states are not essential, because they vary according to the instance of random dual unitary.}
  \label{fig:entropies}
\end{figure}

\subsection{Absence of quantum chaos}

First of all, let us see that our method for sampling random dual unitaries generates an ensemble qualitatively similar to that of Haar-random unitaries (with no duality constraint). Actually, since we generate random dual unitaries with real entries, we need to compare these with Haar-random orthogonal matrices. In order to do so, we generate a set $\mathcal U$ of 85 random dual unitaries $\u \in \mathcal U$ with real entries, and we calculate the spectral form factor
\begin{align}\label{eq:sff}
  K_{\mathcal U} (t) = \frac 1 {|\mathcal U|} \sum_{\u \in \mathcal U} \left| \mathrm{tr} (\u^t) \right|^2\ .
\end{align}
The result are the red points in Figure~\ref{fig:SFF u}, which are contrasted with the form factor of the orthogonal group
\begin{align}
  K_{{\rm SO}(d)} (t) = 
  \left\{\begin{array}{ll}
    2t -t\log_2 \!\left( 1-2 t/d \right) & \mbox{ if $t<d$} \\
    2 d -t\log_2 \!\left( \frac{2 t+d}{2t-d} \right)  & \mbox{ if $t\geq d$}
  \end{array}\right.\ ,
\end{align}
with $d=q^2=9$.
We observe in Figure~\ref{fig:SFF u} that the two form factors are qualitatively similar. In particular, both display the so called ``dip" for $t<9$. 
This behaviour signals the presence of quantum chaos in random dual unitaries.
%However, in what followhen we use such chaotic dual unitaries to construct the evolution operator $T$,  in what follows we observe that this quantum chaos disappears when we constructis the evolution operator $T$.

\begin{figure}\centering
  \includegraphics[width=70mm]{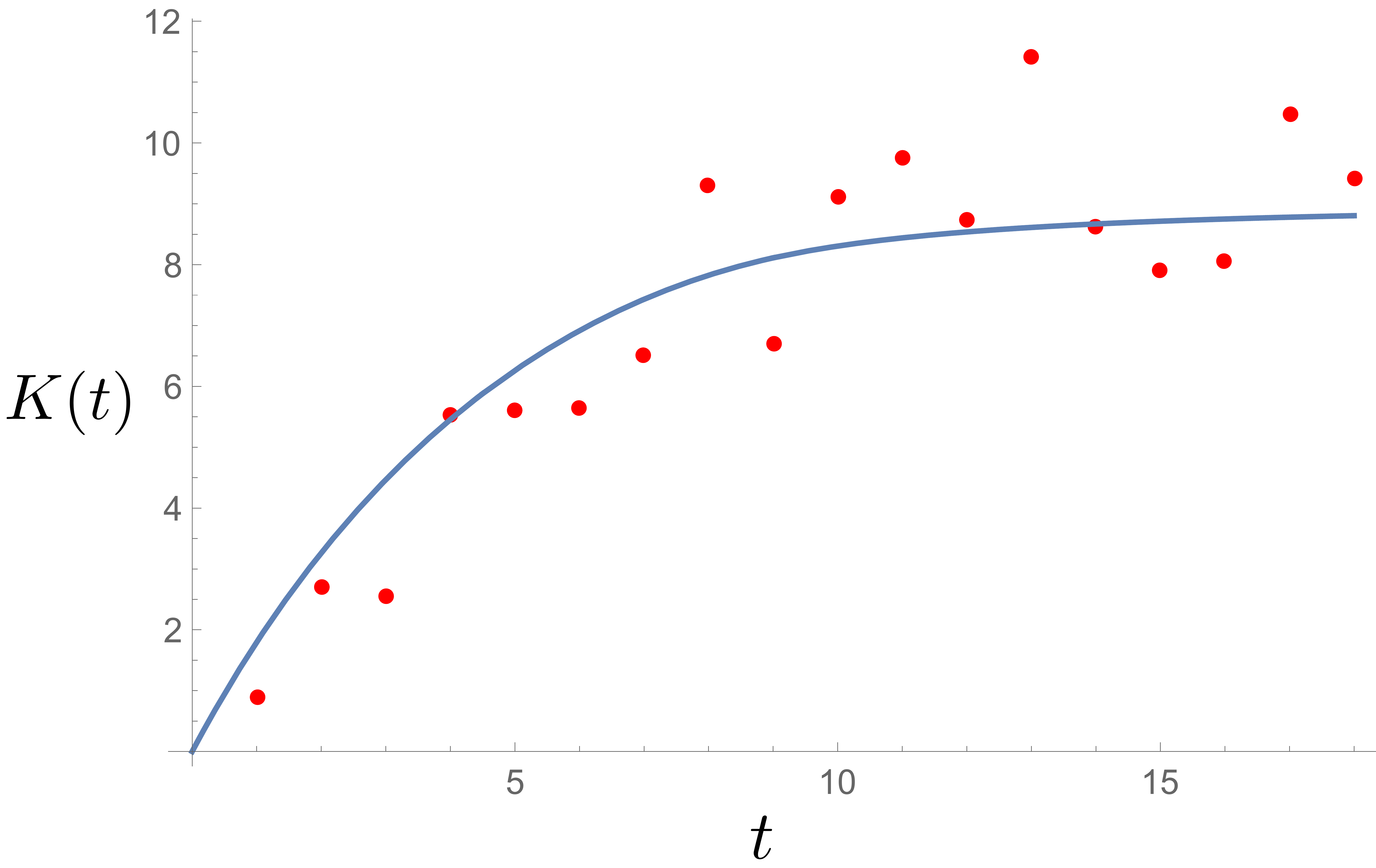}
  \vspace{-5mm}
  \caption{\textbf{Spectral form factor of dual unitary} with real entries and local dimension $q=3$, averaged over 85 random instances (red dots). Spectral form factor of the orthogonal group ${\rm SO}(d)$ of dimension $d=q^2=9$ (blue line). Both plots display the initial ``dip" and are qualitatively similar. This signals the presence of quantum chaos in random dual unitaries.}
  \label{fig:SFF u}
\end{figure}

With the above mentioned 85 instances $\u\in \mathcal U$, we construct 85 instances of the four-layer evolution operator $T$ for the chains with $n=4$ and $n=8$. Using formula \eqref{eq:sff} but replacing $\u$ by $T$, we calculate the corresponding form factors and plot them in Figure~\ref{fig:SFF T}. We observe that both form factors are essentially flat, with no ``dip", a behaviour characteristic of Poisson level statistics, which appears in ``integrable" systems. This is very surprising, because both evolution operators are constructed with unitaries $\u$ which, as discussed above, display quantum chaos. Hence, we conclude that the special structure of the conformal circuit \eqref{eq:T gauge} cancels the chaos present in the building block $\tikz[baseline]{\porta{0}{.3}}$.

It is proven in \cite{Farshi_22} that the asymptotic value of the form factor is
\begin{align}
  K(\infty) = \sum_E g_E^2\ ,
\end{align}
where $E$ are the eigenvalues of the evolution operator and $g_E$ the corresponding degeneracy.
In the absence of degeneracies ($g_E=1$) we have $K(\infty) = q^n$, which is the standard value in generic models including random unitaries.
However, the two approximately constant plots in Figure~\ref{fig:SFF T} show values for $K(t)$ much larger than the Hilbert-space dimension $q^n$.
This signals that the spectrum of $T$ is highly structured in a way that leads to large degeneration. This structure will be studied in future work.

\begin{figure}\centering
  \includegraphics[width=88mm]{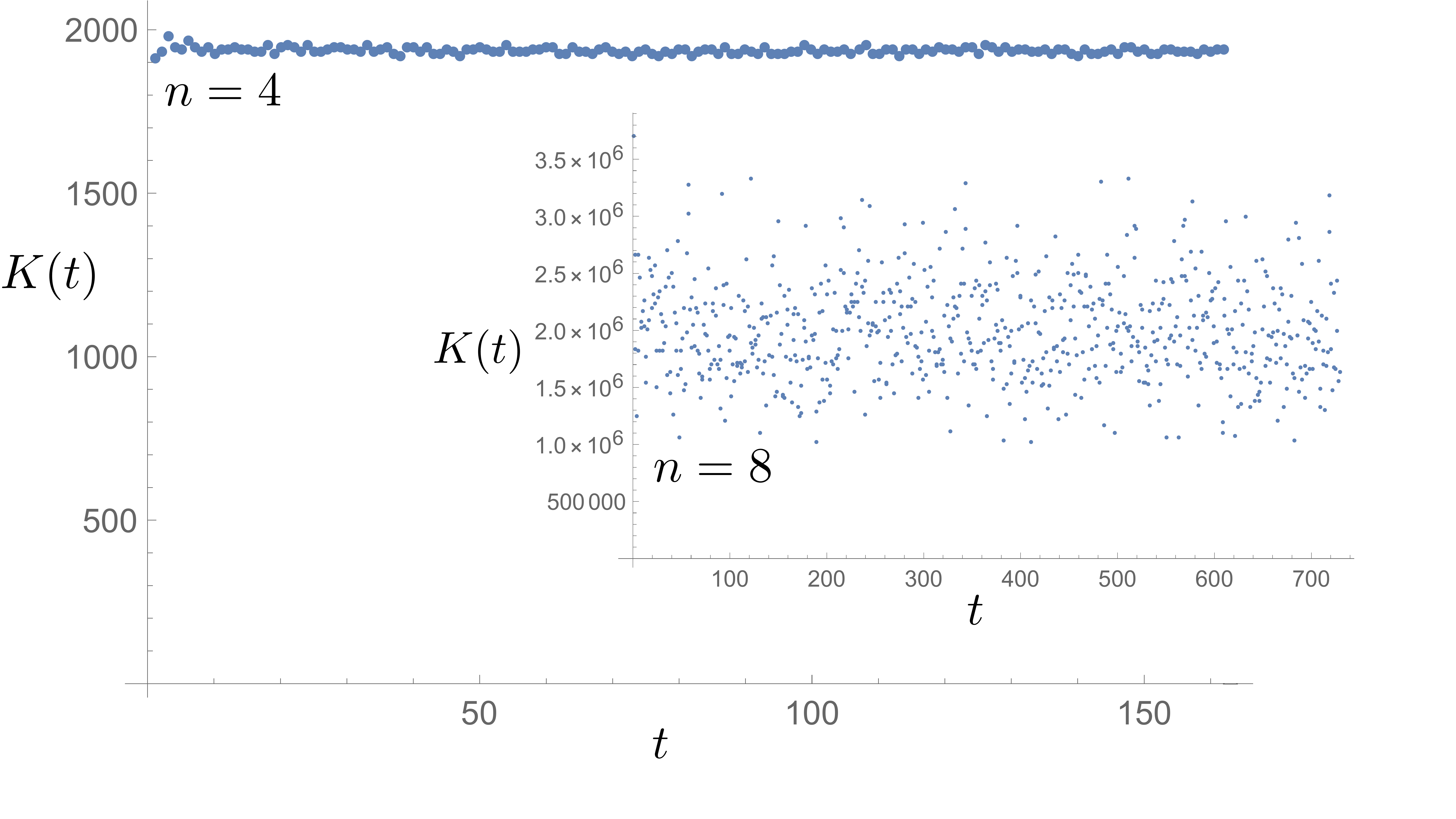}
  \vspace{-5mm}
  \caption{\textbf{Spectral form factor of the evolution operator} $T$ of four layers for the chains with $n=4$ and $n=8$. The two plots are very similar, but they look different because the $n=4$ shows all the vertical scale while the $n=8$ only shows a narrow range. Both plots are essentially flat with no ``dip", a characteristic of ``integrable" systems with Poisson level statistics.}
  \label{fig:SFF T}
\end{figure}

\section{Lorentz transformations}
\label{sec:lorentz transf}

In this section we define the operators $R_l$ and $L_l$ mentioned in Section~\ref{sub LT}, which jointly implement a contraction and a Lorentz boost towards the right and left respectively. 
The treatment of Lorentz transformations is simpler in the infinite chain, so in this section we assume $n=\infty$ and $x\in \mathbb Z$.
Let us now mention a subtle issue. In the infinite chain, global operators like $T, S, R_l, L_l: \H_\infty \to \H_\infty$ are not well-defined. Fortunately, they have a well-defined adjoint action on any operator supported on a finite region; for example, $T$ maps $\a\in \A_0$ onto
\begin{align}
  T\a T^\dagger 
  %= \v_{1,2}\, \v_{3,4}\, \u_{2,3}\, \a_2\, (\v_{1,2}\, \v_{3,4}\, \u_{2,3})^\dagger
  %\u_{2,3}^\dagger \v_{1,2}^\dagger \v_{3,4}^\dagger 
  =
  \begin{tikzpicture}[baseline=-2.5]
    \portaabar{0}{-1}\portabar{0}{1}
    \cables{-1}{3.2}\cables{1}{3.2}
    \draw (-.5,-.5)--(-.5,-.3);
    \draw (-.5,.5)--(-.5,.3);
    \draw (-.5,0) node {$\a$};
    \draw (.5,-.5)--(.5,.5);
    %\draw (.5,1.5)--(.7,1.3)--(.7,-1.3)--(.5,-1.5);
    \portacbar{-1}{2}\portacbar{1}{2}
    \portacabar{-1}{-2}\portacabar{1}{-2}
    \cables{-1}{-2}\cables{1}{-2}
    \draw (-1.5,1.5)--(-1.5,-1.5);
    \draw (1.5,1.5)--(1.5,-1.5);
  \end{tikzpicture}\ .
\end{align}
Therefore, in this section, any global operator $\H_\infty \to \H_\infty$ is understood as an adjoint action on the quasi-local algebra. 
 
\subsection{Local Lorentz contractions}

The operators $R_l, L_l$ are constructed with the building blocks $J_x, K_x$, which we call local Lorentz contractions.
For any even $x$ we define the local Lorentz right-boost contraction isometry as
\begin{align}
  J_x = q^{-\frac 1 2} \cdots
\begin{tikzpicture}[baseline=10]
  \draw[thin,color=black!10] (-1.5,-.7) -- (-1.5,3.7);
  \draw[thin,color=black!10] (-.5,-.7) -- (-.5,3.7);
  \draw[thin,color=black!10] (.5,-.7) -- (.5,3.7);
  \draw[thin,color=black!10] (1.5,-.7) -- (1.5,3.7);
  \draw[thin,color=black!10] (2.5,-.7) -- (2.5,3.7);
  \draw[thin,color=black!10] (3.5,-.7) -- (3.5,3.7);
  \draw[thin,color=black!10] (4.5,-.7) -- (4.5,3.7);
  \draw[thin,color=black!10] (5.5,-.7) -- (5.5,3.7);
  \draw[thin,color=black!10] (6.5,-.7) -- (6.5,3.7);
  \draw[thin,color=black!10] (7.5,-.7) -- (7.5,3.7);
  \portabar{0}{0}\cables{0}{0}
  \barra{-1}{0}\cables{-1}{0}
  \portacbar{1}{1}
  \portabar{2}{0}\cables{2}{0}
  \portacbar{3}{1}
  \portabar{4}{0}\cables{4}{0}
  \portacbar{5}{1}
  \portabar{6}{0} \cables{6}{0}
  \portacbar{7}{1} \cables{7}{0}
  \antibarra{0}{2}\antibarra{-1}{3}\cables{-1}{4.2}
  \antibarra{1}{2}\antibarra{0}{3}
  \antibarra{2}{2}\antibarra{1}{3}\cables{1}{4.2}
  \antibarra{3}{2}\antibarra{2}{3}
  \antibarra{4}{2}\antibarra{3}{3}\cables{3}{4.2}
  \antibarra{5}{2}\antibarra{4}{3}
  \antibarra{6}{2}\antibarra{5}{3}\cables{5}{4.2}
  \antibarra{7}{2}\antibarra{6}{3}\cables{7}{4.2}
  \antibarra{7}{3}
  \draw[thin] (-2.5,-.7) -- (-2.5,3.7);
  \draw[thin] (-3.5,-.7) -- (-3.5,3.7);
  %\draw[thin] (-4.5,-.7) -- (-4.5,3.7);
  \draw (-.4,-1) node {$_x$};
  \draw (-.4,4) node {$_x$};
\end{tikzpicture}  
  \cdots
\end{align}
It is a contraction because, as shown in \eqref{action R}, it has a non-trivial kernel. Also, as can be seen in the picture, in some sense there are two more bottom legs than top legs (although both numbers are infinite).

If we act with $J_x$ on the time-translation operator $T$ defined in \eqref{eq:2 layer} we obtain
\begin{align}
   \label{eq:JxT}
  J_x T &= q^{-\frac 1 2} \cdots
\begin{tikzpicture}[baseline=10]
  %\draw (-.4,4) node {$_x$};
  %\draw[thin,color=black!10] (-1.5,-.7) -- (-1.5,3.7);
  %\draw[thin,color=black!10] (-.5,-.7) -- (-.5,3.7);
  %\draw[thin,color=black!10] (.5,-.7) -- (.5,3.7);
  %\draw[thin,color=black!10] (1.5,-.7) -- (1.5,3.7);
  %\draw[thin,color=black!10] (2.5,-.7) -- (2.5,3.7);
  %\draw[thin,color=black!10] (3.5,-.7) -- (3.5,3.7);
  %\draw[thin,color=black!10] (4.5,-.7) -- (4.5,3.7);
  %\draw[thin,color=black!10] (5.5,-.7) -- (5.5,3.7);
  \portabar{0}{0}\cables{0}{-2}
  \barra{-1}{0}\cables{-1}{-2}
  %\draw (1.58-2,2-1.43) node[scale=.6] {$q^{-1/2}$};
  \portacbar{1}{1}
  \portabar{2}{0}\cables{2}{-2}
  \portacbar{3}{1}
  \portabar{4}{0}\cables{4}{-2}
  \portacbar{5}{1}
  \portabar{6}{0} 
  \cables{6}{-2}
  %\portacbar{7}{1} 
  %\cables{7}{-2}
  \antibarra{0}{2}\antibarra{-1}{3}\cables{-1}{4.2}
  \antibarra{1}{2}\antibarra{0}{3}
  \antibarra{2}{2}\antibarra{1}{3}\cables{1}{4.2}
  \antibarra{3}{2}\antibarra{2}{3}
  \antibarra{4}{2}\antibarra{3}{3}\cables{3}{4.2}
  \antibarra{5}{2}\antibarra{4}{3}
  \antibarra{6}{2}\antibarra{5}{3}\cables{5}{4.2}
  %\antibarra{7}{2}
  \antibarra{6}{3}\cables{6}{4.2}
  %\antibarra{7}{3}
  \draw[thin] (-2.5,-.5) -- (-2.5,3.7);
  \draw[thin] (-3.5,-.5) -- (-3.5,3.7);
  \cables{-3}{-2}\portacbar{-3}{-1}
  \portabar{-2}{-2} \portacbar{-1}{-1}
  \portabar{0}{-2} \portacbar{1}{-1}
  \portabar{2}{-2} \portacbar{3}{-1}
  \portabar{4}{-2} \portacbar{5}{-1}
  \portabar{6}{-2} 
\end{tikzpicture}  
  \cdots
  \\ \nonumber
  &= q^{-\frac 1 2} \cdots
\begin{tikzpicture}[baseline=10]
  %\draw (.58,-.43) node[scale=.6] {$q^{-1/2}$};
  \cables{0}{-2}
  \barra{0}{-1}\cables{-1}{-2}
  \portacbar{1}{1}
  \portabar{2}{0}\cables{2}{-2}
  \portacbar{3}{1}
  \portabar{4}{0}\cables{4}{-2}
  \portacbar{5}{1}
  \portabar{6}{0} 
  \cables{6}{-2}
  %\portacbar{7}{1} \cables{7}{-2} 
  \barra{-1}{-1} \barra{0}{0}
  \antibarra{0}{2}\antibarra{-1}{3}\cables{-1}{4.2}
  \antibarra{1}{2}\antibarra{0}{3}
  \antibarra{2}{2}\antibarra{1}{3}\cables{1}{4.2}
  \antibarra{3}{2}\antibarra{2}{3}
  \antibarra{4}{2}\antibarra{3}{3}\cables{3}{4.2}
  \antibarra{5}{2}\antibarra{4}{3}
  \antibarra{6}{2}\antibarra{5}{3}\cables{5}{4.2}
  %\antibarra{7}{2}
  \antibarra{6}{3}\cables{6}{4.2}
  %\antibarra{7}{3}
  \draw[thin] (-2.5,-.5) -- (-2.5,3.7);
  \draw[thin] (-3.5,-.5) -- (-3.5,3.7);
  \cables{-3}{-2}\portacbar{-3}{-1}
  \portabar{-2}{-2} %\portacbar{-1}{-1}
  \portabar{0}{-2} \portacbar{1}{-1}
  \portabar{2}{-2} \portacbar{3}{-1}
  \portabar{4}{-2} \portacbar{5}{-1}
  \portabar{6}{-2} 
\end{tikzpicture}   
  \cdots
  \\ \nonumber 
  &=q^{-\frac 1 2} \cdots
\begin{tikzpicture}[baseline=10]
  %\draw (1.58,-1.43) node[scale=.6] {$q^{-1/2}$};
  \cables{0}{-2} \cables{-1}{-2}
  \portacbar{1}{1}
  \portabar{2}{0}\cables{2}{-2}
  \portacbar{3}{1}
  \portabar{4}{0}\cables{4}{-2}
  \portacbar{5}{1}
  \portabar{6}{0} 
  \cables{6}{-2}
  %\portacbar{7}{1} \cables{7}{-2} 
  \barra{-1}{-1} \barra{0}{0} \barra{1}{-2}
  \antibarra{0}{2}\antibarra{-1}{3}\cables{-1}{4.2}
  \antibarra{1}{2}\antibarra{0}{3}
  \antibarra{2}{2}\antibarra{1}{3}\cables{1}{4.2}
  \antibarra{3}{2}\antibarra{2}{3}
  \antibarra{4}{2}\antibarra{3}{3}\cables{3}{4.2}
  \antibarra{5}{2}\antibarra{4}{3}
  \antibarra{6}{2}\antibarra{5}{3}\cables{5}{4.2}
  \antibarra{6}{3}\cables{6}{4.2}
  \draw[thin] (-2.5,-.5) -- (-2.5,3.7);
  \draw[thin] (-3.5,-.5) -- (-3.5,3.7);
  \cables{-3}{-2}\portacbar{-3}{-1}
  \portabar{-2}{-2} 
  \barra{0}{-2} \barra{1}{-1}
  \portabar{2}{-2} \portacbar{3}{-1}
  \portabar{4}{-2} \portacbar{5}{-1}
  \portabar{6}{-2} 
\end{tikzpicture}   
  \cdots 
  \\ \nonumber
  &= q^{-\frac 1 2} \cdots
\begin{tikzpicture}[baseline=22]
  %\draw (1.58,2-1.43) node[scale=.6] {$q^{-1/2}$};
  \barra{1}{0}
  \portabar{2}{0}\cables{1}{0}
  \portacbar{3}{1}
  \portabar{4}{0}\cables{3}{0}
  \portacbar{5}{1}
  \portabar{6}{0}\cables{5}{0}
  \cables{6}{0}
  \cables{-3}{6.2}
  %\antibarra{0}{2}%\antibarra{-1}{3}
  \cables{-1}{6.2}
  %\antibarra{1}{2}\antibarra{0}{3}
  \antibarra{2}{2}\antibarra{1}{3}\cables{1}{6.2}
  \antibarra{3}{2}\antibarra{2}{3}
  \antibarra{4}{2}\antibarra{3}{3}\cables{3}{6.2}
  \antibarra{5}{2}\antibarra{4}{3}
  \antibarra{6}{2}\antibarra{5}{3}\cables{5}{6.2}
  \antibarra{6}{3}\cables{6}{6.2}
  \draw[thin] (-2.5,-.7) -- (-2.5,3.5);
  \draw[thin] (-3.5,-.7) -- (-3.5,3.5);
  \draw[thin] (-1.5,-.7) -- (-1.5,3.5);
  \draw[thin] (-.5,-.7) -- (-.5,3.5);
  \portacbar{-3}{5}
  \portabar{-2}{4}\portacbar{-1}{5}
  \portabar{0}{4} \portacbar{1}{5}
  \portabar{2}{4} \portacbar{3}{5}
  \portabar{4}{4} \portacbar{5}{5}
  \portabar{6}{4} 
  \draw (-.4,-1) node {$_x$};
\end{tikzpicture}  
  \cdots = T J_{x+2}
\end{align}
which can be summarised by the following equality
\begin{align}\label{comm:LST}
  T J_x T^\dagger &= J_{x-2}\ ,
  \\ \label{comm:LST 2}
  S J_x S^\dagger &= J_{x+2}\ .
\end{align}
The second equality is just the application of the space-translation operator $S$.
Using these identities and the dual-unitary constraints, we can calculate the action of $J_x$ on a (trace-less) operator $\a_y\in \A_y$ at an arbitrary location $y$,
\begin{align}\label{action R}
  J_x \a_y J_x^\dagger = \left\{
  \begin{array}{ll}
    \a_y & \mbox{ if } y\leq x-2
    \\
    T\a_{y-2} T^\dagger & \mbox{ if } y\geq x+2
    \\
    0 & \mbox{ if } y=x\pm1
    \\
    \eta_+ (\a)_{(x-1, x)} & \mbox{ if } y=x
  \end{array}\right. ,
\end{align}
where we define the completely-positive map
\begin{align}
  \label{def:eta} 
  \eta_+(\a)
  &=\v_{x-1} \Omega_+(\a)_{x-1} \v_{x-1}^\dagger
  \\ \nonumber
  &= q^{-1}
  \v_{x-1}\mathrm{tr}_{x-2} (\u_{x-2} \a_{x-2} \u_{x-2}^\dagger) \v_{x-1}^\dagger 
  \in \A_{x-1, x} \ .
\end{align}

Analogously, for any odd $x$ we define the local Lorentz left-boost contraction as
\begin{align}
  K_x = q^{-\frac 1 2} \cdots
\begin{tikzpicture}[baseline=10]
  \draw[thin,color=black!10] (1.5,-.7) -- (1.5,3.7);
  \draw[thin,color=black!10] (.5,-.7) -- (.5,3.7);
  \draw[thin,color=black!10] (-.5,-.7) -- (-.5,3.7);
  \draw[thin,color=black!10] (-1.5,-.7) -- (-1.5,3.7);
  \draw[thin,color=black!10] (-2.5,-.7) -- (-2.5,3.7);
  \draw[thin,color=black!10] (-3.5,-.7) -- (-3.5,3.7);
  \draw[thin,color=black!10] (-4.5,-.7) -- (-4.5,3.7);
  \draw[thin,color=black!10] (-5.5,-.7) -- (-5.5,3.7);
  \draw[thin,color=black!10] (-6.5,-.7) -- (-6.5,3.7);
  \draw[thin,color=black!10] (-7.5,-.7) -- (-7.5,3.7);
  \portabar{0}{0}\cables{0}{0}
  \antibarra{1}{0}\cables{1}{0}
  \portacbar{-1}{1}
  \portabar{-2}{0}\cables{-2}{0}
  \portacbar{-3}{1}
  \portabar{-4}{0}\cables{-4}{0}
  \portacbar{-5}{1}
  \portabar{-6}{0} \cables{-6}{0}
  \portacbar{-7}{1} \cables{-7}{0}
  \barra{0}{2}\barra{1}{3}\cables{1}{4.2}
  \barra{-1}{2}\barra{0}{3}
  \barra{-2}{2}\barra{-1}{3}\cables{-1}{4.2}
  \barra{-3}{2}\barra{-2}{3}
  \barra{-4}{2}\barra{-3}{3}\cables{-3}{4.2}
  \barra{-5}{2}\barra{-4}{3}
  \barra{-6}{2}\barra{-5}{3}\cables{-5}{4.2}
  \barra{-7}{2}\barra{-6}{3}\cables{-7}{4.2}
  \barra{-7}{3}
  \draw[thin] (2.5,-.7) -- (2.5,3.7);
  \draw[thin] (3.5,-.7) -- (3.5,3.7);
  %\draw[thin] (-4.5,-.7) -- (-4.5,3.7);
  \draw (.5,-1) node {$_x$};
  \draw (.5,4) node {$_x$};
\end{tikzpicture}  
  \cdots
\end{align}
By proceeding as in \eqref{eq:JxT} we obtain
\begin{align}\label{comm:KST}
  T K_x T^\dagger &= K_{x+2}\ ,
  \\ \label{comm:KST 2}
  S K_x S^\dagger &= K_{x+2}\ .
\end{align}
The action of $K_x$ on a (trace-less) operator $\a_y\in \A_y$ at an arbitrary location $y$ is 
\begin{align}
  K_x \a_y K_x^\dagger = \left\{
  \begin{array}{ll}
    T\a_{y+2}T^\dagger & \mbox{ if } y\leq x-2
    \\
    \a_y & \mbox{ if } y\geq x+2
    \\
    0 & \mbox{ if } y=x\pm1
    \\
    \eta_- (\a)_{(x,x+1)} & \mbox{ if } y=x
  \end{array}\right. ,
\end{align}
where we define completely-positive map
\begin{align}
  \nonumber
  \eta_-(\a)
  &=\v_{x} \Omega_-(\a)_{x+1} \v_{x}^\dagger
  \\ &= q^{-1}
  \v_{x}\mathrm{tr}_{x+2} (\u_{x+1} \a_{x+2} \u_{x+1}^\dagger) \v_{x}^\dagger 
  \in \A_{x, x+1} \ .
\end{align}

\subsection{Global Lorentz contractions}

For any even integer $l>0$ we define the the (global) Lorentz right-boost contraction isometry as
\begin{align}\label{def:B_l}
  R_l &= \cdots J_{5l}\, J_{3l}\, J_{l}\, \tilde J_{-l}\, \tilde J_{-3l}\, \tilde J_{-5 l} \cdots  
%  \\ &= 
%  \cdots L_{3l}^\tr\, L_{2l}^\tr\, L_{l}^\tr 
%  \big(S\, T^\dagger L_{-l}^\tr\big)
%  \big(S\, T^\dagger L_{-2l}^\tr\big)
%  \big(S\, T^\dagger L_{-3l}^\tr\big) \cdots  
\end{align}
where we also define $\tilde J_x = S\, T^{-1} J_{x}$.
The reason for using $\tilde J_x$ instead of $J_x$ when $x<0$ is that $\tilde J_x$ is corrected with a spacetime translation so that $B_l$ acts trivially around the origin $x\in [-(l-2),l-2] \subseteq \mathbb Z$. 
Recall that pure Lorentz transformations fix the origin $(x,t)=(0,0)$.
A calculation similar to that in \eqref{eq:JxT} yields
\begin{align}
  \tilde J_x = q^{-\frac 1 2} \cdots
\begin{tikzpicture}[baseline=10]
  \draw[thin,color=black!10] (1.5,-.7) -- (1.5,3.7);
  \draw[thin,color=black!10] (.5,-.7) -- (.5,3.7);
  \draw[thin,color=black!10] (-.5,-.7) -- (-.5,3.7);
  \draw[thin,color=black!10] (-1.5,-.7) -- (-1.5,3.7);
  \draw[thin,color=black!10] (-2.5,-.7) -- (-2.5,3.7);
  \draw[thin,color=black!10] (-3.5,-.7) -- (-3.5,3.7);
  \draw[thin,color=black!10] (-4.5,-.7) -- (-4.5,3.7);
  \draw[thin,color=black!10] (-5.5,-.7) -- (-5.5,3.7);
  \draw[thin,color=black!10] (-6.5,-.7) -- (-6.5,3.7);
  \draw[thin,color=black!10] (-7.5,-.7) -- (-7.5,3.7);
  \portacabar{0}{0}\cables{0}{0}
  \antibarra{1}{0}\cables{1}{0}
  \portaabar{-1}{1}
  \portacabar{-2}{0}\cables{-2}{0}
  \portaabar{-3}{1}
  \portacabar{-4}{0}\cables{-4}{0}
  \portaabar{-5}{1}
  \portacabar{-6}{0} \cables{-6}{0}
  \portaabar{-7}{1} \cables{-7}{0}
  \barra{0}{2}\barra{1}{3}\cables{1}{4.2}
  \barra{-1}{2}\barra{0}{3}
  \barra{-2}{2}\barra{-1}{3}\cables{-1}{4.2}
  \barra{-3}{2}\barra{-2}{3}
  \barra{-4}{2}\barra{-3}{3}\cables{-3}{4.2}
  \barra{-5}{2}\barra{-4}{3}
  \barra{-6}{2}\barra{-5}{3}\cables{-5}{4.2}
  \barra{-7}{2}\barra{-6}{3}\cables{-7}{4.2}
  \barra{-7}{3}
  \draw[thin] (2.5,-.7) -- (2.5,3.7);
  \draw[thin] (3.5,-.7) -- (3.5,3.7);
  %\draw[thin] (-4.5,-.7) -- (-4.5,3.7);
  \draw (.5,-1) node {$_x$};
  \draw (.5,4) node {$_x$};
\end{tikzpicture}  
  \cdots
\end{align}

Analogously, for any odd integer $l>0$, we define the the (global) Lorentz left-boost contraction isometry as
\begin{align}\label{def:L_l}
  L_l &= 
  \cdots K_{-5l}\, K_{-3l}\, K_{-l}\, \tilde K_{l}\, \tilde K_{3l}\, \tilde K_{5 l} \cdots
\end{align}
where we also define $\tilde K_x = S\, T^{-1} K_{x}$ so that the origin is fixed.
The following lemma specifies the action of $R_l$ and $L_l$ on any operator of the form $\a(x,t) = T^t \a_x T^{-t}$.

\begin{Lemma}\label{lemma:lorentz}
If $l$ is an even positive integer, $\a\in \A$ a local operator and $x$ a location such that $|x-lm|>1$ for all odd integers $m$, then
\begin{align}\label{eq:R trans f}
  R_l \a(x,t) R_l^\dagger
  =
  \a \big[ x-2f_l (x-2t),t+f_l (x-2t) \big]\ ,
\end{align}
where we define the function $f_l(x)= \left\lfloor\frac{x-l}{2l} \right\rfloor +1$, plotted in Figure \ref{fig:f(x)}.
If $l$ is an odd positive integer, $\a\in \A$ a local operator and $x$ a location such that $|x-lm|>1$ for all odd integers $m$, then
\begin{align}\label{eq:R trans 2}
  L_l \a(x,t) L_l^\dagger
  =
  \a \big[ x-2f_l (x+2t),t-f_l (x+2t) \big]\ .
\end{align}
\end{Lemma}

\begin{figure}
  \centering
  \includegraphics[width=8cm]{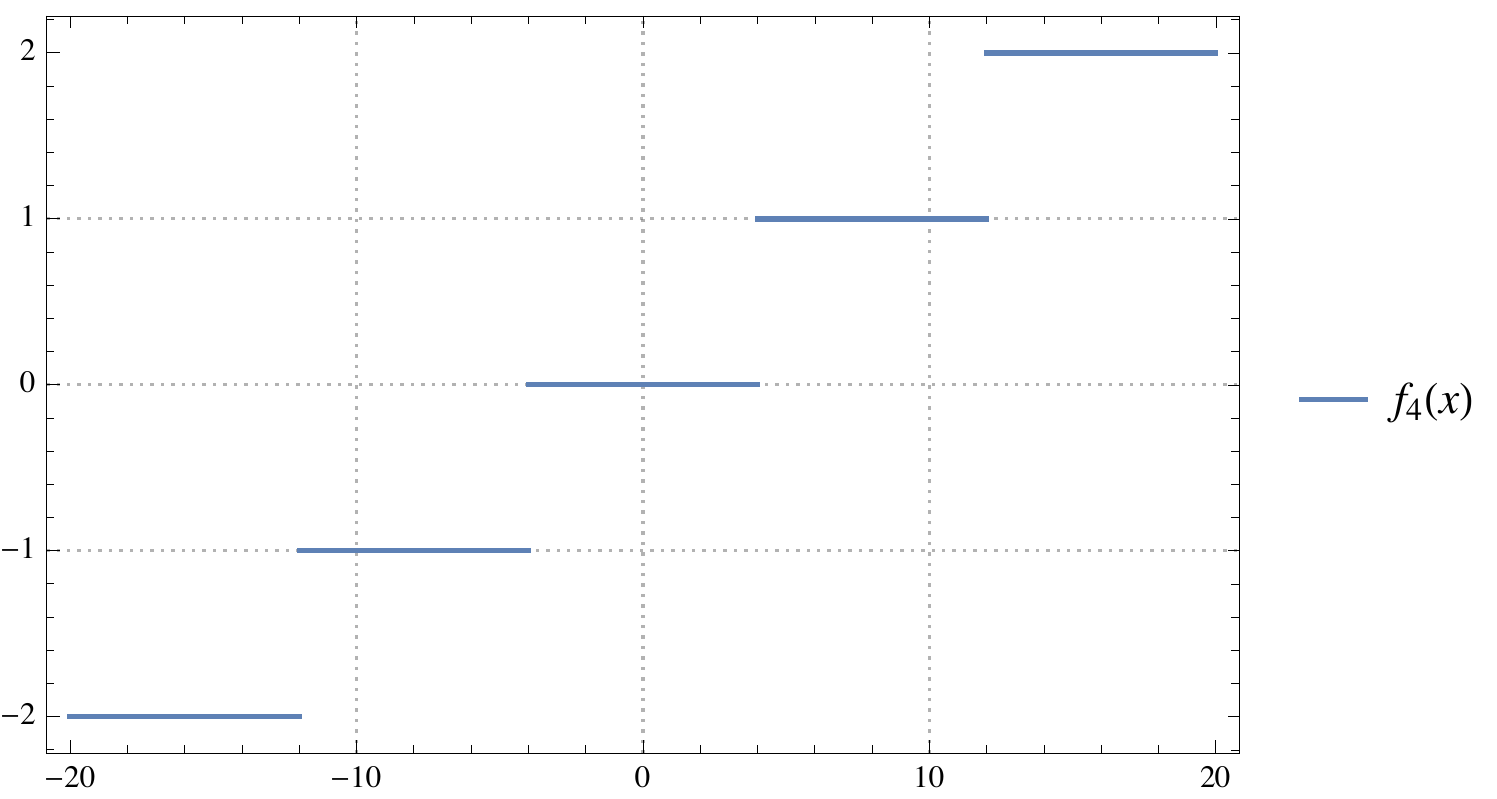}    
  \caption{\textbf{Plot of the function} $f_l (x)$ for $l=4$.\hspace{15mm}}
  \label{fig:f(x)}
\end{figure}

\noindent
%Most values of $x$ satisfy the premise of the lemma, and the rest are analised below. 
We can write the coordinate transformations of the lemma as
\begin{align}%\label{eq:R trans f}
  \left.\begin{array}{l}
    x' = x-2f_l (x \mp 2t) 
    \\
    t' = t \pm f_l (x \mp 2t)
  \end{array}\right\}\ ,
\end{align}
where the upper and lower signs correspond to the right and left boosts respectively.
In the limit $|x|\gg l$ we can use the approximation 
\begin{align}
  f_l (x) \approx \frac x {2l}\ ,
\end{align}
to write the coordinates transformation as
\begin{align}\label{B-transf}
  \left.\begin{array}{l}
    x' \approx \left(1-\frac 1 l \right) x \pm \frac 2 l t
    \\
    t' \approx \left(1-\frac 1 l \right) t \pm \frac 1 {2l} x    
  \end{array}\right\}\ .
\end{align}
Note that, independently of the sign $\pm$, this transformation preserves Minkowski's metric up to a scale factor
\begin{align}
  (2t')^2-x'^2 = \left(1-\frac 2 l \right) \left[ (2t)^2-x^2 \right]\ .
\end{align}
(Recall that, in this model, the speed of light is $c=2$.)
We can also obtain the velocity of the Lorentz boost in \eqref{B-transf} by first, undoing the contraction by dividing $(x',t')$ by the scale factor $\sqrt{1-1/l}$, and second, comparing the resulting transformation to a standard Lorentz boost. This results in the velocity
\begin{align}
  v = \frac {\pm 2} {\sqrt{4-2 l +l^2}}\ .
\end{align}

In Lemma~\ref{lemma:lorentz}, the premise $|x-lm|>1$ warrants that $R_l$ and $L_l$ perform a pure spacetime transformation, leaving the local operator $\a(x,t)$ unaltered (apart from evolving it in time). On the other hand, when $x=lm\pm 1$ for some odd $m$, we have that $R_l \a(x,t) R_l^{\dagger} =0$ as a consequence of \eqref{action R}. And when $x=lm$ for some odd $m$, in addition to a spacetime transformation $(x,t) \mapsto (x',t')$, the local operator $\a$ is processed by the complete positive maps $\eta_\pm$ and $\tilde\eta_\pm$ defined in \eqref{def:eta} and \eqref{def:tildeeta}.
However, in the limit $l\gg 1$, most locations $(x,t)$ satisfy the premise.

\begin{proof}[Proof of Lemma~\ref{lemma:lorentz}]
In order to prove \eqref{eq:R trans f} we need to write the action of $\tilde J_x = S\, T^{-1} J_{x}$ on a (trace-less) operator $\a_y\in \A_y$ at an arbitrary location $y$, 
\begin{align}\label{eq:tildR}
  \tilde J_x \a_{y} \tilde J_x^\dagger = 
  \left\{
  \begin{array}{ll}
    T^\dagger \a_{y+2} T & \mbox{ if } y\leq x-2
    \\
    \a_y & \mbox{ if } y\geq x+2    
    \\
    0 & \mbox{ if } y=x\pm 1
    \\
    \tilde\eta_+ (\a_x) & \mbox{ if } y=x
  \end{array}  \right.
\end{align}
where
\begin{align}
  \label{def:tildeeta}
  \tilde \eta_+(\a)
  &=\u_{x-2}^\dagger \Omega_+(\a)_{x-1} \u_{x-2}
  \\  \nonumber
  &= q^{-1}
  \v_{x-1}\mathrm{tr}_{x-2} (\u_{x-2} \a_{x-2} \u_{x-2}^\dagger) \v_{x-1}^\dagger 
  \in \A_{x-2, x-1} \ .
\end{align}
The above has been obtained by applying the translation $S\, T^{-1}$ to the action \eqref{action R}.
Importantly, the premise of the lemma ($|x-lm|>1$ for all odd numbers $m$) implies that only the first two cases in \eqref{action R} and \eqref{eq:tildR} are relevant, which simplifies this proof.

In what follows we analyse the action of $B_l$ on a local operator $\a_x$ for the three cases where $f_l (x)$ is equal, larger or smaller than zero, separately.
If $f_l (x) =0$ then $\a_x$ commutes with all operators $J_{lm}$ and $\tilde J_{-lm}$ where $m$ is a positive odd integer, therefore \eqref{def:B_l} implies $R_l \a_x R_l^\dagger =\a_x$. 
If $f_l(x) \geq 1$ then $\a_x$ commutes with all operators $\tilde J_{-lm}$ and $J_{lm}$, except for the $J_{lm}$ with $m\in \{1, 3, \dots, 2f_l(x)-1 \}$. Note that $f_l (x)$ counts the number of non-commuting factors in $B_l$.
This implies
\begin{align}
  \nonumber
  R_l \a_x R_l^\dagger 
  &=
%  (J_{f_l(x)l} \cdots J_{l}) [\a_x]
  (J_{[2f_l(x)-1]l} \cdots J_{l}) \a_x 
  (J_{[2f_l(x)-1]l} \cdots J_{l})^\dagger
  \\ \label{RaR+} &= 
  \a[x-2 f_l(x), f_l(x)] \ ,
\end{align}
where the last equality follows from \eqref{action R}.
If $f_l (x) \leq -1$ then the operators which do not commute with $\a_x$ are $\tilde J_{ml}$ with $m\in \{2f_l(x)+2, \ldots, -3, -1\}$. Note that $|f_l (x)|$ counts the number of non-commuting factors in $B_l$.
Simlarly, we obtain
\begin{align}
  \nonumber
  R_l \a_x R_l^\dagger 
  &=
  (\tilde J_{[2f_l(x)+1]l} \cdots \tilde J_{-l}) \a_x
  (\tilde J_{[2f_l(x)+1]l} \cdots \tilde J_{-l})^\dagger 
  \\ \label{RaR-} &= 
  \a[x-2f_l(x),f_l(x)]\ ,
\end{align}
where the last equality follows from \eqref{eq:tildR}.

Finally, we invoke (\ref{comm:LST}-\ref{comm:LST 2}) to obtain the algebraic identity
\begin{align}
  R_l (ST) = (ST) R_l\ ,
\end{align}
which allows to generalise \eqref{RaR+} and \eqref{RaR-} to the case $\a(x,t)$ with $t\neq 0$,
\begin{align}
  \nonumber
  & \hspace{-5mm} R_l \a(x,t) R_l^\dagger
  \\ \nonumber&= 
  (B_l S^t T^t) \a(x-2t,0) (B_l S^t T^t)^\dagger
  \\ \nonumber &=
  (S^t T^t B_l) \a_{x-2t} (B_l S^t T^t)^\dagger
  \\ \nonumber &=
  (S^t T^t) \a[x-2t-2f_l(x-2t), f_l (x-2t)] (S^t T^t)^\dagger
  \\ &= 
  \a[x-2f_l (x-2t),t+f_l (x-2t)]\ .
\end{align}  
This concludes the proof of \eqref{eq:R trans f}.

In order to prove \eqref{eq:R trans 2} we proceed analogously, but this time, we use the action of $\tilde K_x = S^{-1} T^{-1} K_x$ on an arbitrary local operator
\begin{align}\label{eq:tildK}
  \tilde K_x \a_{y} \tilde K_x^\dagger = 
  \left\{
  \begin{array}{ll}
    \a_y & \mbox{ if } y\leq x-2
    \\
    T^\dagger \a_{y-2} T & \mbox{ if } y\geq x+2    
    \\
    0 & \mbox{ if } y=x\pm 1
    \\
    \tilde\eta_- (\a_x) & \mbox{ if } y=x
  \end{array}  \right.
\end{align}
where
\begin{align}
  \tilde \eta_-(\a)
  &=\u_{x+1}^\dagger \Omega_- (\a)_{x+1} \u_{x+1}
  \\   \nonumber
  &= q^{-1}
  \v_{x-1}\mathrm{tr}_{x-2} (\u_{x-2} \a_{x-2} \u_{x-2}^\dagger) \v_{x-1}^\dagger 
  \in \A_{x+1, x+2} \ .
\end{align}
Also, we use the fact that (\ref{comm:KST}-\ref{comm:KST 2}) imply
\begin{align}
  L_l (S^{-1}T) = (S^{-1}T) L_l\ ,
\end{align}
to complete the proof.
\end{proof}

%If instead of a contraction we want to implement a dilation together with a Lorentz transformation, we need to use the operators $L_l^\dagger$ and $R_l^\dagger$ instead of $R_l$ and $L_l$. Note that here the roles are reversed, $L_l^\dagger$ generates a boost with positive velocity while $R_l^\dagger$ does so with negative velocity.

\section{Outlook}

In this work we have introduced conformal QCAs, which are discrete-spacetime versions of CFTs, and studied their properties as models of holography. We have obtained several results, but we have mostly opened venues for future research. In what follows we enumerate some of the problems that will be addressed in future work.
\begin{itemize}

  \item How much of the CFT phenomenology is covered by conformal QCAs? Do they have a continuum limit?
  
  \item Characterise the algebraic structure of scale and Lorentz transformations in conformal QCAs. In this work, the operators implementing these transformations ($C,D,R_l, L_l$) have been constructed so that classical geometries (i.e.~tensor-network states) are mapped to classical geometries, and not superpositions thereof. The reason for this is that it simplifies the visualisation of the dynamics of geometry and matter. However, if we relax this property, other constructions exist which generate a more structured algebra. 
  
  \item What is the mechanism that produces Poisson level statistics on the spectrum of a conformal QCA constructed with a random dual unitary (which has Wigner-Dyson level statistics)?

  \item Our framework allows for calculating the time evolution of arbitrary (discrete) geometries. However, it is still not clear whether this dynamics corresponds to a discrete version of Einstein's field equations. If this is the case then conformal QCAs will provide a new perspective on quantum gravity.

\end{itemize}

%Conformal QCAs are the first holography models which are fully discrete, in the sense that time is not continuous. This higher symmetry between space and time is reflected on the exact invariance that the dynamics has with respect to Lorentz transformations. Other models only display Lorentz invariance of some observables in the scaling limit. Another advantage of our approach is the absence of a duality map, since any tensor-network state of the QCA has a direct geometric interpretation which evolves in a natural way by the dynamics of the QCA.

%Perhaps, some of the ideas in this work can be generalised to the continuum, in the context of CFT. In particular, the construction of tensor-network states with the same building blocks than the dynamics could have a functional generalisation. Given the Hamiltonian density of a $d$-dimensional CFT and a metric distance tensor of a $d$-dimensional spatial, it might be possible to construct a state of the CFT by transforming the continuum limit of a tensor network to a path integral.

%Also, it is worth mentioning that the circuit structure of the dynamics facilitates simulations on quantum computers, since no Trotterisation is required. This feature could simplify the implementation of quantum-gravity experiments in the lab.

\section{Acknowledgements}

I am thankful to Diego Blas, Sougato Bose, Tom Holden-Dye, Arijeet Pal and Andrea Russo for valuable discussion. This work has been supported by the UK’s Engineering and Physical Sciences Research Council (grant number EP/R012393/1).

\bibliography{bib_adscft}

\begin{thebibliography}{10}

\bibitem{Maldacena_1999}
Juan Maldacena.
\newblock The large $n$ limit of superconformal field theories and
  supergravity.
\newblock {\em International Journal of Theoretical Physics}, 38(4):1113--1133,
  1999.

\bibitem{Gubser_1998}
S.S. Gubser, I.R. Klebanov, and A.M. Polyakov.
\newblock Gauge theory correlators from non-critical string theory.
\newblock {\em Physics Letters B}, 428(1-2):105--114, may 1998.

\bibitem{Witten_98}
Edward Witten.
\newblock Anti de sitter space and holography.
\newblock {\em Adv. Theor. Math. Phys.}, 2, 1998.

\bibitem{Ryu_2006}
Shinsei Ryu and Tadashi Takayanagi.
\newblock Aspects of holographic entanglement entropy.
\newblock {\em Journal of High Energy Physics}, 2006(08):045--045, aug 2006.

\bibitem{Ryu_2007}
Shinsei Ryu and Tadashi Takayanagi.
\newblock Holographic derivation of entanglement entropy from the ads/cft
  correspondence.
\newblock {\em Physical Review Letters}, 96(18), may 2006.

\bibitem{Hubeny}
Veronika~E. Hubeny, Mukund Rangamani, and Tadashi Takayanagi.
\newblock A covariant holographic entanglement entropy proposal.
\newblock {\em Journal of High Energy Physics}, 2007(07):062, jul 2007.

\bibitem{Lashkari_2014}
Nima Lashkari, Michael~B. McDermott, and Mark~Van Raamsdonk.
\newblock Gravitational dynamics from entanglement
  {\textquotedblleft}thermodynamics{\textquotedblright}.
\newblock {\em Journal of High Energy Physics}, 2014(4), apr 2014.

\bibitem{Czech_2012}
Bart{\l}omiej Czech, Joanna~L Karczmarek, Fernando Nogueira, and Mark~Van
  Raamsdonk.
\newblock The gravity dual of a density matrix.
\newblock {\em Classical and Quantum Gravity}, 29(15):155009, jul 2012.

\bibitem{Wall_2014}
Aron~C Wall.
\newblock Maximin surfaces, and the strong subadditivity of the covariant
  holographic entanglement entropy.
\newblock {\em Classical and Quantum Gravity}, 31(22):225007, nov 2014.

\bibitem{Headrick_2014}
Matthew Headrick, Veronika~E. Hubeny, Albion Lawrence, and Mukund Rangamani.
\newblock Causality and holographic entanglement entropy.
\newblock {\em Journal of High Energy Physics}, 2014(12), dec 2014.

\bibitem{Esp_ndola_2018}
Ricardo Esp{\'{\i}}ndola, Alberto Güijosa, and Juan~F. Pedraza.
\newblock Entanglement wedge reconstruction and entanglement of purification.
\newblock {\em The European Physical Journal C}, 78(8), aug 2018.

\bibitem{Aaronson_2022}
Scott Aaronson and Jason Pollack.
\newblock Discrete bulk reconstruction.
\newblock {\em arXiv:2210.15601}, 2022.

\bibitem{Bao_2015}
Ning Bao, Sepehr Nezami, Hirosi Ooguri, Bogdan Stoica, James Sully, and Michael
  Walter.
\newblock The holographic entropy cone.
\newblock {\em Journal of High Energy Physics}, 2015(9), sep 2015.

\bibitem{Swingle_12}
Brian Swingle.
\newblock Constructing holographic spacetimes using entanglement
  renormalization, 2012.

\bibitem{Pastawski_2015}
Fernando Pastawski, Beni Yoshida, Daniel Harlow, and John Preskill.
\newblock Holographic quantum error-correcting codes: toy models for the
  bulk/boundary correspondence.
\newblock {\em Journal of High Energy Physics}, 2015(6), jun 2015.

\bibitem{Hayden_2016}
Patrick Hayden, Sepehr Nezami, Xiao-Liang Qi, Nathaniel Thomas, Michael Walter,
  and Zhao Yang.
\newblock Holographic duality from random tensor networks.
\newblock {\em Journal of High Energy Physics}, 2016(11), nov 2016.

\bibitem{Basteiro_2022}
Pablo Basteiro, Giuseppe~Di Giulio, Johanna Erdmenger, Jonathan Karl, Rene
  Meyer, and Zhuo-Yu Xian.
\newblock Towards explicit discrete holography: Aperiodic spin chains from
  hyperbolic tilings.
\newblock {\em {SciPost} Physics}, 13(5), nov 2022.

\bibitem{Erdmenger_2022}
Johanna Erdmenger, Mario Flory, Marius Gerbershagen, Michal~P. Heller, and
  Anna-Lena Weigel.
\newblock Exact gravity duals for simple quantum circuits.
\newblock {\em {SciPost} Physics}, 13(3), sep 2022.

\bibitem{Niermann_2022}
Laura Niermann and Tobias~J. Osborne.
\newblock Holographic networks for (1+1)-dimensional de sitter space-time.
\newblock {\em Physical Review D}, 105(12), jun 2022.

\bibitem{Osborne_2020}
Tobias~J. Osborne and Deniz~E. Stiegemann.
\newblock Dynamics for holographic codes.
\newblock {\em Journal of High Energy Physics}, 2020(4), apr 2020.

\bibitem{Jahn_2019}
A.~Jahn, M.~Gluza, F.~Pastawski, and J.~Eisert.
\newblock Majorana dimers and holographic quantum error-correcting codes.
\newblock {\em Physical Review Research}, 1(3), nov 2019.

\bibitem{Asaduzzaman_2020}
Muhammad Asaduzzaman, Simon Catterall, Jay Hubisz, Roice Nelson, and Judah
  Unmuth-Yockey.
\newblock Holography on tessellations of hyperbolic space.
\newblock {\em Physical Review D}, 102(3), aug 2020.

\bibitem{Jahn_2021}
A.~Jahn, M.~Gluza, C.~Verhoeven, S.~Singh, and J.~Eisert.
\newblock Boundary theories of critical matchgate tensor networks.
\newblock {\em Journal of High Energy Physics}, 2022(4), apr 2022.

\bibitem{Jahn_2020}
Alexander Jahn, Zolt{\'{a} }n Zimbor{\'{a}}s, and Jens Eisert.
\newblock Central charges of aperiodic holographic tensor-network models.
\newblock {\em Physical Review A}, 102(4), oct 2020.

\bibitem{Brower_2021}
Richard~C. Brower, Cameron~V. Cogburn, A.~Liam Fitzpatrick, Dean Howarth, and
  Chung-I Tan.
\newblock Lattice setup for quantum field theory in ads$_2$.
\newblock {\em Physical Review D}, 103(9), may 2021.

\bibitem{Jahn_2022}
Alexander Jahn, Zolt{\'{a}}n Zimbor{\'{a}}s, and Jens Eisert.
\newblock Tensor network models of {AdS}/{qCFT}.
\newblock {\em Quantum}, 6:643, feb 2022.

\bibitem{Asaduzzaman_2022}
Muhammad Asaduzzaman, Simon Catterall, Jay Hubisz, Roice Nelson, and Judah
  Unmuth-Yockey.
\newblock Holography for ising spins on the hyperbolic plane.
\newblock {\em Physical Review D}, 106(5), sep 2022.

\bibitem{Arrigui}
Pablo Arrighi, Vincent Nesme, and Reinhard Werner.
\newblock Unitarity plus causality implies localizability.
\newblock {\em arXiv:0711.3975}, 2007.

\bibitem{G_tschow_2010}
Johannes Gütschow, Sonja Uphoff, Reinhard~F. Werner, and Zolt{\'{a} }n
  Zimbor{\'{a}}s.
\newblock Time asymptotics and entanglement generation of clifford quantum
  cellular automata.
\newblock {\em Journal of Mathematical Physics}, 51(1):015203, jan 2010.

\bibitem{Freedman_2022}
Michael Freedman, Jeongwan Haah, and Matthew~B. Hastings.
\newblock The group structure of quantum cellular automata.
\newblock {\em Communications in Mathematical Physics}, 389(3):1277--1302, jan
  2022.

\bibitem{Farrelly_2020}
Terry Farrelly.
\newblock A review of quantum cellular automata.
\newblock {\em Quantum}, 4:368, nov 2020.

\bibitem{Bertini_2019}
Bruno Bertini, Pavel Kos, and Tomaz Prosen.
\newblock Exact correlation functions for dual-unitary lattice models in 1+1
  dimensions.
\newblock {\em Physical Review Letters}, 123(21), nov 2019.

\bibitem{Bertini_2021}
Bruno Bertini, Pavel Kos, and Tomaz Prosen.
\newblock Random matrix spectral form factor of dual-unitary quantum circuits.
\newblock {\em Communications in Mathematical Physics}, 387(1):597--620, jul
  2021.

\bibitem{Piroli_2020}
Lorenzo Piroli, Bruno Bertini, J.~Ignacio Cirac, and Tomaz Prosen.
\newblock Exact dynamics in dual-unitary quantum circuits.
\newblock {\em Physical Review B}, 101(9), mar 2020.

\bibitem{kos_20}
Bruno Bertini, Pavel Kos, and Tomaz Prosen.
\newblock {Operator Entanglement in Local Quantum Circuits II: Solitons in
  Chains of Qubits}.
\newblock {\em SciPost Phys.}, 8:068, 2020.

\bibitem{Barcelo:2000ta}
C.~Barcelo and Matt Visser.
\newblock {Brane surgery: Energy conditions, traversable wormholes, and voids}.
\newblock {\em Nucl. Phys. B}, 584:415--435, 2000.

\bibitem{Barvinsky_2006}
A.~O. Barvinsky and D.~V. Nesterov.
\newblock Quantum effective action in spacetimes with branes and boundaries.
\newblock {\em Physical Review D}, 73(6), mar 2006.

\bibitem{Farshi_22}
Tom Farshi, Jonas Richter, Daniele Toniolo, Arijeet Pal, and Lluis Masanes.
\newblock Absence of localization in two-dimensional clifford circuits.
\newblock {\em arXiv:2210.10129}, 2022.

\end{thebibliography}

\begin{figure*}
  \centering  
  \includegraphics[width=175mm]{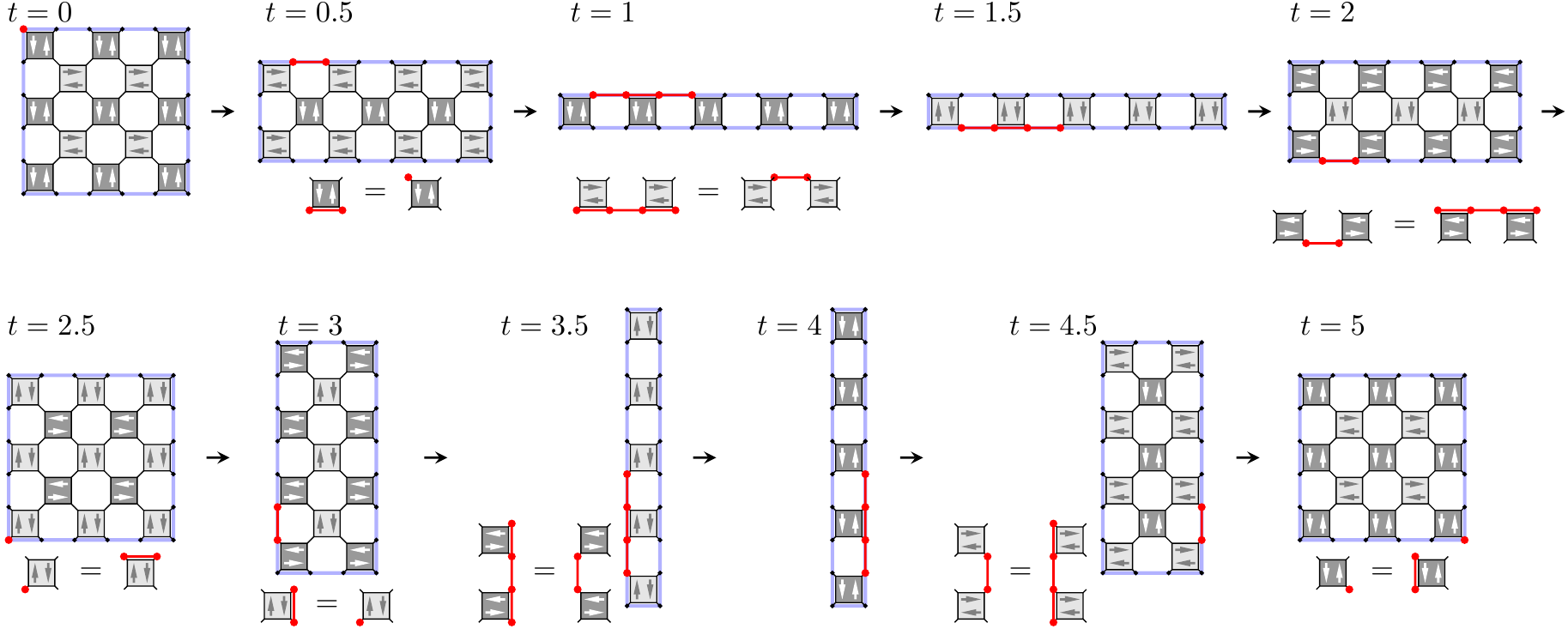}
  \caption{\textbf{Dynamics of flat space with one particle at the corner.} Same assumptions (the dual unitary and $\a$ are real) and notation than in Figure~\ref{flat_m_cycle}. Operator inserted at $t=.5$ is the transpose of that at $t=2$, and the same relation holds for the pairs of times $(0,2.5)$, $(1,1.5)$, $(3,4.5)$, $(3.5,4)$ and $(2.5,5)$. The particle reaches the antipodal point at $t=5$, and the period is $\Delta t= \frac n 2 =10$.}
  \label{flat_m_corner} 
\end{figure*} 

\begin{figure*}
  \centering  
  \includegraphics[width=15cm]{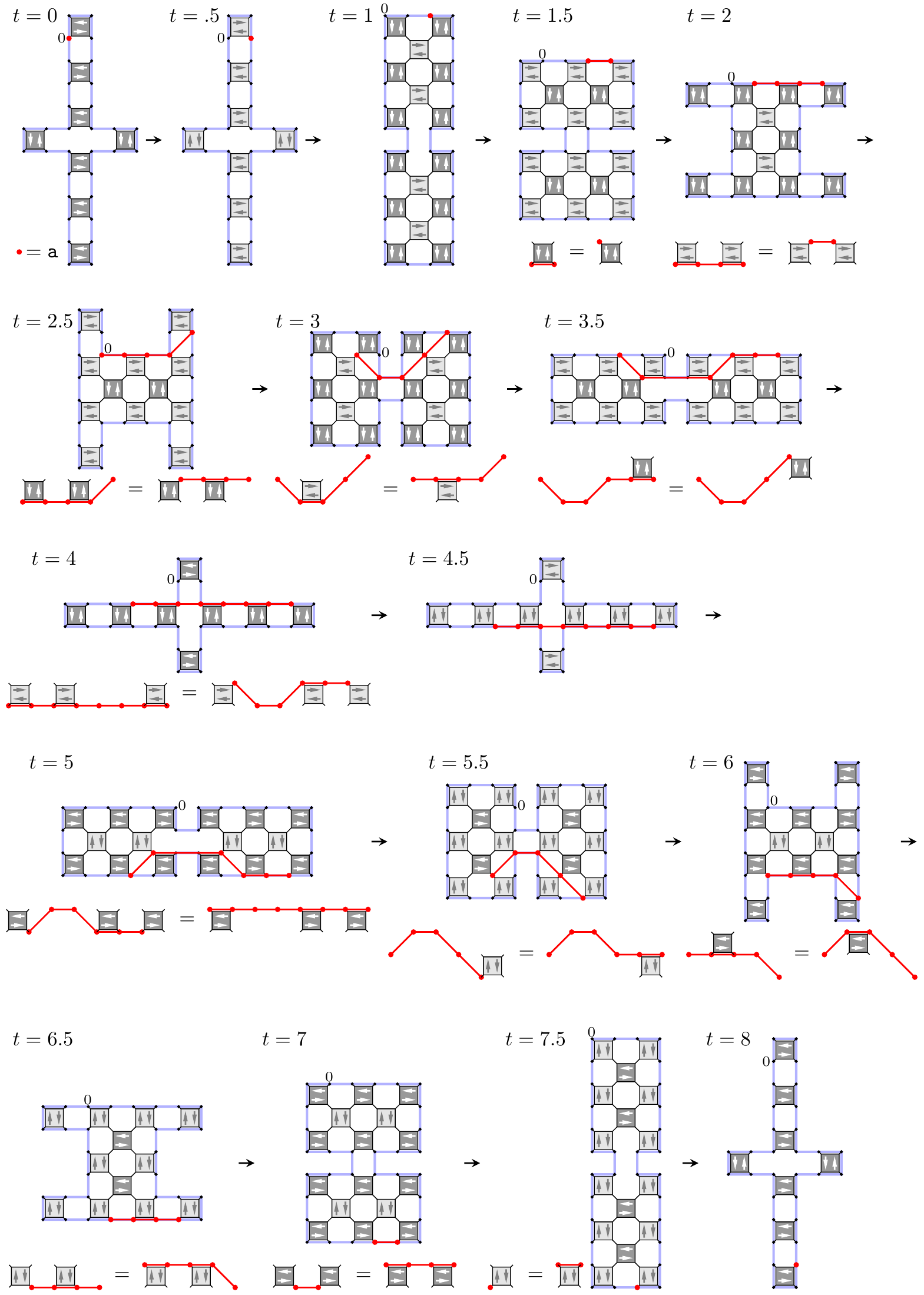}
  \caption{\textbf{Dynamics of toy AdS with one particle.} Same assumptions (the dual unitary and $\a$ are real) and notation than Figure~\ref{flat_m_cycle}. Operator inserted at $t=1.5$ is the transpose of that at $t=7$, and the same relation holds for the pairs of times $(2,6.5)$, $(2.5,6)$, $(3,5.5)$, $(3.5,5)$ and $(4,4.5)$. The particle reaches the antipodal point at $t= \frac n 4 = 8$, and the period is $\Delta t= \frac n 2 =16$.}
  \label{ads_m_cycle} 
\end{figure*} 

\end{document}